\documentclass{article}

\usepackage{apptools}
\usepackage{bbold}
\usepackage[usenames,dvipsnames]{xcolor}
\usepackage[margin=1in]{geometry}
\usepackage[utf8]{inputenc}
\usepackage{amsfonts}
\usepackage{amsmath}
\usepackage{amssymb}
\usepackage{amsthm}
\usepackage{graphicx}
\usepackage[colorlinks=true, linkcolor=blue, citecolor=gray]{hyperref}
\usepackage{indentfirst}
\usepackage{mathrsfs}
\usepackage{multirow}
\usepackage{physics}
\usepackage{subfig}
\usepackage{tikz} 

\usetikzlibrary{quantikz}

\AtAppendix{\counterwithin{theorem}{section}}

\newcommand{\thm}[1]{\hyperref[thm:#1]{Theorem~\ref*{thm:#1}}}
\newcommand{\defn}[1]{\hyperref[defn:#1]{Definition~\ref*{defn:#1}}}
\newcommand{\lem}[1]{\hyperref[lem:#1]{Lemma~\ref*{lem:#1}}}
\newcommand{\prop}[1]{\hyperref[prop:#1]{Proposition~\ref*{prop:#1}}}
\newcommand{\conj}[1]{\hyperref[conj:#1]{Conjecture~\ref*{conj:#1}}}
\newcommand{\fig}[1]{\hyperref[fig:#1]{Figure~\ref*{fig:#1}}}
\newcommand{\tab}[1]{\hyperref[tab:#1]{Table~\ref*{tab:#1}}}
\renewcommand{\sec}[1]{\hyperref[sec:#1]{Section~\ref*{sec:#1}}}
\newcommand{\app}[1]{\hyperref[app:#1]{Appendix~\ref*{app:#1}}}
\newcommand{\cor}[1]{\hyperref[cor:#1]{Corollary~\ref*{cor:#1}}}
\newcommand{\obs}[1]{\hyperref[obs:#1]{Observation~\ref*{obs:#1}}}
\newcommand\Tstrut{\rule{0pt}{3ex}}         
\newcommand\Bstrut{\rule[-1.5ex]{0pt}{0pt}}   

\DeclareMathOperator{\LambertW}{LambertW}
\DeclareMathOperator{\openone}{\mathbb{1}}
\DeclareMathOperator{\QFT}{QFT}
\DeclareMathOperator{\rnd}{rnd}
\DeclareMathOperator{\sgn}{sgn}

\newtheorem{theorem}{Theorem}
\newtheorem{conjecture}{Conjecture}

\newtheorem{lemma}[theorem]{Lemma}

\theoremstyle{definition}
\newtheorem{definition}[theorem]
{Definition}

\title{Time Dependent Hamiltonian Simulation \\ Using Discrete Clock Constructions}
\author{Jacob Watkins$^1$, Nathan Wiebe$^{2,3,4}$, Alessandro Roggero$^{5,6,7}$, Dean Lee$^1$\smallskip\\
$^1$ {\small \it Facility for Rare Isotope Beams and Department of Physics and Astronomy,}\\ {\small\it Michigan State University, MI 48824, USA}\\
$^2$ {\small\it Department of Computer Science, University of Toronto, Toronto, ON M5S 2E4, Canada}\\
$^3$ {\small\it Pacific Northwest National Laboratory, Richland, WA 99354, USA}\\
$^4$ {\small\it Department of Physics, University of Washington, Seattle, WA 98195, USA}\\
$^5$ {\small\it InQubator for Quantum Simulation (IQuS), Department of Physics, }\\{\it\small University of Washington, Seattle, WA 98195, USA}\\
$^6$ {\small\it Dipartimento di Fisica, University of Trento, I–38123, Povo, Trento, Italy}\\
$^7$ {\small\it INFN-TIFPA Trento Institute of Fundamental Physics and Applications,  Trento, Italy}\\}
\date{}

\begin{document}
\maketitle

\begin{abstract}
    Compared with time independent Hamiltonians, the dynamics of generic quantum Hamiltonians $H(t)$ are complicated by the presence of time ordering in the evolution operator. In the context of digital quantum simulation, this difficulty prevents a direct adaptation of time independent simulation algorithms for time dependent simulation. However, there exists a framework within the theory of dynamical systems which eliminates time ordering by adding a ``clock" degree of freedom. In this work, we provide a computational framework, based on this reduction, for encoding time dependent dynamics as time independent systems. As a result, we make two advances in digital Hamiltonian simulation. First, we create a time dependent simulation algorithm based on performing qubitization on the augmented clock system, and in doing so, provide the first qubitization-based approach to time dependent Hamiltonians that goes beyond Trotterization of the ordered exponential.  Second, we define a natural generalization of multiproduct formulas for time-ordered exponentials, then propose and analyze an algorithm based on these formulas. Unlike other algorithms of similar accuracy, the multiproduct approach achieves commutator scaling, meaning that this method outperforms existing methods for physically-local time dependent Hamiltonians. Our work reduces the disparity between time dependent and time independent simulation and indicates a step towards optimal quantum simulation of time dependent Hamiltonians.
\end{abstract}

\newpage
\tableofcontents
\newpage

\section{Introduction}
    Quantum simulation has, within recent years, emerged as the preeminent application of quantum computing~\cite{Lloyd1998Universal,Berry2007Sparse,Childs2012LCU,Childs2021Theory,Low2019Qubitization,Campbell2019Random, Berry2017Sparse}. Specifically, there now exists a large family of digital Hamitonian simulation algorithms, each of which successfully compiles complicated quantum dynamics into elementary sequences of gates. This has raised the possibility for solving problems in the domain sciences spanning chemistry~\cite{Lanyon2010Chemistry,Reiher2017Reaction,Burg2021Catalysis,Lee2021Tensor}, materials science~\cite{Babbush2018Materials}, nuclear and neutrino physics~\cite{Roggero2020Scattering,Roggero2021Oscillations,Baroni2022Nuclear}, field theory~\cite{Jordan2012Quantum,Klco2018Schwinger,Shaw2020Schwinger} and beyond.  Further work has shown the central role that these algorithms play in translating continuous time quantum walk algorithms, as well as alternative models of quantum computing such as adiabatic quantum computing, to the standard circuit model~\cite{Aharonov2003Adiabatic}.
    
    Presently, there are three fundamental strategies for simulating quantum dynamics: qubitization~\cite{Low2019Qubitization,Gilyen2019QSVT,Dong2021QSP}, linear combinations of unitaries (LCU) ~\cite{Childs2012LCU,Berry2015Optimal,Berry2015Taylor,Low2019Multiproduct, Kieferova2019Dyson} and product formulas~\cite{Berry2007Sparse,Wiebe2010Higher,Childs2019Randomization,Childs2021Theory}. Of these methods, product formulas have the unique characteristic that the error depends on the commutators of the Hamiltonian terms; however, the error scaling of these methods is super-polynomially worse than either of the other two strategies. This can be addressed through the use of multiproduct formulas, which hybridize product formula and LCU techniques~\cite{Childs2012LCU,Low2019Multiproduct,Faehrmann2022Randomizing}. However, to date this approach has not been successfully applied to simulate time dependent Hamiltonians. This is because the formalism for analyzing ordered operator exponentials is more complicated than for ordinary operator exponentials, which makes translating MPFs to the time dependent domain highly nontrivial.
    
    Our work reduces the discrepancy between time dependent and time independent Hamiltonian simulation techniques by providing a computational framework, based on the $(t,t')$-formalism~\cite{Hale1991Bifurcations,Casati1989Chaos}, for encoding time dependent Hamiltonian dynamics into a time independent Hamiltonian. Essentially, this is accomplished by promoting the $t$ in $H(t)$ to a degree of freedom and introducing a new evolution parameter. Unlike existing approaches based on a continuous time parameter~\cite{Peskin1993Solution,Suzuki2003Methodology}, our clock spaces are all finite dimensional, providing the advantages of direct simulability as well as formal well-behavedness. After proving the accuracy of clock space evolutions, we propose an algorithm for time dependent Hamiltonian simulation based on direct simulation of the augmented clock system. We focus on qubitization because of its inability to handle time ordering directly, and due to its asymptotically superior performance in the time independent setting. We perform a rigorous error and query complexity analysis starting from a Hamiltonian input as a LCU with coefficients computed by an oracle. While our asymptotic bounds do not match simulation lower bounds, we suspect this may be improved upon with better analysis. Regardless, we provide, to our knowledge, the first non-trivial qubitization algorithm for time dependent simulation.
    
    Additionally, we take advantage of the finite clock space formalism to propose a natural generalization of multiproduct formulas (MPFs) to the time ordered setting, and offer evidence that these are valid extrapolations to the exact time evolution. We then create an algorithm based on these ``time dependent MPFs" and analyze its performance. The algorithm is, to the authors' knowledge, the most efficient existing quantum algorithm for simulating time dependent Hamiltonians which also exhibits ``commutator scaling." By this we mean that the algorithm exhibits zero error for time independent Hamiltonian with easily simulatable and commuting terms. In practice, this constitutes a super-polynomial advantage over the error scaling provided by product formula methods for solving the same problems. We perform two rudimentary numerical demonstrations which indicate the the approach works as anticipated in reducing product formula error.

    The rest of the paper is laid out as follows. In~\sec{Main_Results}, we summarize our main findings and compare with leading algoriths for time dependent Hamiltonian simulation. In~\sec{Basic_Ideas}, we review the notions of time evolution, product formulas, the $(t,t')$-formalism, and MPFs needed for our work, while also introducing some notation regarding the various vector and functional norms. In~\sec{Finite_Clock_Space}, we introduced the notion of finite dimensional clock spaces, and show that they encode simulations of $H(t)$ with an appropriate time independent clock Hamiltonian. We then apply these constructions, first to direct simulation by qubitization in~\sec{TDQubitization}, then as a tool for proof of time dependent MPFs in~\sec{TDMPF}. We conclude in~\sec{Conclusion} with remarks on the broader impact of this work and directions for future research.
    
\section{Main Results} \label{sec:Main_Results}
    \tab{babbush} summarizes the query complexity of our proposed algorithms (green) and displays, for comparison, several leading algorithms for time dependent simulation. To facilitate comparison, the Hamiltonian is taken in an LCU model
    \begin{equation}
        H(t) = \sum_{j=1}^L \alpha_j(t) U_j
    \end{equation}
    for unitary and Hermitian $U_j$, and with each $\alpha_j$ real valued. The Trotter, QDrift, and MPF approaches apply for more general Linear Combination of Hamiltonians (LCH) inputs, but Trotter and MPF also have stronger smoothness assumptions on $\alpha_j$. The various norms are defined in~\sec{Basic_Ideas}, and the scalar function $\Lambda$ characterizes the Hamiltonian ``difficulty" at each time in terms of the size of $H$ and its derivatives (see~\defn{LambdaBound}). The tilde in ``$\widetilde{O}$" indicates the exclusion of subdominant multiplicative logarithmic factors from the complexity.
    \begin{table}[t]
        \centering
        \resizebox{\textwidth}{!}{%
        \renewcommand{\arraystretch}{1.5}
        \begin{tabular}{|c|l|l|l|}
            \hline
            \large{Method} & \large{Query Complexity} & \large{Auxiliary Qubits} & \large{CS?}\\
            \hline
            Trotter~\cite{Wiebe2010Higher} & $O\big(L (\|\Lambda\|_1)^{1+ o(1)}/\epsilon^{o(1)}\big)$ & 0 & Yes \Tstrut\Bstrut\\
            QDrift~\cite{Berry2020L1} & $O(\|\alpha\|_{1,1}^2 /\epsilon)$ & 0 & No \Tstrut\Bstrut\\
            Dyson~\cite{Berry2020L1} & $O\big(L^2 \|\alpha\|_{1,1} \log(1/\epsilon)\big)$ & $\widetilde{O}\big (\log(\|\dot{\alpha}\|_{1,1} /\epsilon)+\log(\|\alpha\|_{1,\infty} T/\epsilon)\big)$ & No \Tstrut\Bstrut\\
            \color{Green}
            Qubitization & $\begin{aligned}&\widetilde{O}\big(\|\alpha\|_{\infty,1}^\mathrm{rev}T + \log1/\epsilon\big)  \\
            &+ \max_t \|\dot{H}\| T^2/\epsilon\big)\end{aligned}$ & $\begin{aligned}O\big(&\log\left(L\|\alpha\|_{\infty,1}^\mathrm{rev} T/\epsilon\right) \\ 
            +&\log\left(L\max_t \|\dot{H}\| T^2/\epsilon\right)\big)\end{aligned}$ & {No}\Tstrut\Bstrut\\
            \color{Green}
            MPF & $\widetilde{O}\big({L\|\Lambda\|_1 \log^2(1/\epsilon)}\big)$ & $\widetilde{O}\left(\log\left(L\|\Lambda\|_1 \|\dot{\alpha}\|_{\infty,\infty} T^2/\epsilon \right) \right)$ & Yes \Tstrut\Bstrut\\
            \hline
        \end{tabular}}
        \caption{Summary of our results (green) and comparison to leading quantum simulation methods for time dependent Hamiltonians. We assume that $H=\sum_{j=1}^L \alpha_j(t) U_j$ for Hermitian unitaries $U_j$ and real-valued $\alpha_j(t)$. $\Lambda$ is a positive, time dependent function with dimensions of $H$, and quantifies the size of $\alpha_j$ and its derivatives (see~\defn{LambdaBound}). $\|\alpha\|_{p,q}$ refers to a nested vector-$p$ and functional-$q$ norm for the coefficients $\alpha =(\alpha_j)_{j=1}^L$, and $\|\alpha\|_{p,q}^\mathrm{rev}$ indicates these are taken in the reverse order.  Commutator scaling (CS) here means the simulation error vanishes in the limit where $H$ is time independent and $[U_j, U_k] = 0$ for all $j, k \in [L]$.} 
        \label{tab:babbush}
    \end{table}

    The qubitization algorithm is obtained from simulating the finite clock space construction of~\sec{Finite_Clock_Space}. We carry out a procedure for taking the LCU input $H(t)$ and producing a time independent LCU input on the larger space. Unfortunately, our analysis shows a Trotter-like error term in the query complexity, as can be seen in the Table. Thus, our analysis does not show that our qubitization algorithm improves over Trotter simulation. In~\sec{Qubitization_Discussion} we discuss how this error term might be eliminated with improved analysis and modification of the clock construction. If this term were absent, the query complexity of our approach would match the $T$ and $\epsilon$ lower bounds set by the no fast-forwarding theorem~\cite{Berry2007Sparse,Atia2017Fast}.  

    The time dependent Multiproduct Formulas (MPFs) we use for simulation do not require the clock space. However, there is a lack of formal proof that the extrapolants work. We clearly state the needed fact in~\conj{TDMPF_existence}, and outline a possible path to proof using the clock space formalism. The generalization from standard MPFs is rather intuitive, and the numerics of~\sec{NumericalDemos} support the truth of the conjecture. Assuming this, we go on to analyze the error and query complexity of an algorithm based on these formulas, using adaptive time steps to handle harder parts of the simulation with a greater share of resources. In the query model, the algorithm performs overall comparably to the Dyson method, as can be seen in the Table. However, the MPF method does have some notable strengths. For one, it exhibits commutator scaling, meaning it is errorless in the time independent, commuting limit. Additionally, the dependence on the number of terms $L$ is quadratically improved, though not as good as QDrift. 
    
\section{Basic Ideas and Notation} \label{sec:Basic_Ideas}
    Here we review fundmental concepts behind our work and meanwhile introduce notation. The reader is encouraged to read as desired, or reference this section and continue to new results starting in~\sec{Finite_Clock_Space}.

    \subsection{Nested Norms}
        In our characterization of simulation errors and computational complexity, we will make frequent use of vector $p$ or functional $q$ norm, and often both in conjunction. For a vector $v \in \mathbb{C}^n$ the vector $p$-norm is given by
        \begin{equation}
            \|v\|_p := \left(\sum_{i=1}^n \abs{v_i}^p\right)^{1/p}.
        \end{equation}
        For a function $f:[a,b]\rightarrow \mathbb{C}$ the functional $p$ norm is analogously defined as
        \begin{equation}
            \|f\|_p := \left(\int_a^b \abs{f(t)}^p dt\right)^{1/p}
        \end{equation}
        provided the integral exists. These expressions hold for $p$ positive integers. We also use $p = \infty$ in the standard way.
        \begin{equation}
            \|v\|_\infty = \max_j \abs{v_j},\qquad \|f\|_\infty = \sup_{t\in[a,b]} \abs{f(t)}
        \end{equation}
        For a vector-valued function $v:[a,b]\rightarrow \mathbb{C}^n$, the norm $\|v\|_{p,q}$ is just the nesting of these two starting with the vector norm. For integer $p,q$
        \begin{equation}
            \|v\|_{p,q} := \left(\int_a^b \Big(\sum_{j=1}^n\abs{v_j(t)}^p\Big)^{q/p} dt\right)^{1/q}
        \end{equation}
        and similar expression hold for $p$ or $q$ being $\infty$. The norm $\|v\|_{p,q}^\mathrm{rev}$ indicates that the functional $p$ norm is taken, followed by the vector $q$ norm. For example,
        \begin{equation}
            \|v\|_{\infty,1}^{\mathrm{rev}} = \sum_j \sup_t \abs{v_j(t)}.
        \end{equation}
    
\subsection{Digital Hamiltonian Simulation} 
    According to the postulates of quantum mechanics, a closed system with Hilbert space $\mathcal{H}$ has dynamics which are generated by some self-adjoint operator $H$ on the space, called the \emph{Hamiltonian}. In certain cases of physical interest, the parameters of the physical system may change over time, or one may work in a ``non-inertial frame'' such as an interaction picture. In such instances, a proper description requires that $H$ be a function of time. In saying $H(t)$ generates the dynamics, we mean there exists a unique unitary-operator-valued function $U$, termed the \emph{time evolution operator}, which solves the following initial value problem.
    \begin{align} \label{eq:Schrödinger_eqtn}
    \begin{aligned}
        i\partial_{t} U(t, t_0) &= H(t) U(t, t_0) \\
        U(t_0, t_0) &= \openone
    \end{aligned}
    \end{align}
    (Here, and throughout, we choose units where Planck's constant $\hbar$ is one.) The initial value problem~\eqref{eq:Schrödinger_eqtn} is the \emph{Schrödinger equation} for the time evolution operator. Typically, we set $t_0 = 0$ and use $T$ to denote the final time of interest. The solution $U$ encodes maximal knowledge about the system dynamics. For any initial state $\ket{\psi(t_0)} = \ket{\psi_0}$ (only pure states need be considered here), one obtains the time-evolved quantum state $\ket{\psi(t)} = U(t,t_0)\ket{\psi_0}$ simply by applying $U$. Hence, $U$ ``propagates" our state in time, and is sometimes called the \emph{propagator}. 
    
    In the case where $H(t)$ is a constant function, we say it is \emph{time independent}. Such behavior naturally arises in systems whose dynamical laws exhibit time-translation invariance. In this case, the solution $U$ to equation \eqref{eq:Schrödinger_eqtn} takes the expression
    \begin{equation}
        U(T, 0) = e^{-i H T}.
    \end{equation}
    We say that $U$ is an \emph{ordinary} (operator) exponential of $H$. For more arbitrary $H(t)$, in contrast, the solution $U$ is typically written as an \emph{ordered} (operator) exponential of $H(t)$.
    \begin{equation} \label{eq:gen_solution}
        U(T, 0) = \mathcal{T} \exp{-i\int_{0}^{T} H(\tau) d\tau}
    \end{equation}
    There are several different ways to understand the meaning of~\eqref{eq:gen_solution}, but for our purposes, the most insightful is through through product integration~\cite{Dollard1984Product}. Given a family of partitions $\{t_j\}_{j=1}^n$ of the interval $[s,t]$, with maximum width $\delta_n$ tending to zero as $n\rightarrow \infty$, a solution is given by the product integral
    \begin{equation}
         \mathcal{T} \exp{-i\int_{0}^{T} H(\tau) d\tau} = \lim_{n\rightarrow\infty} \prod_{j=1}^{n-1} e^{-i H(t) \Delta t_j}
    \end{equation}
    where $\Delta t_j = t_{j+1} - t_j$. One feature of this approach is that, for sufficiently large but finite $n$, we achieve an approximation that is amenable to simulation by time independent methods. This is exactly the approach taken for product formula simulation. However, simulation of this product by LCU or qubitization is generally pointless, since the errors of discretization outweigh any accuracy benefits achieved of using these protocols over Trotter.
    
    In the setting of digital quantum computation, one can only hope to calculate the propagator to some approximation, which can be made better at increasing cost. Constructing the approximate circuit for $U$ defines the problem of Hamiltonian simulation.
    \begin{definition}[Heuristic]
        Let $d$ be some distance measure on the set of quantum channels on $n$ qubits, and define the \emph{Hamiltonian simulation problem} as follows. Given an interval $[0, T]$, $\epsilon > 0$ and a Hermitian valued function $H: [0,T] \rightarrow \mathrm{Herm}(\mathbb{C}^{2^n})$, construct a quantum circuit $V$ such that
        \begin{equation*}
            d(U, V) \leq \epsilon
        \end{equation*}
        where $U \equiv U(T, 0)$ is the exact propagator generated by $H(t)$. Any circuit-valued function $V(\epsilon, T, H)$ which solves the problem above for some subset of the domain of parameters is known as a  \emph{Hamiltonian simulation algorithm} (abbreviated ``simulation algorithm'') over that domain.
    \end{definition}
    Additional technical requirements to are needed for the definition above to be precise, such as an input model for $H$ and that the circuit is efficiently compilable. In this work, we will take $d(U, V) = \norm{U-V}$, where $\norm{\cdot}$ will denote the spectral norm (also called the induced 2-norm).
    \begin{equation} \label{defn:spectral_norm}
        \|A\| := \max_{v \neq 0} \frac{\|Av\|}{\|v\|}
    \end{equation}
    The norm appearing on on the right is the Euclidean norm. The spectral error $\|U - V\|$ should be thought of as a worst-case simulation error for any initial state. For various reasons, such as partial measurements with post-processing, $V$ may not be unitary, but the actual underlying channel will still a valid quantum operation. The parameter $\epsilon$ is called the \emph{error tolerance}, and $d(U, V)$ is the \emph{simulation error}. We will $\epsilon$ the error, accuracy, or precision of the simulation interchangeably.
    
    Any decent approximation to $U(t,t_0)$ should become arbitrarily accurate as $t \to t_0$ (and approach the identity). The quality of the approximation, in the context of product and multiproduct formulas, is typically quoted in terms of a power law convergence. This is captured by the following definition.
    \begin{definition}
        For finite-dimensional $\mathcal{H}$, let $L:[0,T]^2\rightarrow L(\mathcal{H})$. We say that $L_p:[0,T]^2\rightarrow L(\mathcal{H})$ is a $p$th-order approximation to $L$ if, for all $t \in [0,T)$, 
        \begin{equation*}
            \|L(t+\tau, t) - L_p(t+\tau,t)\| \in O(\tau^{p+1})
        \end{equation*}
        where $\tau$ is taken asymptotically to $0$.
    \end{definition}
    As an important example, product formulas approximate operator exponentials of sums, by splitting exponential to match terms in a power series of error operators. The simplest example is $1\mathrm{st}$-order Trotter, with
    \begin{equation}
        e^{\lambda(A+B)} = e^{\lambda A}e^{\lambda B} + O(\lambda^2).
    \end{equation}
    Such splittings are not exact because $A$ and $B$ don't generally commute. 
    
    A linear function $L:[0,T]^2\rightarrow L(\mathcal{H})$ is said to be \emph{symmetric} if it possesses the following ``time reversal symmetry": $L(t_1,t_0) = L^{\dagger}(t_0,t_1)$. Symmetric operators are closed under addition and scalar multiplication by a real number. They are also closed under multiplication of the form
    \begin{equation}
        L^{(2)}(t, t_0) := L(t, t') L(t', t_0)
    \end{equation}
    for any $t' \in [0,T]$. This is demonstrated by the following calculation.
    \begin{align}
    \begin{aligned}
        V^{(2)}(t, t_0)^\dagger &= V(t', t_0)^\dagger V(t, t')^\dagger \\
        &= V(t_0, t') V(t', t) \\
        &\equiv V^{(2)}(t_0, t).
    \end{aligned}
    \end{align}
    This symmetry is valuable, because approximation schemes for $U$ involving symmetric operations are ensured to have an odd error series~\cite{Berry2007Sparse}. Consequently, any symmetric formula is of order $2n$ for integer $n > 0 $.

\subsection{Continuous Clock Space} \label{sec:ContinuousClockSpace}
    Mathematically, and more broadly than the Hamiltonian setting considered here, the distinction between time dependent and time independent systems can be cast as a distinction between autonomous and nonautonomous dynamical systems. Dynamical systems are differential equations in a single evolution parameter $t$, which can always be expressed as a first-order differential equation
    \begin{equation} \label{eq:nonautonomous_system}
        \dot{x} = f(x,t), \qquad x(0) = x_0
    \end{equation}
    possibly by standard reduction-of-order techniques. Here $x \in \mathbb{R}^n$ consists of $n$ evolution parameters which implicitly depend on time $t$. It has long been recognized that a simple transformation allows for the reduction of nonautonomous systems to autonomous ones~\cite{Hale1991Bifurcations}. The trick is to promote $t$ to a coordinate, thereby making $f(x,t)$ satisfy the requirement of only depending on coordinates. Letting $\tau$ take the place of the evolution parameter (time), we still want $t$ and $\tau$ to be essentially the same. This is supplied by the simple equation 
    \begin{align}
        \frac{dt}{d\tau} = 1,\qquad t(0) = 0.
    \end{align}
    With this, we have the following autonomous system
    \begin{align}
        (\dot{x}, \dot{t}) = (f(x,t), 1),\qquad (x(0), t(0)) = (x_0, 0)
    \end{align}
    whose solution encodes the solution to the original system~\eqref{eq:nonautonomous_system}.

    When the dynamics of interest are generated by a classical Hamiltonian $H(q,p,t)$, autonomous systems are those in which $H$ has no time dependence. Playing the same trick as before and promoting $t$ to a coordinate, we want $dt/d\tau = 1$ as above. By Hamilton's equations, we require a Hamiltonian $H_F$ such that
    \begin{align}
        \frac{\partial H_F}{\partial (-E)} = 1
    \end{align}
    where $-E$ is the conjugate momentum to $t$ (the minus sign chosen so that $E$ represents energy; we anticipate our results). But we also want $H_F$ to reproduce the same dynamics for $q,p$ as $H$. This altogether implies the choice
    \begin{align}
        H_F(q,p;t,E) := H(q,p;t) - E
    \end{align}
    which is known as the Floquet Hamiltonian in condensed matter physics. Simply put, the term $-E$ pulls the $t$ coordinate at constant velocity to the right, changing $H(q,p,t)$ just as needed to enact the appropriate effect on the other coordinates and momenta.

    For \emph{quantum} Hamiltonians, we can simply employ a standard quantization procedure on $H_F$.
    \begin{equation}
        [\hat{t},-\hat{E}] = i I.
    \end{equation}
    A natural choice of Hilbert space for $t$ is $L^2([0,T])$: square integrable functions on the interval of simulation. We then get a representation of $t$ and $E$ as $t$ multiplication and $t$-derivatives, respectively.
    \begin{equation}
        (\hat{t}\psi)(t) = t\psi(t),\qquad (\hat{E}\psi)(t) = i\partial_t \psi 
    \end{equation}
    The Floquet Hamiltonian becomes
    \begin{equation} \label{eq:cont-clock-Ham}
        H_F = H - i\partial_t
    \end{equation}
    which looks remarkably similar to a rearranged Schrödinger operator, but as a caution, $t$ is no longer the evolution parameter. As representing a physical system $H_F$ is certainly odd and infeasible for a number of reasons, including unboundedness from below and that the $t$-system is unaffected by the state of the main system. As a manufactured system, however, it can be useful both for simulation and formal purposes, as we've alluded to.
    
    The framework presented, sometimes called the $(t,t')$-formalism because of the two distinct ``times", finds use in  periodically-driven quantum systems~\cite{Casati1989Chaos}. But for our purposes, the elimination of explicit dependence of on the evolution parameter in $H$ is most exciting, because it implies the time evolution operator requires no time ordering, which still encoding the full time dynamics~\cite{Peskin1993Solution}.

    Having taken a developmental approach, let's provide a more concrete characterization of the continuous clock space. The Hilbert space is given by
    \begin{equation}
        \mathcal{H} = \mathcal{H}_s \otimes \mathcal{H}_c.
    \end{equation}
    where $\mathcal{H}_c \cong L^2(\mathcal{M})$, and $\mathcal{M}$ is the connected one-dimensional smooth manifold representing $t$ which contains the interval $[0,T]$. On $\mathcal{H}_c$, $E$ acts as a generator of translations, but is an unbounded operator. Nevertheless, the exponentials of $E$ above are well defined through the spectral theorem and functional calculus for unbounded operators~\cite{Hall2013ch10}. States $\psi \in \mathcal{H}$ can be expressed as a certain class of integrable functions on $\mathcal{M}$ whose values $\psi(t)$ are states on $\mathcal{H}_s$. The inner product on $\mathcal{H}$ is the natural one
    \begin{equation}
        \langle \phi \vert \psi\rangle := \int_\mathcal{M} \braket{\phi(t)}{\psi(t)}_s dt
    \end{equation}
    where $\langle \cdot\vert\cdot\rangle_s$ denotes the inner product on $\mathcal{H}_s$. 
    
    If $\mathcal{M}$ is not exactly $[0,T]$, then a time dependent observable $A(t)$ acting on $\mathcal{H}_s$ will need to be defined on the entire clock space. Once done, $A(t)$ is promoted to a parameter-independent observable $\mathbf{A}$ on $\mathcal{H}$, acting on $\psi \in \mathcal{H}$ in a manner corresponding with the original space.
    \begin{equation}
        (\mathbf{A}\psi)(t) := A(t)\psi(t)
    \end{equation}
    We observe that $\mathbf{A}$ is \emph{local} in $\mathcal{H}_c$ in the sense of acting via multiplication in $t$-space. Let $\textbf{H}$ be the promoted Hamiltonian operator, and let $\mathbf{U}(\tau)$ be the unitary operator given by 
    \begin{equation} \label{eq:clockU}
        \mathbf{U}(\tau) = e^{iE \tau}e^{-i(\mathbf{H}-E)\tau}.
    \end{equation}
    One can verify that $\mathbf{U}$ solves the following Schrödinger equation,
    \begin{align}
    \begin{aligned}
        i \partial_\tau \mathbf{U}(\tau) &=  \mathbf{H}(s) \mathbf{U}(\tau) \\
        \mathbf{U}(0) &= \openone.
    \end{aligned}
    \end{align}
    Here,
    \begin{equation}
        \mathbf{H}(s) \equiv e^{iEs}\mathbf{H} e^{-iE s}
    \end{equation}
    is a $\tau$-dependent Hamiltonian corresponding to simple, uniform translation along the clock space. For any state $\Psi_0 \in \mathcal{H}$, the function 
    \begin{equation}
        \Psi(s) := \mathbf{U}(\tau) \Psi_0
    \end{equation}
    solves the Schrödinger equation generated by $\mathbf{H}(\tau)$, but more importantly, it encodes solutions to the dynamics under $H(t)$. Indeed, for any $t \in \mathcal{M}$, we have a state $\psi(\tau, t) \in \mathcal{H}_s$ defined by
    \begin{equation}
        \psi(\tau, t) := [\Psi(\tau)](t) = [\mathbf{U}(\tau) \Psi_0] (t)
    \end{equation}
    which solves the Schrödinger equation of interest.
    \begin{align}
    \begin{aligned}
        i\partial_\tau \psi(\tau, t) &= i \partial_\tau [\mathbf{U}(s) \Psi_0](t) \\
        &= [\mathbf{H}(\tau)\Psi_0 ](t) \\
        &=H(\tau + t) \psi(0, t)
    \end{aligned}
    \end{align}
    The interpretation is that each $t$ constitutes an initial time for performing the simulation, so we have a family of solutions parameterized by $t$ with initial state $\psi(0,t)$. The evolution parameter acts, as expected, as the total time elapsed in the simulation. Finally, we can obtain a collection of time-evolution operators $U(t + \tau, t)$ on $\mathcal{H}_s$ for each $t \in [0,T]$ as follows.
    \begin{equation}
        U(t + \tau, t) \psi_0 := [\mathbf{U}(\tau) \Psi_0] (t)
    \end{equation}
    where $\Psi_0 \in \mathcal{H}$ is any state for which $\Psi_0 (t) = \psi_0$. This operator is unitary and solves the operator Schrödinger equation~\eqref{eq:Schrödinger_eqtn}. 
    
    To summarize, the operator $\mathbf{U}$ of equation~\eqref{eq:clockU} encodes a one-parameter family of time evolution operators for the system of interest, parameterized by the initial time. Thus, we have shown how the propagator $U$ generated by a time dependent $H$ can be cast as an ordinary operator exponential on an augmented space. Interesting in its own right, this framework also allows for a natural unification of ideas regarding ``Trotterization." This term is used to refer to both (a) the splitting up of an (ordinary) operator exponential of $H = \sum_j H_j$ into exponentials of the various $H_j$, or (b) the simulation of a time dependent Hamiltonian by time independent simulations over small time intervals. These can be seen as manifestations of the same phenomenon. To illustrate with a pertinent example, consider a symmetric Trotterization of the unitary $\mathbf{U}$.
    \begin{align}
    \begin{aligned}
        U_2(t + \tau, t) &:= e^{iE\tau}\left(e^{-iE\tau/2}e^{-i H(t)\tau}e^{-iE\tau/2}\right) \\
        &= e^{-i H(t + \tau/2) \tau}
    \end{aligned}
    \end{align}
    We have just derived the \emph{midpoint formula}~\cite{Wiebe2010Higher,Suzuki2003Methodology} from scratch. The Trotter product theorem says that 
    \begin{equation}
        \lim_{k\rightarrow\infty} e^{iE\tau}\left(e^{-iE\tau/2k}e^{-i H(t)\tau/k}e^{-iE\tau/2k}\right)^k = U(t + \tau, t).
    \end{equation}
    Note that this holds even though $E$ is unbounded~\cite{Reed1980MathPhysics}. Thus,
    \begin{equation}
        U_2^{(k)} (t + \tau, t) = e^{iE\tau}\left(e^{-iE\tau/2k}e^{-i H(t)\tau/k}e^{-iE\tau/2k}\right)^k
    \end{equation}
    constitutes a good approximation to $U$ for sufficiently large $k \in \mathbb{Z}_+$. This opens up the possibility of a more unified approach to Hamiltonian simulation algorithms that has not yet been properly considered.
        
\subsection{Multiproduct Formulas} \label{sec:MultiproductFormulas}
    Multiproduct formulas (MPFs) are a generalization of the celebrated product formulas, and span two of the great pillars of quantum simulation.  The aim of the MPF is to approximate the time evolution operator $U$ as a linear combination of lower-order Trotter formulas, in such a way that higher order errors are cancelled~\cite{Chin2010Multiproduct,Childs2012LCU,Low2019Multiproduct}. They are, fundamentally, nothing more than a Richardson extrapolation of a product formula $\mathcal{P}$ to Trotter step size $s\rightarrow 0$. This extrapolation is done to address the primary deficiency of product formulas, which is that the number of exponentials used in the $2n^\mathrm{th}$-order formula scales as $5^n$. Product formulas, unfortunately, cannot be easily optimized beyond this. As the MPF is a sum of product formula approximations, the number of error terms in the expansion does not grow exponentially. This allows us to approximate the quantum dynamics using polynomially many, rather than exponentially many, operator exponentials. 
    
    We now reference a theorem~\cite{Low2019Multiproduct} which justifies the effectiveness of MPFs in the time independent setting, while also implicitly defining them.
    \begin{theorem}[Time Independent MPFs (Theorem 1 of~\cite{Low2019Multiproduct})]\label{thm:low}
        Let $H$ be a bounded, time independent Hamiltonian, and let $U_2(t)$ be the $2^\mathrm{nd}$-order Suzuki-Trotter formula for the time evolution operator $U(t) = e^{-i H t}$. Let $a = (a_1, a_2, \dots, a_m) \in \mathbb{R}^m$ and $\vec{k} = (k_1, k_2, \dots, k_m) \in \mathbb{Z}_+^m$. There exist choices of $a$ and $\vec{k}$ such that multiproduct formula,
        \begin{equation*}
            U_{2,m}(t) := \sum_{j=1}^m a_j  U_2^{k_j}(t/k_j)
        \end{equation*}
        is order $2m$ and satisfies
        \begin{equation*}
            \max_j k_j \in O(m^2),\qquad \|a\|_1 \in O(\mathrm{polylog}(m)).
        \end{equation*}
    \end{theorem}
    The details of the proof can be seen in~\cite{Low2019Multiproduct}, but at a high level, the MPF $U_{2,m}$ is a Richardson extrapolation of $U_2$ with respect to the Trotter step size parameter $1/k$. Such an extrapolation is possible for arbitrary $m$ because there exists an error series~\cite{Blanes1999Extrapolation} 
    \begin{equation}
        U_2^k(t/k) - U(t) = \sum_{j=1}^\infty E_{2j+1} \frac{t^{2j+1}}{k^{2j}}
    \end{equation}
    with $E_{2j+1}$ independent of $k$ (but not $t$ generically). The existence of this series suffices for a $1/k \rightarrow 0$ Richardson extrapolation~\cite{Sidi2003Richardson}. In particular, cancellation occurs for coefficients $a_j$ satisfying the following Vandermonde linear system.
    \begin{equation} \label{eq:vandermonde}
        \begin{pmatrix} 
            1 &\cdots &1 \\
            {k_1}^{-2} &\cdots &{k_m}^{-2} \\
            \vdots &\ddots & \vdots\\
            {k_1}^{-2m+2} &\cdots &{k_m}^{-2m+2} \\
            \end{pmatrix}
            \begin{pmatrix}
            a_1\\
            a_2\\
            \vdots\\
            a_m
            \end{pmatrix}
            =
            \begin{pmatrix}
            1 \\ 0 \\ \vdots \\ 0
        \end{pmatrix}
    \end{equation}
    Though the matrix is ill-conditioned, this is irrelevant to the matrix inversion, as the inverse Vandermonde matrix admits an analytic solution that may be reasoned from the theory of polynomial interpolation. What matters for our application is the one-norm $\|a\|_1$ of the coefficients, which serves as our condition number because of the need to amplify an amplitude of size $1/\|a\|_1$ in the LCU procedure. The content of Theorem~\ref{thm:low} is that Trotter steps $\vec{k}$ may be chosen such that $\|a\|_1$ is not too large. For time-ordered $U$, the analysis of~\cite{Blanes1999Extrapolation} does not carry over, although reasonable ``time dependent" MPFs can be defined heuristically. One of our motivations in constructing a clock space is to be able to eliminate time ordering an rigorously show these formulas work. 

    As discussed in~\cite{Low2019Multiproduct}, specific choices of $k_j$ can be found numerically to minimize $\|a\|_1$, and this may be the best approach in practice. However, for our analytical results it will be most appropriate to utilize the specific $k_j$ chosen in their constructive proof of well-conditioned MPFs. Thus, for all results we will take the powers $k_j$ as follows.
    \begin{equation} \label{eq:kj_def}
        k_j = \left\lceil \frac{\sqrt{8}m}{\pi}\abs{\sin\left(\frac{\pi(2j-1)}{8m}\right)}^{-1}\right\rceil,\qquad j = 1,\dots,m
    \end{equation}
    We will use these same coefficients even in the time dependent MPFs to be introduced in~\sec{TDMPF}. For error analysis, it will be useful to have simple, concrete bounds on $k_j$. We can achieve this by noting that $\sin(x)\leq x$ and $\sin(x)\geq 4x/5$ for $x\in[0,1]$. This gives the lower bound
    \begin{equation} \label{eq:kjLB}
        k_j\geq \left\lceil\frac{8^{3/2}m^2}{\pi^2(2j-1)}\right\rceil\geq\left\lceil\frac{8^{3/2}m^2}{\pi^2(2m-1)}\right\rceil>\frac{\sqrt{128}m}{\pi^2}>m
    \end{equation}
    and the upper bound
    \begin{equation}\label{eq:kjUB}
        k_j\leq \left\lceil \frac{5\times8\sqrt{8}m^2}{4(2j-1)\pi^2}\right\rceil\leq \left\lceil \frac{5\times8\sqrt{8}m^2}{4\pi^2}\right\rceil < 3m^2.
    \end{equation}
    Note the consistency of~\eqref{eq:kjUB} with the big-$O$ scaling of~\thm{low}.

\section{Finite Dimensional Clock Spaces} \label{sec:Finite_Clock_Space}
        
    We now introduce our finite dimensional clock space, which we will sometimes call the ``clock register" to distinguish from the continous version. We discretize the clock variable $t$ into $N_c = N_p\times N_q$ basis states, where $N_p \in \mathbb{Z}_+$ will represent the number of ``Trotter steps" used in the simulation. Each is further divided into $N_q \in \mathbb{Z}_+$ steps for reasons that will be discussed shortly. We label these orthonormal basis states $\ket{j}$ for $j \in [0,N_c-1] \cap \mathbb{Z}$. We will find it useful to consider, for our purposes, only periodic Hamiltonians. This is natural since translation operators like $E$ act most naturally on $\mathbb{R}$ or with periodic boundary conditions. Nonperiodic Hamiltonians can be accommodated by a simple reflection, defining $H(T + t) := H(T - t)$ for $t \in [0,T]$. In our work below, we will want $H(t)$ to be a differentiable bounded function within the grid points, and although the reflection introduces nonsmoothness, we can simply take one of the grid points to be the midpoint of simulation.

    For simplicity, and for lack of a compelling alternative, we will take these grid points $(t_j)_{j=0}^{N_c-1}$ to be uniformly spaced over the interval $[0,T]$: $t_j = T j/N_c$ (taking $N_c$ to be an even integer, so that the midpoint requirement discussed directly above is satisfied). We let $\delta t := T/N_c$ denote the grid width. We also take the natural discretization of $H(t)$ onto the clock space.
    \begin{equation} \label{eq:CH_def}
        H(t) \mapsto \sum_{j=0}^{N_c - 1} H_j \otimes \vert j\rangle\langle j \vert =: C(H)
    \end{equation}
    where $H_j \equiv H(t_j)$. Observe that $C(H)$ has no dependence on the evolution parameter; it is time independent. The notation $C(H)$ is used to suggest a controlled operation, where the control is on the clock register.

    Choosing the appropriate discretization of $E$ is somewhat more tricky, though the choice appears obvious in hindsight. Since $E$ acts as a derivative, it makes sense to take the discretized version to be a finite difference operator. For example,
    \begin{equation} \label{eq:Deltadef}
        \Delta := -i \frac{U_+ - U_-}{2 \delta t}
    \end{equation}
    where $U_+$ is the shift operator defined by $U_+\vert j\rangle = \vert j+1\rangle$ and $U_- = U_+^\dagger$ is the backwards shift (all increments taken mod $N_c$). This is the approach we ultimately take. However, we note that the authors began by considering a distinct approach via the logarithm of the translation operator
    \begin{equation}
        \tilde{\Delta} = i \log U_+.
    \end{equation}
    While apparently sensible, given the analogous relation between $E$ and shifts on the clock space, this operator is not nicely behaved. For example, its commutator with the ``position operator" $\sum_j t_j \vert j\rangle\langle j
    \vert$, rather than being near-identity, has long off-diagonal tails. This behavior may be of independent interest, but from now on we will concern ourselves with $\Delta$ as the discrete version of $-E$.

    With these choices, our full clock Hamiltonian becomes 
    \begin{equation}
        H_c := C(H) - \Delta. 
    \end{equation}
    Already, we can show some reasonable properties carry over to this setting.
    \begin{lemma} \label{lem:canonical_commutator}
        In the notation above, let $H:[0,T]\rightarrow \mathrm{Herm}(\mathcal{H})$ be a time dependent Hamiltonian on a finite-dimensional vector space $\mathcal{H}$. Then
        \begin{equation}
            [\Delta, C(H)] = i \Re\left(U_+\sum_j \frac{H_{j+1} - H_j}{\delta t} \otimes \vert j\rangle\langle j\vert\right)
        \end{equation}
        where $\Re(A) := (A + A^\dagger)/2$ denotes the Hermitian part of $A$. If $H$ is differentiable in each subinterval with bounded derivative, then we further have
        \begin{equation}
            \|[\Delta, C(H)]\| \leq \max_{t \in [0,T]} \|\dot{H}(t)\|.
        \end{equation}
    \end{lemma}
    We remark here the connections to the canonical commutation relation $[f(x),p] = i f'(x)$. The additional shift by $U_+$ is a relatively small deviation from a finite difference approximation being performed on the Hamiltonian. The proof is relatively straightforward and provided in~\app{clocklemmas}.

    Having defined the clock space and Hamiltonian, we wish to prepare a suitable initial state. A seemingly adequate and natural choice is to take $\vert\psi_0\rangle\otimes \vert0\rangle$, where $\vert\psi_0\rangle$ is the initial state of the system of interest and $\vert0\rangle$ is the clock state at the initial time $t = 0$. However, problems immediately arise which can be traced to the fact that the continuous version of $\ket{0}$ is $\delta(t)$, which is not a normalizable state vector. This formal problem finds its way into the discrete setting, in that the finite difference $\Delta$ does not properly compute a derivative of $\ket{0}$. Thus, $\Delta$ fails to translate $\ket{0}$ properly into later times, and the time dependent simulation fails.

    To fix this issue, we take a cue from the continuous setting, where the best we can do is take a wavepacket of small enough width to suit our purposes. For simplicity, this wavepacket may as well be Gaussian, with some width $\sigma$ to be chosen with care. Thus, we introduce Gaussian functions
    \begin{equation}
        \phi_{\mu}(t;\sigma) = \frac{1}{\sqrt{\mathcal{N}}}e^{-\lvert t -\mu\rvert_c^2/\sigma^2}.
    \end{equation}
    of width $\sigma \in \mathbb{R}_+$ and center $\mu \in [0,T)$. Here $\lvert\cdot\rvert_c$ is the shortest distance to $0$ modulo $T$,
    \begin{equation}
        \lvert t\rvert_c := \min\left\{\lvert t\rvert, \lvert T - t\rvert\right\}
    \end{equation}
    so that, with $0$ and $T$ identified, $\phi_\mu$ is smooth everywhere except $\mu + T/2\mod{T}$. Moreover, $\mathcal{N} \in \mathbb{R}_+$ is chosen such that the discretized vector
    \begin{equation} \label{eq:gaussian-state}
        \vert\phi_\mu \rangle = \sum_j \phi_\mu (t_j;\sigma) \vert j\rangle
    \end{equation}
    is normalized in the Euclidean sense (i.e., a quantum state vector). Technically, $\mathcal{N}$ has some dependence on $\mu$, but in our case we will only consider $\mu = t_j$ for some $j$, in which case $\mathcal{N}$ only depends on parameters such as $N_c$ and $\sigma$. Because of this choice, we will more simply write $\vert\phi_j\rangle \equiv \vert \phi_{t_j} \rangle$.  

    We are now ready to more clearly elucidate the overall strategy of the clock space construction. \fig{clock_schematic} gives a schematic of the relevant components. As stated above, each of the $N_p$ should be thought of as a single Trotter step in the evolution under $H(t)$. The $N_q$ subintervals ensure that $\delta t$ is sufficiently small such that the approximation of $\Delta$ to a derivative of $\phi_j$ holds. In particular, we will desire $\sigma \gg \delta t$. On the other hand, we want the variation of $H$ within the envelope of $\phi_j$ to be small. That is, we want $\sigma < T/N_p$. Because, presumably, we've chosen each Trotter step sufficiently small, this ensures that $H$ is approximately constant over the bulk of $\vert\phi_j\rangle$. Of course, we will want to ensure all of the above conditions with as few resources, such as clock register states, as possible.
    \begin{figure}
        \centering
        \includegraphics[scale=0.2]{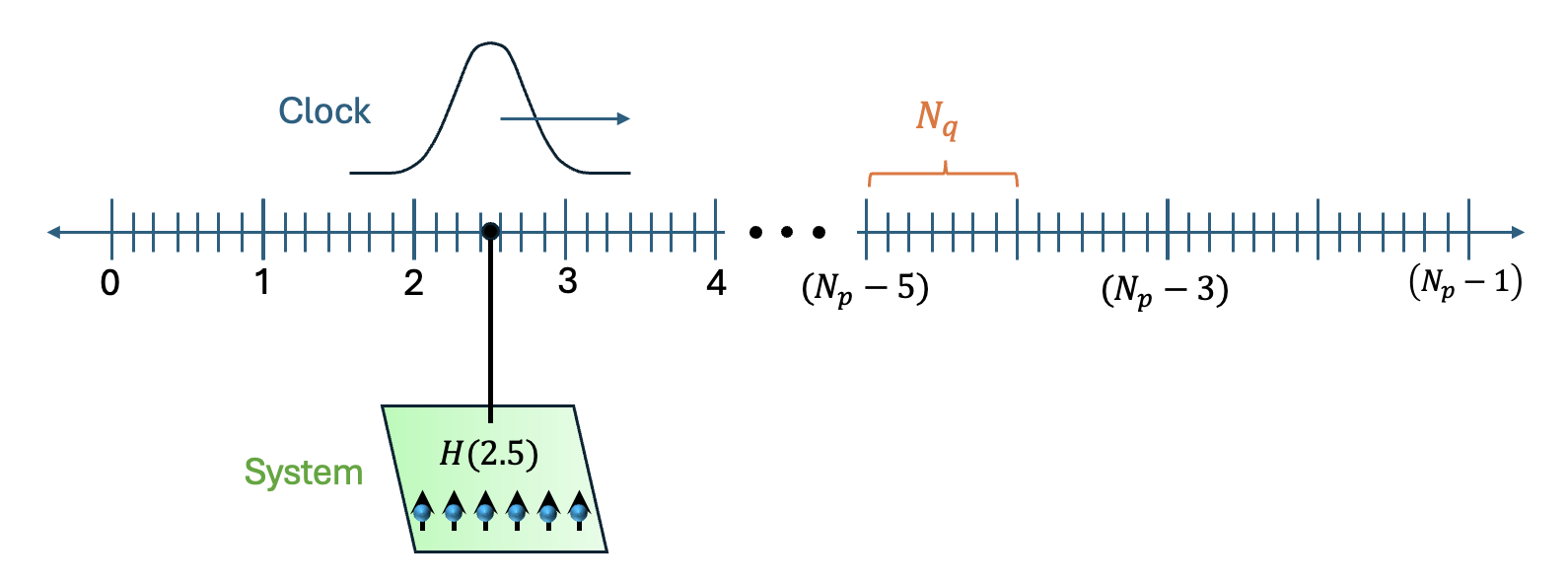}
        \caption{Schematic of the discrete clock Hilbert space. The clock register has an initially prepared Gaussian state which is translated uniformly under the clock Hamiltonian. Its location controls the Hamiltonian applied to the system of interest. The Hamiltonian varies little over each of the $N_p$ large steps, and the Gaussian is wide compared to the $N_q$ subdivisions within each large step. }
        \label{fig:clock_schematic}
    \end{figure}

    We now begin to characterize the simulation error in using the clock space for approximating $U(T,0)$. First, it will be helpful to have a characterization of the size of the normalization $\mathcal{N}$.
    \begin{lemma}\label{lem:normalization}
        In the notation above, the normalization constant $\mathcal{N} \in\mathbb{R}_+$ for Gaussian states $\vert\phi_j\rangle$ peaked at $\mu = t_j$ satisfies
        \begin{equation}
            \frac{1}{\sqrt{\mathcal{N}}} \in O(\sqrt{\delta t/\sigma}). 
        \end{equation}
    \end{lemma}
    The proof is provided in~\app{clocklemmas}. With this technical lemma in hand, we turn to showing that $\Delta$ indeed acts as a generator of translations on the clock space for $\vert\phi_j\rangle$, provided $\sigma$ is large relative to $\delta t$ and that the Gaussian is not truncated by small $T$.
    \begin{lemma} \label{lem:Delta_effect}
        In the notation introduced in this section, for any $m \in \mathbb{Z}_+$ we have
        \begin{equation}
            e^{i \Delta m\delta t}\vert\phi_j\rangle = \vert\phi_{j+m}\rangle + O\left(m(\delta t/\sigma)^2 + m\sqrt{\delta t/\sigma} e^{-(T/2\sigma)^2}\right)
        \end{equation}
        where the asymptotics $O$ are understood to be taken as $\delta t/\sigma \rightarrow 0$ and $\sigma/T \rightarrow 0$.
    \end{lemma}
    \begin{proof}
        Performing a 1st order Taylor expansion of the exponential,
        \begin{equation} \label{eq:Delta_Taylor_poly}
            e^{i \Delta \delta t}\vert\phi_j\rangle = \vert\phi_j\rangle + i \delta t \Delta \vert\phi_j\rangle + R_1(\delta t)\vert\phi_j\rangle,
        \end{equation}
        where $R_1$ is the Taylor remainder operator
        \begin{equation} \label{eq:R1_Delta2}
            R_1(\delta t) = \delta t \int_0^{\delta t} \frac{\partial^2}{\partial \tau^2}e^{-i \Delta \tau} d\tau = - \int_0^{\delta t} e^{-i \Delta \tau} d\tau (-\delta t \Delta^2).
        \end{equation}
        Thus, the error can be bounded, via the triangle inequality for integrals, as
        \begin{equation}
            \| R_1(\delta t) \vert\phi_j\rangle \| \leq \delta t^2 \|\Delta^2 \vert\phi_j\rangle\|
        \end{equation}
        The action of $\Delta$ on discretized functions $\ket{g}$ of the clock space is given by
        \begin{align}
        \begin{aligned}
            \Delta \ket{g} &= -i \sum_{j=0}^{N_c -1} g(t_j) \left(\frac{\ket{j+1} - \ket{j-1}}{2\delta t}\right) \\
            &= i \sum_j \frac{g(t_{j+1}) - g(t_{j-1})}{2\delta t} \ket{j} \\
            &= i \ket{D_{\delta t} g}.
        \end{aligned}
        \end{align}
        Here $D_{\delta t}f(x) := \frac{f(x+\delta t) - f(x-\delta t)}{2\delta t}$ is the symmetric finite difference of halfwidth $\delta t$ at point $x$. Thus, $\Delta^2 \vert\phi_j\rangle = - \vert D_{\delta t}^2 \phi_j\rangle$.
        We consider the error of this finite difference in terms of an approximation to the derivative for values of $t$ within $T/2 - 2 \delta t$ of $t_j$ in circle distance. On this part of the domain, $\phi_j(t \pm 2\delta t)$ is smooth, hence
        \begin{equation}
            \vert D_{\delta t}^2 \phi_j\rangle = \vert \partial_t^2 \phi_j\rangle + O(\delta t^2 \phi_j^{(4)})
        \end{equation}
        where the superscript $(4)$ indicates a fourth derivative. Near the edge of the Gaussian, the second-derivative property does not hold; however, these parts of the state vector have amplitude which is on the order $O(\mathcal{N}^{-1/2}e^{-(T/2\sigma)^2})$, which by Lemma~\ref{lem:normalization} is $O(\sqrt{\delta t/\sigma} e^{-(T/2\sigma)^2})$. This gets multiplied by $\delta t^{-2}$ due to the second finite difference $D_{\delta t}$ being taken. Taking the two sources independently as an upper bound, we have
        \begin{align}
        \begin{aligned}
            \|\Delta^2 \vert\phi_j\rangle\| &= \|\ket{\partial_t^2 \phi_j}\| +  O\left(\delta t^2/\sigma^4 + (\sigma \delta t^3)^{-1/2}e^{-(T/2\sigma)^2}\right) \\
            &\in O\left(1/\sigma^2 + \delta t^2/\sigma^4 + (\sigma \delta t^3)^{-1/2}e^{-(T/2\sigma)^2}\right)
        \end{aligned}
        \end{align}
        where $\sigma^{-4}$ comes from the four derivatives of the Gaussians. Thus, the total Taylor remainder may be upper bounded using~\eqref{eq:R1_Delta2} as
        \begin{equation}
            \| R_1(\delta t) \vert\phi_j\rangle\| \in O\left((\delta t/\sigma)^4 + \sqrt{\delta t/\sigma} e^{-(T/2\sigma)^2}\right).
        \end{equation}
        
        To complete the proof we return to the linear Taylor expansion in~\eqref{eq:Delta_Taylor_poly}. Using similar reasoning to above,
        \begin{align}
        \begin{aligned}
            \vert\phi_j\rangle - i\delta t\Delta \vert \phi_j\rangle &= \vert\phi_j\rangle + \delta t \vert D_{\delta t} \phi_j\rangle \\
            &= \vert\phi_j\rangle + \delta t\vert \partial_t \phi_j\rangle + O\left((\delta t/\sigma)^2 + \sqrt{\delta t/\sigma} e^{-(T/2\sigma)^2}\right).
        \end{aligned}
        \end{align}
        Finally, what remains is a linear approximation to $\vert\phi_{j+1}\rangle$, with error also $(\delta t/\sigma)^2$. Keeping only the leading terms, notice that the Taylor remainder error is subdominant. Altogether,
        \begin{equation}
            e^{-i \Delta \delta t}\vert\phi_j\rangle = \vert\phi_{j+1}\rangle + O\left((\delta t/\sigma)^2 + \sqrt{\delta t/\sigma} e^{-(T/2\sigma)^2}\right).
        \end{equation}
        So far, we've proved the result for $m = 1$. The full result follows by noting that $e^{-i\Delta m \delta t} = (e^{-i \Delta \delta t})^m$ and taking, as upper bound, $m$ times the error of a single step.
    \end{proof}
    We note that the error in $\Delta$ generating translations comes from two sources: the discretization at small scales and the boundary effects at large scales. We might name these, in the language of lattice field theory, ultraviolet and infrared truncation effects, respectively.

    Our next intermediate result will be concerned with the evolution of $C(H)$ controlled on the Gaussian state $\vert\phi_j\rangle$. We want the result to be, approximately, an evolution under $H(t_j)$ on the main register of interest. In what follows, it will be convenient to take $\tau := T/N_p$ as the time duration of a larger subdivision of steps.
    \begin{lemma} \label{lem:C(H)_effect}
        Let $H:[0,T]\rightarrow \mathrm{Herm}(\mathcal{H})$ be a bounded differentiable function with bounded derivative. For any $\eta \in \mathbb{R}$, we have
        \begin{equation}
        \begin{split}
            e^{-i C(H) \eta} \vert\psi\rangle \vert\phi_j\rangle &= e^{-i H(t_j)\eta}\vert\psi\rangle \vert\phi_j\rangle \\
            &+ O\left(\eta \tau \max_{t\in[0,T]}\|\dot{H}(t)\| + (1 + \eta \max_{t\in[0,T]} \|H(t)\|) e^{-\tau^2/4\sigma^2}\right)
        \end{split}
        \end{equation}
        where $\tau := T/N_p$.
    \end{lemma}
    \begin{proof}
        We begin by grouping the terms of $C(H)$ into two chunks: one with significant overlap with the Gaussian, the other with small overlap. Specifically, we take $C(H) = H_\mathrm{av} + H_\perp$, with
        \begin{align}
        \begin{aligned}
            H_\mathrm{av} :&= \sum_{k = j - N_q/2}^{j+N_q/2 -1} H_k \otimes \vert k\rangle\langle k\vert \\
            H_\perp :&= C(H) - H_\mathrm{av}.
        \end{aligned}
        \end{align}
        Because $H_\mathrm{av}$ and $H_\perp$ commute, we can Trotterize with no error
        \begin{equation}
            e^{-i C(H) \eta}\vert\psi\rangle\vert\phi_j\rangle = e^{-i H_\perp \eta} e^{-i H_\mathrm{av} \eta}\vert\psi\rangle\vert\phi_j\rangle.
        \end{equation}
        We will show that the $H_\mathrm{av}$ term gives approximately $H(t_j)$, while $H_\perp$ acts as approximately the identity (with the right parameter values).

        First, consider $e^{-i H_\mathrm{av}\eta}$. Define $P_j = \sum_k \vert k\rangle\langle k\vert$ as the projector onto the clock states on which $H_\mathrm{av}$ has support ($k \in \mathbb{Z} \cap [j - N_q/2, j + N_q/2 -1]$). We have
        \begin{equation}
            \|e^{-i H_\mathrm{av} \eta} - e^{-i H_j \otimes P_j \eta} \| \leq \eta \|H_\mathrm{av} - H_j \otimes P_j\|.
        \end{equation}
        Meanwhile,
        \begin{equation}
            \left\|H_\mathrm{av} - H_j \otimes P_j\right\| = \left\|\sum_{k = j - N_q/2}^{j+N_q/2 -1} (H_k - H_j)\otimes \vert k\rangle\langle k\vert\right\| = \max_k \|H_k - H_j\|
        \end{equation}
        By an simple Taylor bound, $\| H_k - H_j\| \leq (\tau/2) \max_t \|\dot{H}(t)\|$, were the max is over $[t_k, t_j]$ (taking the appropriate ordering of $t_j, t_k$ if needed). We can therefore say
        \begin{equation}
            \|e^{-i H_\mathrm{av} \eta} - e^{-i H_j \otimes P_j \eta} \| \leq \eta \tau \max_{t\in[0,T]}\|\dot{H}\|
        \end{equation}
        so that, up to this error, we can replace a simulation by $H_\mathrm{av}$ with $H_j\otimes P_0$. Moving on to this situation, we have
        \begin{equation}
            e^{-i H_j \otimes P_0 \eta} \vert\psi\rangle \otimes \vert\phi_j\rangle 
            = e^{-i H_j \eta}\vert\psi\rangle P_j\vert\phi_j\rangle + \vert\psi\rangle (I-P_j)\vert\psi_j\rangle.
        \end{equation}
        Thinking of $\sigma < \tau$ and taking $\tau/\sigma$ increasing, we have $P_0 \vert\phi_j\rangle = \vert\phi_j\rangle + O\left(e^{-\tau^2/4\sigma^2}\right)$. Thus,
        \begin{equation}
            e^{-i H_j \otimes P_0 \eta} \vert\psi\rangle \vert\phi_j\rangle = e^{-i H_j \eta}\vert\psi\rangle \vert\phi_j\rangle + O\left(\eta \tau \max_{t\in[0,T]}\|\dot{H}\| + e^{-\tau^2/4\sigma^2}\right)
        \end{equation}
        
        For the remainder of the proof take, $\vert\psi'\rangle = e^{-i H_j \eta}\vert\psi\rangle$ for notational convenience. We now consider the action of $H_\perp$ on the remaining state, which we anticipate to be small. First,
        \begin{equation}
            \left\|e^{-i H_\perp \eta}\vert\psi'\rangle\vert\phi_j\rangle - \vert\psi'\rangle\vert\phi_j\rangle\right\| \leq \eta \|H_\perp \vert\psi'\rangle\vert\phi_j\rangle\|.
        \end{equation}
        Let $\mathcal{J}$ be an index set for all the time steps included in the summation $H_\mathrm{av}$. We have
        \begin{align}
        \begin{aligned}
            \|H_\perp\vert\psi'\rangle\vert\phi_j\rangle\| &= \left\|\sum_{k\notin \mathcal{J}} H_k \vert\psi'\rangle \vert k\rangle \langle k\vert\phi_j\rangle\right\| \\
            &\leq \sqrt{\sum_{k\notin\mathcal{J}} \frac{1}{\mathcal{N}} e^{-2 \vert t_j - t_k\vert_c^2/\sigma^2}\|H_k\|^2}.
        \end{aligned}
        \end{align}
        Employing a Hölder inequality on the inner product, followed by~\lem{normalization},
        \begin{align}
        \begin{aligned}
            \sqrt{\sum_{k\notin\mathcal{J}} \frac{1}{\mathcal{N}} e^{-2 \vert t_j - t_k\vert_c^2/\sigma^2}\|H_k\|^2} &\leq \max_{k\notin\mathcal{J}} \|H_k\| \sum_{k\notin \mathcal{J}} \frac{e^{-|t_j - t_k|_c^2/\sigma^2}}{\mathcal{N}}\\
            &\in O\left(\max_{t}\|H(t)\| (\delta t/\sigma)\sum_{k = N_q/2}^\infty e^{-k^2 \delta t^2/\sigma^2}\right).
        \end{aligned}
        \end{align}
        Following a similar procedure to before, we convert to an error function $\erf$ and take an exponential upper bound. Doing so gives
        \begin{equation}
            \|H_\perp \vert\psi'\rangle\vert\phi_j\rangle\| \in O\left(\max_{t\in[0,T]} \|H(t)\| e^{-\tau^2/2\sigma^2}\right)
        \end{equation}
        Thus, $e^{-i H_\perp \eta}$ acts trivially on this state up to $O\left(\eta \max_t\|H(t)\| e^{-(\tau/2\sigma)^2}\right)$. 
        
        Combining the errors together, we take the widest exponential $e^{-\tau^2/4\sigma^2}$ as a simple upper bound for all exponentials that appear. Putting all the error sources together gets us the result of the Lemma statement.
    \end{proof}
    With the previous two lemmas, we have the ingredients needed for a clock space simulation: controlled operations and time shifts. We combine them to show that our clock space indeed encodes time dependent dynamics.
    \begin{theorem} \label{thm:DiscreteClockErrorScaling}
        Let $H:[0,T]\rightarrow\mathrm{Herm}(\mathcal{H})$ be a time dependent Hamiltonian on a finite dimensional vector space $\mathcal{H}$, such that $H(t)$ as a function is bounded and differentiable with bounded derivative. Then, the clock Hamiltonian, with Gaussian input $\vert\phi_0\rangle$ approximately applies the time evolution operator $U(T,0)$ to an initial state $\vert\psi_0\rangle \in \mathcal{H}$. Precisely, 
        \begin{equation}
        \begin{split}
            &e^{-i H_c T}\vert\psi_0\rangle\vert\phi_0\rangle= \left(U(T,0)\vert\psi_0\rangle\right) \vert\phi_0\rangle \\
            &+O\left(T \delta t/\sigma^2 + \sqrt{N_c T/\sigma} e^{-T^2/4\sigma^2} + \max_{t}\|\dot{H}\| \frac{T^2}{N_p} + e^{-\tau^2/4\sigma^2}(N_p + \max_{t}\|H\|T)\right)
        \end{split}
        \end{equation}
    \end{theorem}
    \begin{proof}
        Let $\tau = T/N_p$. We begin with a first-order Trotterization of $H_c$ into $N_p$ steps.
        \begin{equation}
            e^{-i H_c T} = \left(e^{-i \Delta \tau}e^{-i C(H) \tau}\right)^{N_p} + O\left(\max_{t\in[0,T]} \|\dot{H}(t)\| \frac{T^2}{N_p}\right)
        \end{equation}
        With initial state $\vert\psi_0\rangle\vert\phi_0\rangle$, combining Lemmas~\ref{lem:C(H)_effect} and~\ref{lem:Delta_effect} gives the following error for a single Trotter step.
        \begin{equation}
        \begin{split}
            &e^{-i \Delta \tau}e^{-i C(H)\tau} \vert\psi_0\rangle\vert\phi_0\rangle = e^{-i H_0\tau}\vert\psi_0\rangle \vert\phi_{N_q}\rangle  \\
            &+O\left(\tau \delta t/\sigma^2 + \sqrt{N_q \tau/\sigma} e^{-T^2/4\sigma^2} + \tau^2 \max_{t}\|\dot{H}\| + (1 + \tau \max_{t}\|H\|) e^{-\tau^2/4\sigma^2}\right).
        \end{split}
        \end{equation}
        Thus, after all $N_p$ steps, we can multiply the single step error above to get an upper bound of
        \begin{equation}
        \begin{split}
            &\left(e^{-i \Delta \tau}e^{-i C(H) \tau}\right)^{N_p}\vert\psi_0\rangle\vert\phi_0\rangle = e^{-i H_{N_q(N_p-1)}\tau}\dots e^{-i H_{N_q}\tau}e^{-i H_0 \tau}\vert\psi_0\rangle\vert\phi_0\rangle \\
            &+ O\left(T \delta t/\sigma^2 + \sqrt{N_c T/\sigma} e^{-T^2/4\sigma^2} + \max_{t}\|\dot{H}\| \frac{T^2}{N_p} + e^{-\tau^2/4\sigma^2}(N_p + \max_{t}\|H\|T)\right)
        \end{split}
        \end{equation}
        The right side, without the error, is a 1st order Suzuki Trotter splitting, which approximates $U(T,0)$ to order $\max_{t\in[0,T]} \|\dot{H}(t)\| T^2/N_p$. This can be absorbed into the third term of the big-$O$. This gives the result stated in the Theorem.
    \end{proof}
    With this result in hand, we now show that the parameters $(N_p, N_q, \sigma)$ of the clock can be chosen such that any desired degree of approximation to $U(T,0)$ can be achieved.
    \begin{theorem} \label{thm:clock_limit}
        In the context of the previous theorem, for any $\epsilon \in \mathbb{R}_+$, there exists clock parameters $(N_p, N_q, \sigma)$ such that 
        \begin{equation}
            \left\|e^{-i H_c T}\vert\psi_0\rangle\vert\phi_0\rangle - U(T,0)\vert\psi_0\rangle\vert\phi_0\rangle\right\| < \epsilon
        \end{equation}
        with $(N_p, N_q)$ scaling as
        \begin{align}
        \begin{aligned}
            N_p \in \Theta\left(\max_{t\in[0,T]}\|\dot{H}\|\frac{T^2}{\epsilon}\right),\quad N_q \in \Theta\left( \frac{\max_t\|\dot{H}\| T^2}{\epsilon^2} x^2\right),\quad \sigma\in \Theta\left(\frac{\epsilon}{\max_t \|\dot{H}\| T x}\right). 
        \end{aligned}
        \end{align}
        Here,
        \begin{align}
        \begin{aligned}
            x &:= \sqrt{\log\left(\frac{\Gamma\,T}{\epsilon}\right)} \\
            \Gamma &:= \max\left\{\max_{t\in[0,T]} \|\dot{H}\| T, \epsilon \max_{t\in[0,T]}\|H\| \right\}.
        \end{aligned}
        \end{align}
        In particular, there exists a sequence $(N_p(j), N_q(j), \sigma(j))$ of clock space parameters, such that 
        \begin{equation}
            \lim_{j\rightarrow\infty} \Tr_c(e^{-i H_c T}\vert\psi_0\rangle\vert\phi_0\rangle) = U(T,0)\vert\psi_0\rangle
        \end{equation}
        where $\Tr_c$ is a partial trace over the clock register, and $\Tr_c(\vert \Psi\rangle) \equiv \Tr_c(\vert\Psi\rangle\langle\Psi\vert)$.
    \end{theorem}
    \begin{proof}
        To ensure a total error within $\epsilon$ is achievable, it suffices to ensure that each of the five terms constituting the error in~\thm{DiscreteClockErrorScaling} is within $O(\epsilon)$ independently. From the onset, we will choose $N_p \in \Theta\left(\max_{t\in[0,T]}\|\dot{H}\| T^2/\epsilon\right)$ to satisfy the third term. 

        We next move to understand the necessary $\sigma$ scaling. We parameterize it as
        \begin{equation}
            \sigma = \tau/x
        \end{equation}
        with the hope that $x$ can be chosen to increase slowly (i.e., that the Gaussian states have width only slightly smaller than the Trotter step size). For this, we focus on the last two terms, since they have no $N_q$ dependence (which will set the smallest scales). We seek
        \begin{equation}
            \frac{\max_t \|\dot{H}\| T^2}{\epsilon} e^{-x^2/4} \in O(\epsilon),\qquad \max_t \|H\| T e^{-x^2/4} \in  O(\epsilon)
        \end{equation}
        which can be satisfied provided that $x$ is asymptotically lower bounded as
        \begin{align}
        \begin{aligned}
            x^2 &\in \Omega\left(\log \max\left\{\frac{\max_t \|\dot{H}\| T^2}{\epsilon^2},\frac{\max_t \|H\| T}{\epsilon}\right\}\right) \\
            &= \Omega\left(\log (\Gamma\,T /\epsilon)\right)
        \end{aligned}
        \end{align}
        This sets the scaling for $\sigma$.
        
        We move next to the first term to fix $N_q$, since the 2nd term is expected to be quite small. We require $T\delta t/\sigma^2 \in O(\epsilon)$, which is equivalent to
        \begin{equation}
            \frac{N_p x^2}{N_q} \in O( \epsilon).
        \end{equation}
        Therefore, there exists an $N_q \in \Theta\left(N_p x^2/\epsilon\right)$, satisfying the bound. All that remains is the second term, whose contribution can be easily shown to be subdominant compared to the other sources. Therefore, the choice of parameter scaling suffice to achieve the desired error $\epsilon$. 

        We have shown that any desired precision $\epsilon$ for dynamical simulation can be accommodated for by appropriate choice of clock space parameters. Taking a sequence $\epsilon_j \rightarrow 0$, we see there exists a sequence of clock space evolutions whose limit, restricted to the main register, is $U(T,0)$. 
    \end{proof}

    We've thus shown that finite clock space constructions exist which, for $H(t)$ differentiable on $[0,T]$, approximate the dynamics of $H$ to arbitrary precision. One expects that the differentiability condition can be somewhat relaxed, since it doesn't appear in the continuous setting. Any improvements in error analysis here will enhance the performance guarantees of the qubitization algorithm presented in the following section. 

\section{Time Dependent Qubitization} \label{sec:TDQubitization}
    In~\sec{Finite_Clock_Space}, we developed a clock space construction which encoded a time dependent Hamiltonian as a time independent one on an augmented, finite dimensional space. The removal of time ordering using a clock register opens the door for quantum algorithms for time independent Hamiltonian simulation to simulate the full clock-system dynamics directly. In particular, qubitization is an asymptotically optimal~\cite{Low2019Qubitization} simulation method that can only be applied to time independent $H$. In this section, we propose the simulation of time dependent Hamiltonians using qubitization on the augmented system. To be concrete, we will work with an input model in which $H(t)$ is a linear combination of fixed unitaries with time-varying coefficients. This describes, for example, Pauli matrices on $n$ qubits with fluctuating coefficients. 

    \subsection{Overview}
        We take our main register, encoding the quantum system of interest, and append $n_c$ qubits to provide a clock register of size $N_c = 2^{n_c}$. The product state $\ket{\psi_0}\ket{\phi_0}$ is prepared on the joint register, where $\ket{\psi_0}$ is the initial state of the main register and $\ket{\phi_0}$ is a Gaussian as per equation~\eqref{eq:gaussian-state}. Many protocols for preparing Gaussian states exist~\cite{Klco2020Preparation, Rattew2021Efficient,Rattew2022Functions,Iaconis2024Preparation, Grover2002Creating}. For our purposes, we will simply refer to the approach by Kitaev and Webb~\cite{Kitaev2009Wavefunction,Bauer2021Practical} as efficient enough for our purposes. The Gaussian in our application has nonnegligible support over $O(N_q)$ clock states, and their algorithm scales polynomially in the number of qubits $n_q = \log N_q$ over the Gaussian. This cost is negligible compared to the other simulation costs that we will discuss presently.

        Once the initial state is prepared, we employ qubitization to approximate $e^{-iH_c T}$ on the full register. Given $H(t)$ in LCU form, we need to express $H_c$ in LCU form as well, which is not immediate. This is done through several applications of the Signature Matrix Decomposition (see~\app{signature_matrix_decomposition}). We also truncate $\Delta$ at high frequencies to reduce computational cost, with little loss in accuracy. Details of the LCU decomposition are provided in the next subsection.

        Once $H_c$ is in LCU form, select $\mathtt{SEL}$ and prepare $\mathtt{PREP}$ circuits may be constructed to block encode $H_c$ as
        \begin{equation}
            H_c/\|c\|_1 = (\bra{0}\mathtt{PREP}^\dagger \otimes \openone) \mathtt{SEL} (\mathtt{PREP}\ket{0}\otimes \openone)
        \end{equation}
        where $\|c\|_1$ is the one-norm of the LCU coefficients. Standard qubitization can now be done on this block encoded Hamiltonian~\cite{Low2019Qubitization}. The $\mathtt{PREP}$ circuit must create a ``quasi-uniform" distribution over some number $N$ of states, in the sense that, on the LCU auxiliary register,
        \begin{equation}
            \ket{\mathtt{PREP}} = \sum_{j=1}^{K-1} \sqrt{\delta} \ket{j} + \sum_{j=K}^{N} \sqrt{\delta'} \ket{j}
        \end{equation}
        with $\delta, \delta', K$ and $N$ determined by parameters of simulation. Meanwhile the $\mathtt{SEL}$ circuit will need to apply controlled $U_i$ operations, where $U_i$ is a unitary in the $H(t)$ decomposition, and controlled signature matrices. These second operations can be done with classical, reversible comparator circuits implemented quantumly. Each $\mathtt{SEL}$ will also require a Quantum Fourier Transform and its inverse on the clock register.

    \subsection{LCU Block Encoding}
        We assume $H(t)$ is of the form
        \begin{equation} \label{eq:LCU_form_Ht}
            H(t) = \sum_{i=1}^L \alpha_i(t) U_i
        \end{equation}
        where $U_j$ are Hermitian and unitary (e.g., $n$-qubit signed Pauli operators) and $\alpha_j(t)$ are nonnegative, real-valued functions on $[0,T]$. When we discretize, the coefficients $\alpha_{ij} \equiv \alpha_i(t_j)$ will be particularly important. Expanding out $C(H)$ from equation~\eqref{eq:CH_def} using~\eqref{eq:LCU_form_Ht},
        \begin{align}
        \begin{aligned}
            C(H) &= \sum_{j = 0}^{N_c - 1} \left(\sum_{i=0}^{L-1} \alpha_{ij} U_i\right)\otimes \vert j\rangle\langle j\vert \\
            &= \sum_{i=0}^{L-1} U_i \otimes D_i
        \end{aligned}
        \end{align}
        where
        \begin{align}
            D_i := \sum_{j=0}^{N_c-1} \alpha_{ij} \vert j\rangle\langle j\vert
        \end{align}
        is a diagonal operator on the clock register. Let $\Lambda_i(\delta) \equiv \lceil\max_{j} \abs{\alpha_{ij}}/\delta\rceil$. Using a signature matrix decomposition (see~\app{signature_matrix_decomposition}) we can write
        \begin{align}
            D_i = \sum_{k=1}^{\Lambda_i(\delta)} \delta S_{ik}(\delta) + O(\delta)
        \end{align}
        for $\delta > 0$, where
        \begin{align}
            S_{ik}(\delta) = \sum_{j=0}^{N_c -1} (-1)^{k[k > \alpha_{ij}/\delta]} \ketbra{j}
        \end{align}
        and $[P]$ is the Boolean function for proposition $P$, with $[\mathrm{True}] = 1$ and $[\mathrm{False}]=0$. Thus, we obtain an LCU decomposition of $C(H)$ as 
        \begin{align} \label{eq:CH_LCU}
            C(H) = \delta \sum_{i = 1}^L \sum_{k = 1}^{\Lambda_i(\delta)} U_i \otimes S_{ik}(\delta) + O(L\delta)
        \end{align}
        The prepare circuit $\mathtt{PREP}$ is simple enough because the linear combination is uniform. Therefore, it can be accomplished using a Hadamard gate on each of 
        \begin{align}
            n_{C(H)} \in O\left( \log\sum_{i=0}^L \max_j \abs{\alpha_{ij}}/\delta\right)
        \end{align}
        auxiliary qubits needed for a binary encoding. The unitaries $U_i\otimes S_{ik}(\delta)$ can be selected using two different $\mathtt{SEL}$ circuits: one for the original $U_i$ (presumed available to us) and one for the signature matrices $S_{ik}(\delta)$. These unitaries can be constructed using classical comparator circuits provided that each $\alpha_{ij}$ is computable.
    
        We turn out attention now to $\Delta$, defined in~\eqref{eq:Deltadef}. Although already in LCU form, the coefficient has size $2/\delta t$ and is too large to be desirable. However, $\Delta$ may be truncated at high-frequencies without significant loss of accuracy, reducing the coefficient sizes. To show this, we start by converting $\Delta$ to Fourier space, i.e., diagonalizing via the Quantum Fourier Transform. The result may be computed by diagonalizing $U_+$, and is found to be
        \begin{align} \label{eq:diagDelta}
        \begin{aligned}
            \Delta &= \QFT\sum_{j=0}^{N_c - 1} \frac{N_c}{T} \sin\left(2\pi \frac{j}{N_c}\right) \vert j\rangle\langle j\vert \QFT^\dagger \\
            &= \QFT\sum_{j=-N_c/2}^{N_c/2 - 1} \frac{N_c}{T} \sin\left(2\pi \frac{j}{N_c}\right) \vert j\rangle\langle j\vert \QFT^\dagger
        \end{aligned}
        \end{align}
        where, in the second line, we define indices $-j = N_c - j$ for $j > 0$ and write the diagonalized $\Delta$ symmetrically about $j = 0$. The benefit of this parameterization is that small $\abs{j}$ correspond to low-frequency modes. Let $\Delta_J$ be $\Delta$ truncated at frequencies above those of the index $J \in [0,N_c/2] \cap \mathbb{Z}$.
        \begin{align} \label{eq:Delta_J}
           \Delta_J := \QFT\sum_{j=-J}^J \frac{N_c}{T} \sin\left(2\pi \frac{j}{N_c}\right) \vert j\rangle\langle j\vert \QFT^\dagger
        \end{align}
        The error in a clock space evolution using $\Delta_J$ rather than $\Delta$ is upper bounded by $T\norm{\Delta\ket{\phi_0} - \Delta_J\ket{\phi_0}}$, which can be evaluated and upper bounded as
        \begin{align} \label{eq:Delta_evolve_error2}
        \begin{aligned}
            T \norm{\Delta\ket{\phi_0} - \Delta_J\ket{\phi_0}} &= \bigg\|\sum_{\abs{j}> J} N_c \sin(2\pi \frac{j}{N_c}) \ketbra{j} \mathrm{QFT^\dagger \ket{\phi_0}}\bigg\| \\
           &\leq N_c \sqrt{\sum_{\abs{j}> J} \lvert\bra{j}\QFT^\dagger\ket{\phi_0}\rvert^2}.
        \end{aligned}
        \end{align}
        We thus desire a characterization of $\QFT^\dagger\ket{\phi_0}$, which we naturally expect to be another Gaussian up to errors arising from the difference between discrete and continuous Fourier Transforms. This analysis was performed in Appendix C of~\cite{Rendon2024Improved}, and we adapt that work to our present situation. As the reference shows, the error in each component $j$ arises from three sources: 
        \begin{enumerate}
            \item Truncation of the time variable to $O(T)$, which we denote $\epsilon_\mathrm{trunc}$.
            \item Truncation of the frequency variable to $O(N_c/T)$ (``aliasing"), which we denote $\epsilon_\mathrm{alias}$.
            \item Differences in normalizing in the continuum vs the discrete setting, which we denote $\epsilon_\mathrm{norm}$.
        \end{enumerate}
    
        In our notation and setting, Rendon et al.~\cite{Rendon2024Improved} show that these errors satisfy the following asymptotic bounds.
        \begin{align}
        \begin{aligned}
            \epsilon_\mathrm{trunc} &\in O\left(\sqrt{\frac{\sigma}{T}} e^{-\Omega(T^2/\sigma^2)}\right) \\
            \epsilon_\mathrm{alias} &\in O\left(\sqrt{\frac{\sigma}{T}} e^{-\Omega(N_c^2 \sigma^2/T^2}\right) \\
            \epsilon_\mathrm{norm} &\in O\left(e^{-\Omega(N_c)}\right)
        \end{aligned}
        \end{align}
        Let's take these errors to all be $O(\epsilon_{\QFT})$, with the required $\epsilon_\mathrm{QFT}$ to be determined. The results from Theorem 16 and Appendix C of~\cite{Rendon2024Improved} imply that
        \begin{equation}
            \QFT^\dagger \ket{\phi_0} =  \sum_{j= -N_c/2}^{N_c/2 - 1} \left(\sqrt{\frac{\pi N_c}{\mathcal{N}}} \frac{\sigma}{T}e^{-(\pi j \sigma/T)^2} + O (\epsilon_\mathrm{QFT})\right) \ket{j}.
        \end{equation}
        With in hand, we return to~\eqref{eq:Delta_evolve_error2}. First,
        \begin{equation}
            \lvert\bra{j}\QFT^\dagger \ket{\phi_0}\rvert^2 = \frac{\pi N_c}{\mathcal{N}} \frac{\sigma^2}{T^2} e^{-2(\pi j\sigma/T)^2} + O(\sqrt{\frac{\sigma}{T}}e^{-(\pi j \sigma/T)^2}\epsilon_\mathrm{QFT})
        \end{equation}
        where we assume the error $\epsilon_\mathrm{QFT}$ is smaller asymptotically than the amplitude itself, to be justified. Taking the sum over high frequencies,
        \begin{align}
        \begin{aligned}
            \sqrt{\sum_{\abs{j}>J} \lvert \bra{j}\QFT^\dagger \ket{\phi_0}\rvert^2} &\in O\left(\sqrt{\frac{N_c}{\mathcal{N}}}\frac{\sigma}{T} e^{-\Omega(J^2 \sigma^2/T^2)} + \sqrt{\frac{\sigma}{T}} e^{-\Omega(J^2 \sigma^2/T^2)} \epsilon_\mathrm{QFT}\right) \\
            &\subseteq O\left(\sqrt{\frac{\sigma}{T}} e^{-\Omega(J^2 \sigma^2/T^2)}(1 + \epsilon_\mathrm{QFT})\right).
        \end{aligned}
        \end{align}
        We next observe that $\epsilon_\mathrm{QFT} \in O(1)$ by previous assumptions, and can now be removed. From~\eqref{eq:Delta_evolve_error2}, we get the full simulation error by multiplying by $N_c$
        \begin{equation}
            \epsilon_J \in O\left(N_c \sqrt{\frac{\sigma}{T}}e^{-\Omega(J^2\sigma^2 /T^2)}\right).
        \end{equation}
        In order for $\epsilon_J \in O(\epsilon)$, we want the cutoff $J$ to satisfy
        \begin{equation}
            e^{-J^2 \sigma^2/T^2} \in O\left(\sqrt{\frac{T}{\sigma}}\frac{\epsilon}{N_c}\right)
        \end{equation}
        which can be satisfied provided $J$ scales as
        \begin{equation}
            J \in \Theta\left(\frac{T}{\sigma} \left(\log T/\sigma + \log N_c + \log1/\epsilon\right)\right) \subseteq \tilde{\Theta}(T/\sigma).
        \end{equation}
        Letting $\tilde{\Delta} \equiv \Delta_J$ for this choice of $J$, we now switch to considering the simulation of $\tilde{\Delta}$. Let $\delta' > 0$, and let $\Gamma(\delta') := \lceil (N_c/T\delta ') \sin(2\pi J/N_c)\rceil$. We have
        \begin{equation}
            \sum_{j=-J}^J \frac{N_c}{T} \sin\left(2\pi \frac{j}{N_c}\right)\ketbra{j} = \delta' \sum_{\ell=1}^{\Gamma(\delta')} S_k^{(\Delta)}(\delta') + O(\delta')
        \end{equation}
        where
        \begin{equation}
            S_k^{(\Delta)}(\delta') := \sum_{j=-J}^J \sgn(j) (-1)^{k[k> (N_c/T\delta')\sin(2\pi j/N_c)]}.
        \end{equation}
        Defining the unitary $V_\ell(\delta') := \QFT\,S_\ell^{(\Delta)}(\delta')\QFT^\dagger$, we have obtained an LCU decomposition of $\Delta$. The $\mathtt{PREP}$ circuit is, as with $C(H)$, only a column of Hadamards on 
        \begin{equation}
            n_\Delta \in O\left(\log \left((N_c/T\delta') \sin (2\pi J/N_c)\right)\right) \subseteq \tilde{O}\left(\log\frac{1}{\sigma\delta'}\right)
        \end{equation}
        auxiliary qubits. Meanwhile the $\mathtt{SEL}$ circuit  may be constructed as $\QFT\,\mathtt{SEL}' \,\QFT^\dagger$, where $\mathtt{SEL}'$ is a select circuit using the $S_\ell^{(\Delta)}$ signature matrices that can, as before, be implemented with comparator circuits that compute sine. 
        
        Combining with~\eqref{eq:CH_LCU}, we obtain an approximate LCU decomposition of the approximate clock Hamiltonian $\tilde{H}_c$.
        \begin{equation}
            \tilde{H}_c = \delta \sum_{i=1}^L\sum_{k=1}^{\Lambda_i(\delta)} U_i \otimes S_{ik}(\delta) + \delta' \sum_{\ell = 1}^{\Gamma(\delta')} \openone \otimes V_\ell(\delta') + O(\epsilon/T + L\delta + \delta')
        \end{equation}
        To achieve an $\epsilon$-accurate simulation, we will require $\delta \in O(\epsilon/LT)$ and $\delta' \in O(\epsilon/T)$. The $1$-norm $\norm{c}_1$ of all of the coefficients is given by
        \begin{align}
        \begin{aligned}
            \norm{c}_1 &= \delta \sum_{i=0}^{L-1} \Lambda_i(\delta) + \delta' \Gamma(\delta') \\
            &\in O\left(\sum_{i=0}^{L-1} \max_j \abs{\alpha_{ij}} + \frac{N_c}{T}\sin(2\pi J/N_c)\right) \\
            &\subseteq O\left(\norm{\alpha}_{\infty,1}^\mathrm{rev} + J/T\right) \\
            &\subseteq \tilde{O}\left(\norm{\alpha}_{\infty,1}^\mathrm{rev} + \sigma^{-1}\right) \\
            &\subseteq \tilde{O}\left(\norm{\alpha}_{\infty,1}^\mathrm{rev} + \frac{\max_t \|\dot{H}\| T}{\epsilon}\right)
        \end{aligned}
        \end{align}
        where $\norm{\alpha}_{\infty,1} \equiv \sum_{i=0}^{L-1} \max_t \abs{\alpha_i(t)}$ and $\tilde{O}$ suppresses multiplicative logarithmic factors. Thus, the number of queries to $\mathtt{SEL}$ and $\mathtt{PREP}$ circuits in an LCU encoding scales as
        \begin{equation} \label{eq:qubitization-complexity}
        \boxed{
            Q \in \tilde{O}\left(\norm{\alpha}_{\infty,1}^\mathrm{rev} T + \frac{\max_t \|\dot{H}\| T^2}{\epsilon} + \frac{\log 1/\epsilon}{\log\log1/\epsilon}\right)
        }\;.
        \end{equation}
        The number of auxiliary qubits needed for the clock register is 
        \begin{equation}
            n_c = \log N_p + \log N_q \in O\left(\log (\max_t \|\dot{H}\| T^2) + \log 1/\epsilon\right)
        \end{equation}
        while the number of auxiliary qubits needed for the LCU block encoding is given by
        \begin{align}
        \begin{aligned}
            n_\mathrm{LCU} &= n_{C(H)} + n_\Delta \\
            &\in O\left( \log \frac{\norm{\alpha}_{\infty,1}^\mathrm{rev}}{\delta} + \log \frac{1}{\sigma\delta'}\right) \\
            &\subseteq O\left(\log \frac{L \norm{\alpha}_{\infty,1}^\mathrm{rev} T}{\epsilon} + \log \frac{\max_t \|\dot{H}\| T^2}{\epsilon^2}\right) \\
            &\subseteq O\left(\log L + \log (\norm{\alpha}_{\infty,1}^\mathrm{rev} T) + \log (\max_t \|\dot{H}\|T^2) + \log 1/\epsilon\right)
        \end{aligned}
        \end{align}
        for a total number of auxiliary qubits $n \in O(n_\mathrm{LCU})$.

    \subsection{Discussion} \label{sec:Qubitization_Discussion}
        In this section, we've provided an algorithm for time dependent simulation by qubitization for instances when the Hamiltonian is input as a linear combination of unitaries. We provide a procedure for constructing an LCU-block encoding on the augmented clock space, and use the errors analysis of~\sec{Finite_Clock_Space} to provide a query complexity for the method. 

        The presence of the Trotter term in the complexity~\eqref{eq:qubitization-complexity} is unfortunate because, if it were absent, the query complexity would match lower bounds for simulation in $T$ and $\epsilon$. As a note of optimism, we believe this term is not due to the method itself but a fault of the analysis. Specifically, forcing our Hamiltonian to vary slowly over the $N_p$ larger subdivisions should prove unnecessary. This was done essentially to make the evolution consistent across the clock Gaussian state. In reality, the Hamiltonian should only need to vary smoothly over the smallest increment $\delta t$. We are currently investigating modifications to the clock scheme that would make this more apparent.
        
        Besides an LCU encoding, other natural block encodings of $H_c$ may be possible. For example, a very general input model for $H(t)$ is to take it as a $d$-sparse matrix with query access to the nonzero entries. This seems quite promising an avenue to take, because then $H_c = C(H) + \Delta$ is $d+2$ sparse, and there is a natural way to query the entries of $H_c$. Hence, such a Hamiltonian should immediately simulatable by qubitization (or other quantum walk methods). The trouble is that the largest entry in absolute value $\|H_c\|_\mathrm{max}$ of $H_c$ comes from $\Delta$, which is of size $N_c/2T$. This is too large to yield an effective simulation algorithm. Of course, there is something odd about the need to care for the operator norm $\|\Delta\|$, since the typical state being acted on is a Gaussian $\vert\phi_j\rangle$. Thinking of $\Delta$ in frequency space, modes of frequency $\Omega(\sigma^{-1})$ should not be relevant for Gaussian states of width $O(\sigma)$ on the clock register. This suggests that a high-frequency truncation of $\Delta$, say $\tilde{\Delta}$ would act approximately the same on the Gaussians while decreasing the norm. However, there is no guarantee that the modified operator, $\tilde{\Delta}$, is sparse in the basis of clock times. Perhaps considering a reduced clock Hamiltonian $\tilde{H}_{ij} =\langle\phi_i\vert H_c\vert\phi_j\rangle$, with all small elements set to zero, would have the sparseness conditions required, along with a subspace norm of $\norm{\Delta}_\phi \in O(\sigma^{-1})$.

\section{Time Dependent Simulation by Multiproduct Formulas} \label{sec:TDMPF}
    As suggested in~\sec{MultiproductFormulas}, MPFs have already been considered extensively in the Hamiltonian simulation community~\cite{Childs2012LCU,Faehrmann2022Randomizing,Zhuk2023DynamicMPF}. However, one of the deficiencies of MPFs is that they have yet to be generalized, formally, for use in time dependent Hamiltonian simulations. Because $U$ generally has time ordering, the techniques used in~\cite{Blanes1999Extrapolation} involving Baker-Campbell-Hausdorff-type expansions do not carry over directly. An approach based instead on the Magnus expansion might be expected to work in its place, but no subset of terms in the expansion represents the exact evolution separated from error terms. Without this generalization, MPFs cannot be applied to interaction picture algorithms as well as simulations of physical systems that have intrinsic time dependence. 
    
    It is relatively easy to propose a generalization of MPFs that would be expected to work well in the time dependent case, by Trotterizing the continuous clock Hamiltonian~\eqref{eq:cont-clock-Ham}. When this is done in $k_j$ steps, this amounts to replacing the $k_j$th power in~\thm{low} with a sequence of $k_j$ unitaries at each time slice. This heuristic argument motivates the following definition.
    \begin{definition}[Time Dependent Multiproduct Formulas] \label{defn:timedependentMPF}
        For finite dimensional $\mathcal{H}$ and $L:[0,T]^2\rightarrow L(\mathcal{H})$, let $L_p:[0,T]^2\rightarrow L(\mathcal{H})$ be a $p$th-order formula for $L$. Given $m \in \mathbb{Z}_+$, $\vec{k} \in \mathbb{Z}_+^m$, and  $a\in\mathbb{R}^m$, define the time dependent multiproduct formula $L_{m,p}:[0,T]^2\rightarrow L(\mathcal{H})$ to be
        \begin{equation*}
            L_{p,m}(t,t_0) := \sum_{j=1}^m a_j L_p^{(k_j)}(t,t_0)
        \end{equation*}
        where
        \begin{equation*}
            L_p^{(k)}(t,t_0) := \prod_{\ell = 0}^{k-1} L_p(t_{\ell+1}, t_{\ell})
        \end{equation*}
        and $t_{\ell} = t_0 + (t - t_0)\ell/k$.
    \end{definition}
    As a limiting case, observe that $L_{p,1} = L_p$ with $a_1 = 1$. The choice to take the $t_\ell$ as equally spaced is not entirely coincidental, for the same reason that, in the time independent setting, we take $U_2^{k}(t/k)$ instead of, say,
    \begin{align}
        \prod_{j=1}^k U_2(s_j t)
    \end{align}
    where $s = (s_1, \dots, s_k)$ is a probability vector. Taking a simple power of $k$ makes working with the BCH expansion much simpler. While these definitions could be applied in more general contexts, our interest in Hamiltonian simulation means we will consider $L = U$ to be a time evolution operator.

        We finally turn to the question of whether the time dependent MPFs of~\defn{timedependentMPF} may be constructed for improved approximants. At the beginning of this section, we mentioned the difficulty presented by time ordering in adopting the techniques from~\cite{Blanes1999Extrapolation}. The reader of the previous chapter may recognize that clock spaces may be used to remove time ordering, circumventing the issue. However, when the clock variable $t$ is continuous, the shift term $-E$ in the clock Hamiltonian is an unbounded operator, complicating a BCH-type analysis. We conjecture, and provide a heuristic argument, that time dependent MPFs indeed boost the approximation order for sufficiently smooth Hamiltonians.

        \begin{conjecture} \label{conj:TDMPF_existence}
            Let $H = \sum_{i = 1}^L H_i(t)$, and let 
            \begin{equation*}
                U_2(t + \tau, t) = \prod_{i = L}^1 e^{-i H_i(t + \tau/2)\tau}\prod_{i=1}^L e^{-i H_i(t + \tau/2)\tau}
            \end{equation*}
            be the symmetric, 2nd order Trotterized midpoint formula. Suppose each $H_i$ is $2m+1$ time differentiable. Then the time dependent multiproduct formula $U_{2,m}(t + \tau,t)$ with base formula $U_2$ approximates $U(t + \tau,t)$ to order $2m$ in $t$.
        \end{conjecture}
        We now discuss a potential path to proof of this conjecture. Without loss of generality, we take $t = 0$. Let $k\in \mathbb{Z}_+$, and consider a sequence of  discrete clock constructions on interval $[0,\tau]$, with parameters $(N_p(\ell), N_q(\ell), \sigma(\ell))$, such that $k$ always divides $N_c = N_p N_q$, and such that the limit reproduces the dynamics of $H(t)$ on the main register, as per~\thm{clock_limit}. Consider one of the elements of this sequence. Using the form of $H$ given in the conjecture statement, we may write
        \begin{equation}
            C(H) = \sum_{i=1}^L C(H_i).
        \end{equation}
        Thus, the clock Hamiltonian $H_c$ admits the following 2nd order symmetric Trotterization.
        \begin{equation}
            V_2(\tau) = e^{-i \Delta \tau/2} \left(\prod_{i=L}^1 e^{-i C(H_i) \tau/2} \prod_{i=1}^L e^{-i C(H_i) \tau/2}\right) e^{-i \Delta \tau/2}
        \end{equation}
        From~\cite{Blanes1999Extrapolation}, we have that
        \begin{equation}
            V(\tau) - V_2^k(\tau/k) = \sum_{j=1}^{m-1} \mathcal{E}_{2j+1}(\tau) \frac{\tau^{2j+1}}{k^{2j}} + \mathcal{E}(\tau,k)
        \end{equation}
        where $\mathcal{E} \in O(\tau^{2m+1})$ is analytic in $\tau$. Thus the standard, well-conditioned multiproduct formula $V_{2,m}$ of Theorem~\ref{thm:low} with base formula $V_2$ satisfies
        \begin{equation} \label{eq:MPF_clock_diff}
            V(\tau) - V_{2,m}(\tau) = \sum_{j=1}^m a_j \mathcal{E}(\tau,k_j).
        \end{equation}
        We now wish to look at the action on the main register. Applying equation~\eqref{eq:MPF_clock_diff} to the state $\ket{\psi}\ket{\phi_0}$ of the full register, where $\ket{\psi}$ is arbitary, and then taking the trace $\Tr_c$ over the clock register, one obtains
        \begin{equation} \label{eq:Tr_clock}
            \Tr_c(V(\tau) \ket{\psi}\ket{\phi_0}) - \Tr_c(V_{2,m}(\tau) \ket{\psi}\ket{\phi_0}) = \sum_{j=1}^m a_j E(\tau,k_j) (\ket{\psi})
        \end{equation}
        where $E(\tau,k)$ is a linear map on the main register defined by
        \begin{equation}
            E(\tau,k) (\ket{\psi}) := \Tr_c(\mathcal{E}(\tau,k)\ket{\psi}\ket{\phi_0}).
        \end{equation}
        The above holds for every clock space in the sequence defined by $(N_p(\ell), N_q(\ell), \sigma(\ell))$. Taking the limit as $\ell\rightarrow\infty$ of equation~\eqref{eq:Tr_clock} we may pass the limits through the finite sums and scalar multiplications  
        \begin{equation}
            \lim_{\ell\rightarrow\infty} \Tr_c(V(\tau) \ket{\psi}\ket{\phi_0}) - \lim_{\ell\rightarrow\infty} \Tr_c(V_{2,m}(\tau) \ket{\psi}\ket{\phi_0}) = \sum_{j=1}^m a_j \lim_{\ell\rightarrow\infty} E(\tau,k_j) \ket{\psi}
        \end{equation}
        provided that these limits exist. Indeed, by Theorem~\ref{thm:clock_limit}, 
        \begin{equation}
            \Tr_c(V(\tau) \ket{\psi}\ket{\phi_0}) = U(\tau,0)\ket{\psi}.
        \end{equation} 
        As for the MPF, taking $k$ steps of the Trotterization, we \emph{should} find that 
        \begin{equation}
            \lim_{\ell\rightarrow\infty} \Tr_c(V_2(\tau/k)^k \psi \ket{\phi_0}_c) = U_2^{(k)}(\tau,0)\ket{\psi}
        \end{equation}
        though this must be shown. This shouldn't be too hard, as the idea is clear: perform a sequene of clock shifts followed by 2nd order Trotter on the main register. By passing the limit through the multiproduct sum,
        \begin{equation}
            \lim_{\ell\rightarrow\infty} \Tr_c(V_{2,m}(\tau,0) \ket{\psi}\ket{\phi_0}_c) = U_{2,m}(\tau,0)\ket{\psi}.
        \end{equation}
        
        It remains to show that the limit $\lim_\ell E(\tau,k)$ exists, and moreover is in $O(\tau^{2m+1})$. This is where the main challenge lies. To show that the limit of a sequence with terms of order $O(\tau^{2m+1})$ is also $O(\tau^{2m+1})$, we can show that the $2m+1$ derivative is bounded at $\tau = 0$. Unfortunately, our current clock constructions have the width $\sigma$ of the clock state shrinking to infinity, which means the derivatives grow as well. If a different clock construction can be provided where the clock state can have width $\sigma \in O(1)$, a bound can be placed and thus the limit will be $O(\tau^{2m+1})$. 
    
        Current ongoing work is being undertaken to fill in the gaps of the previous argument. However, the numerics of~\sec{NumericalDemos} strongly suggest that the time dependent MPFs indeed work as expected. Moreover, the form of the time-dependent MPF of~\defn{timedependentMPF} can be obtained by a naive Trotterization of the continuous clock space, which is very suggestive that, beyond formal issues, the approach is reasonable. Thus, we proceed assuming Conjecture~\ref{conj:TDMPF_existence} is true.

    \subsection{Time Dependent MPF Simulation}\label{sec:pseudoMPF}
        Having argued, informally, that good time dependent MPFs exist, we now propose an algorithm for Hamiltonian simulation using these formulas. We will provide some accompanying discussion to explain our choices, and at the end we will more directly state the approach. 
        
        A natural input model for $H(t)$ is a linear combination of Hamiltonians
        \begin{equation}
            H(t) = \sum_{i=1}^L \alpha_i(t) H_i
        \end{equation}
        where each $\alpha_i(t) \in \mathbb{R}$ is assumed $2m+1$ differentiable for an $m$-term MPF. Without loss of generality we take $\norm{H_i} \leq 1$. Because we utilize the well-conditioning results of~\cite{Low2019Multiproduct}, we want the base formula to be 2nd order and symmetric. A reasonable choice is
        \begin{equation} \label{eq:Trotterized_midpoint_formula}
            U_2(t + \tau, t) := \prod_{i=L}^1 \exp\left\{-i H_i \alpha_i\left(t + \frac{\tau}{2}\right) \tau\right\}\prod_{i=1}^L \exp\left\{-i H_i \alpha_i\left(t + \frac{\tau}{2}\right) \tau\right\},
        \end{equation}
        which is a 2nd-order Trotter splitting of the midpoint formula. Thus, from now on we will be interested in the MPF
        \begin{equation}
            U_{2,m}(t,0) = \sum_{j=1}^m a_j U_2^{(k_j)}(t,0).
        \end{equation}
        As a caution, we remark that, despite notation, the MPF $U_{2,m}$ is not generally unitary for $m > 1$, though when suitably constructed it will approximate the unitary $U$, hence be approximately unitary.
        
        That $U_2$ is second-order can be seen from Taylor expanding the Dyson series of $U$ about $\tau = 0$ ($H$ needs to be at least, say, twice differentiable). Moreover, $U_2$ is time-reversal symmetric in the same sense as $U$: $U_2(t,t_0) = U_2(t_0,t)^\dagger$. This gives the nice property that the error series for $U(t+\tau, t) - U^{(k)}(t + \tau, t)$ has only even terms, such that higher order formulas can be reached with approximately half the number of addends. 
        
        From the onset, there are a couple of choices to make. The MPFs, in principle, could approximate the entire interval $[0,T]$ provided that the Trotter steps $k_i$ are sufficiently large. However, this has several disadvantages. First, there is no flexibility to treat some subintervals of $[0,T]$ as more difficult than others and allocate resources appropriately. Second, the well-conditioned scheme of~\cite{Low2019Multiproduct} would have to be abandoned or modified to accommodate larger $\vec{k}$. Instead, we divide $[0, T]$ into a mesh of $r$ subintervals, not necessarily uniform, but rather constructed to account for more difficult parts of the simulation. We provide a greedy algorithm for constructing such a mesh at the end of this chapter. The algorithm requires a computable $\Lambda_{2m+1}$-bound to work (see~\defn{LambdaBound}), however, a practitioner might prefer a more heuristic approach to constructing the time mesh. For the moment, we will simply say that, given $t_i$, the next time point $t_{i+1}$ is incremented roughly as $1/\Lambda_{2m+1}(t)$ for $t$ in a neighborhood of $t_i$, where $\Lambda_{2m+1}$ is a positive real-valued function of $H$ and its derivatives that grows for larger or faster fluctuating $H$.

        Once the mesh points $t_0, t_1, \dots, t_r$ are determined, a time dependent MPF is performed over each subinterval $[t_i, t_{i+1}]$ in sequence. We assume the MPF is implemented using the LCU technique. The base midpoint formula $U_2$ must be implemented by some scheme which depends on the structure of $H(t)$, though the approximating unitary $\tilde{U}_2$ should be at least 2nd-order and preserve the time-reversal symmetry of $U_2$ (and $U$). We take~\eqref{eq:Trotterized_midpoint_formula} as our base formula for the subsequent analysis. It is known that such Trotter formulas exhibit commutator scaling, meaning that, in the limit where all $H_j$ commute pairwise and all $\alpha_j$ are constant functions, the simulation error goes to zero. Hence, the MPF will also inherit this desirable property.

        Let us now supply our pseudo-algorithm for the MPF procedure. Given fundamental parameters, $[0,T]$, $\epsilon$, and a description of $H(t)$:
        \begin{enumerate}
            \item Compute a $\Lambda_{2m+1}$ bound (\defn{LambdaBound}) for some $M$ larger than the expected number of MPF terms. This is more a less a bound on the ``difficulty" of $H(t)$ at various times.
            \item Construct a time mesh of $r$ steps using the algorithm of~\app{adapt_deltat}.
            \item Perform a sequence of MPFs over each time slice, with 2nd order base formula $W_2$ approximating the midpoint formula.
        \end{enumerate}
        Specific information about the parameter choices, such as $m$ and $r$, is provided in the subsequent error analysis, though sometimes only in a big-$O$ sense.

    \subsection{Error Analysis} \label{sec:ErrorAnalysis}
        In this section, we analyse the errors arising between the exact unitary $U$ and the MPF approximation $\tilde{U}$ given by
        \begin{equation} \label{eq:Utilde}
            \tilde{U}(T, 0) = \prod_{i=1}^r U_{2,m}(t_i,t_{i-1}).
        \end{equation}
        This analysis will ignore hardware imperfections and decoherence, assume that $U_2$ is implemented perfectly, and assume exact coefficients $a_j$. In the query complexity analysis of~\sec{QueryComplexity} we will consider additional algorithmic errors arising from a more precise specification of the Hamiltonian input model.
    
        We introduce a useful definition to quantify errors succinctly. It is well understood that MPFs, like regular product formulas, have smoothness requirements to ensure convergence. To quantify errors and costs of MPFs, we provide a metric which captures the ``size" of $H$ and its derivatives at each point in time, in order to characterize the difficulty of simulation.
        \begin{definition}\label{defn:LambdaBound}
            Let $H(t) = \sum_{i=1}^L \alpha_i(t) H_i$ be a time dependent, finite-dimensional Hamiltonian with $H_i$ Hermitian and $\alpha_i(t) \in \mathbb{R}$ having $n \in \mathbb{N}\cup\{\infty\}$ continuous derivatives. For each $i$ define a $\Lambda_{i,n}$-bound ("Lambda  i n bound") as any continuous function $\Lambda_{i,n}: [0,T] \to \mathbb{R}_+$ satisfying the following bounds with respect to $H$ and its derivatives
            \begin{equation*}\label{eq:LambdaBound}
                \Lambda_{i,n} (t) \geq \sup_{j\in [n]} \sqrt[j+1]{\|{\alpha_i^{(j)}(t)}\|}\quad \forall t\in [0,T] 
            \end{equation*}
            where $f^{(n)}$ represents an $n$th derivative of $f$, and $[n] :=\{j\in\mathbb{N}\mid j\leq n\}$. Assuming such bounds exist for all $i = 1,\dots, L$, we say that $H(t)$ is $\Lambda_n$-bounded. We further say that $H(t)$ is $\Lambda_n$-boundable if it admits some $\Lambda_n$-bound. For convenience, we define $\Lambda_i \equiv \Lambda_{i,\infty}$. We also define a $\Lambda_n$ bound as any continuous on $[0,T]$ satisfying
            \begin{equation*}
                \Lambda_n (t) \geq \max_{i\in [L]} \Lambda_{i,n}(t).
            \end{equation*}
        \end{definition}
        For near-constant $\alpha_i(t)$, $\Lambda_{i,n}$ is simply an upper bound on $\abs{\alpha_i}$, while for rapid oscillations the derivative terms will dominate. Observe that for finite $n$, our assumptions imply that $\Lambda_{i,n}(t)$ exists ($H$ is $\Lambda_{i,n}$-boundable), since $\lvert{\alpha_i^{(j)}}\rvert$ is continuous on a compact interval and hence a bounded function. Also in the finite case, the supremum may be replaced with a simple $\max$, and $\Lambda_{i,n}(t)$ may be taken as equal to the right hand side because it is the maximum of a finite set of continuous functions, which is continuous. For this ``minimal choice," $\Lambda_{i,n}(t)$ is a nondecreasing sequence in $n$. For each $n$, there also exists a $\Lambda_{i,n}$ that is constant in $t$. Allowing $\Lambda_{i,n}$ to vary in time, however, takes into consideration the possibility that the expense of simulating $H$ will vary with time.  We note that $\Lambda_{i,n}$-bounds are additive in the sense that, for $H(t)$ and $G(t)$ admitting $\Lambda^{H}_{i,n}$ and $\Lambda^{G}_{i,n}$-bounds, respectively, $\Lambda_{i,n}^{H} + \Lambda_{i,n}^{G}$ is a $\Lambda_{i,n}$-bound on $H + G$.
        
        In contrast to finite $n$, the existence of a $\Lambda_{i,\infty}$-bound is not guaranteed, and amounts to the assumption that the derivatives of $H$ grow at most exponentially for asymptotically large $j$ and fixed $t$. There are smooth, even analytic functions which do not satisfy this, many of which are physically interesting. A simple example is a Gaussian pulse
        \begin{equation}
            \alpha(t) = e^{-t^2}
        \end{equation}
        whose derivatives, generating the Hermite polynomials, grow factorially with $n$ at $t = 0$. Other interesting cases, such as harmonic oscillations or exponential growth and decay, do admit a $\Lambda$-bound. Despite these restrictions, we adopt this approach for simplicity and in order to facilitate comparison with prior work on general-order Suzuki-Trotter formulas~\cite{Wiebe2010Higher}. Admittedly, a modification of~\defn{LambdaBound} to be an upper bound on
        \begin{equation}
            \max_j j^{-1} \sqrt[j+1]{\|\alpha_i^{(j)}(t)\|}
        \end{equation}
        would expand the class of functions admitting $\Lambda_\infty$-bounds to analytic functions (though not generic smooth functions).
        
        We now begin the error analysis of~\eqref{eq:Utilde} in earnest. From a triangle inequality the error can be bounded as the error in each step.
        \begin{equation} \label{eq:error_sum}
            \|U(T, 0)-\tilde{U}(T, 0)\| \leq \sum_{i=1}^r \|U(t_i, t_{i-1})-U_{2,m}(t_i,t_{i-1})\|
        \end{equation}
        Therefore, to ensure an error at most $\epsilon$, it suffices that each subinterval has error at most $\epsilon/r$. We thus focus a single subinterval. An upper bound on this error is supplied by the following theorem, which the main technical result of this section.         
        \begin{theorem} \label{thm:errorMain}
            Let $H: [t_0,t_1]\rightarrow \mathrm{Herm}(\mathcal{H})$ be a time dependent Hamiltonian on finite-dimensional $\mathcal{H}$ with $2m+1$ continuous derivatives on $[t_0,t_1]$ and $\Lambda_{2m+1}$-bound. Suppose further that 
            \begin{equation*}
                e L \max_{\tau\in[t_0,t_1]}\Lambda_{2m+1}(\tau) (t_1 - t_0) < 1.
            \end{equation*}
            Then for any $m \in \mathbb{Z}_+$ there exists $\vec{k}\in \mathbb{Z}_+^{m}$ and $a\in \mathbb{R}^m$ such that
            \begin{equation*}
                \|U(t_1,t_0)-U_{2,m}(t_1,t_0)\| < \frac{\|a\|_1}{\sqrt{\pi m}}\left(5 L \max_{\tau\in[t_0,t_1]}\Lambda_{2m+1}(\tau) (t_1-t_0)\right)^{2m+1}      
            \end{equation*}
            and $\|a\|_1 \in O(\log(m))$. 
        \end{theorem}
        Observe that convergence of the above error bound to zero as $m \rightarrow\infty$ is conditioned on sufficiently small $t_1 - t_0$. This is potentially unsurprising, as the Suzuki-Trotter formulas also do not provide an unconditionally converging sequence of approximations to the time evolution operator. Note as well the parallel roles between $m$ and the Suzuki-Trotter order $k$ in reducing the error. In our case, however, we shall see that the simulation cost increases only polynomially in $m$, whereas for product formulas the cost is necessarily exponential in $k$. 
        
        The term $\norm{a}_1/\sqrt{\pi m}$ is $o(1)$ for large $m$ and can be more or less ignored. Unfortunately, the $\Lambda_{2m+1}$ scales as the ``worst" coefficient $\alpha_i$ multiplied by the number of terms $L$, which seems too cynical. However, improving on this may greatly complicate the proof of the error bound.~\thm{errorMain} will be the important result that informs the algorithmic choices and complexity analysis of subsequent sections. Having characterized the error on a single subinterval of $[0,T]$, the full error over $r$ subintervals may be found simply using~\eqref{eq:error_sum}.
        
        We prove~\thm{errorMain} using a similar strategy to that used to provide error estimates for the Suzuki-Trotter formulas~\cite{Berry2007Sparse,Wiebe2010Higher,Childs2019Randomization}. As $H$ is continuously differentiable at least $2m+1$ times, $U_{2,m}$ is a valid extrapolant by Conjecture~\ref{conj:TDMPF_existence}, and cancels the first $m$ terms in the error series. We can thus express the difference $U_{2,m} - U$ using the integral Taylor remainder formulas
        \begin{equation} 
            U_{2,m}(t,t_0) - U(t,t_0)  =  R_{2m} - \mathcal{R}_{2m}
        \end{equation}
        with
        \begin{align} 
            \mathcal{R}_{2m} &:= \frac{1}{2m!} \int_{t_0}^t (t-\tau)^{2m} U^{(2m+1)}(\tau, t_0) d\tau \label{eq:taylor_remainder_R}\\
            R_{2m} &:= \frac{1}{2m!} \int_{t_0}^t (t-\tau)^{2m} U_{2,m}^{(2m+1)}(\tau, t_0) d\tau,\label{eq:taylor_remainder_Rc}
        \end{align}
        where $U^{(n)}$ refers to derivatives in the first argument. By the triangle inequality,
        \begin{equation} \label{eq:triangleR}
            \|U_{2,m}(t, t_0) -U(t, t_0)\| \leq \|\mathcal{R}_{2m}\| + \|R_{2m}\|
        \end{equation}
        and we upper bound each remainder in separate lemmas. 

        The easier bound is $\mathcal{R}_{2m}$, so we begin with the corresponding lemma.
        \begin{lemma} \label{lem:Rc2m}
            The remainder term $\mathcal{R}_{2m}$ in equation~\eqref{eq:taylor_remainder_Rc} satisfies
            \begin{equation*}
                \|\mathcal{R}_{2m}\| < \frac{1}{2\sqrt{\pi m}}\left(2 L \max_{\tau\in [t_0,t]}\Lambda_{2m+1}(\tau) (t-t_0)\right)^{2m+1}.
            \end{equation*}
        \end{lemma}
        \begin{proof}        
            Recall that $U$, as the exact propagator, satisfies the Schrödinger equation \eqref{eq:Schrödinger_eqtn}. Higher derivatives can easily be found through repeated application of the product rule. The result will be a polynomial in the derivatives of $H$ times $U$ itself. Under the spectral norm, using the triangle and submultiplicative properties, the ordering of terms doesn't matter, and therefore equivalent to the expression one gets taking derivatives of a scalar exponential. Noting that $\|U\| = 1$, the resulting polynomial is the complete exponential Bell polynomial from Faà di Bruno's formula (see~\app{milieu}). Letting $n = 2m+1$, we have
            \begin{equation}
                \|\partial_{t}^n U(t,t_0)\| \leq Y_n\left(\|H(t)\|, \|\dot{H}(t)\|, \dots, \|H^{(n-1)}(t)\|\right).
            \end{equation}
            From the definition of $\Lambda_{i,n}$, we have
            \begin{align}
            \begin{aligned}
                \|H^{(j)}(t)\| &\leq \sum_{i=1}^L \lvert \alpha_i^{(j)}(t)\rvert \\
                &\leq \sum_i \Lambda_{i,n}(t)^{j+1}\\
                &\leq (L\Lambda_n(t))^{j+1}
            \end{aligned}
            \end{align}
            and since the Bell polynomials $Y_n$ are monotonic in each argument,
            \begin{align}
            \begin{aligned}
                Y_{n}\left(\|H\|, \|\dot H\|, \dots, \|H^{(n-1)}\|\right) &\leq Y_n(L\Lambda_n(t), (L\Lambda_n(t))^2, \dots , (L\Lambda_n(t))^n) \\
                &= (L\Lambda_n(t))^n b_n 
            \end{aligned}
            \end{align}
            where $b_n$ are the Bell numbers (\app{milieu}). Thus, 
            \begin{equation} \label{eq:Bell_num_bound}
                \|\partial_{t}^n U(t,t_0)\| \leq (L\Lambda_n(t))^n b_n.
            \end{equation}
            Finally, returning to  the bound on $\mathcal{R}_{2m}$, we have from the integral triangle inequality that
            \begin{align}
            \begin{aligned}
                \|\mathcal{R}_{2m}\| &\leq \frac{1}{(2m)!} \int_{t_0}^{t} (t-\tau)^{2m} \|\partial_\tau^{2m+1} U(\tau, t_0)\| d\tau\\
                &\leq \frac{1}{(2m)!} \int_{t_0}^t (t-\tau)^{2m} (L\Lambda_{2m+1}(\tau))^{2m+1} b_{2m+1} d\tau
            \end{aligned}
            \end{align}
            where we made use of equation~\eqref{eq:Bell_num_bound}. This, in turn, can be bounded by maximizing $\Lambda_{2m+1}$ over $[t_0,t]$.
            \begin{align} \label{eq:Bell_over_factorial}
            \begin{aligned}
                \|\mathcal{R}_{2m}\| &\leq \frac{b_{2m+1}}{(2m)!}(L\max_{\tau\in [t_0,t]}\Lambda_{2m+1}(\tau))^{2m+1} \int_{t_0}^{t} d\tau (t-\tau)^{2m} \\
                &\leq \frac{b_{2m+1}}{(2m+1)!} \left(L\max_{\tau\in [t_0,t]}\Lambda_{2m+1}(\tau) (t-t_0)\right)^{2m+1}
            \end{aligned}
            \end{align}
            Finally, we upper bound the prefactor using a Stirling bound and bounds from~\cite{Berend2010Improved} on the bell numbers. For all $m\in \mathbb{Z}_+$,
            \begin{align}
            \begin{aligned}
                \frac{b_{2m+1}}{(2m+1)!} &< \frac{\left(\frac{0.792(2m+1)}{\log(2m+2)}\right)^{2m+1}}{\sqrt{2\pi(2m+1)}\left(\frac{2m+1}{e}\right)^{2m+1}} \\
                &= \frac{1}{\sqrt{2\pi (2m+1)}} \left(\frac{.792 e}{\log(2m+2)}\right)^{2m+1}.
            \end{aligned}
            \end{align}
            Plugging this into equation~\eqref{eq:Bell_over_factorial},
            \begin{align}
            \begin{aligned}
                \|\mathcal{R}_{2m}\| &< \frac{1}{\sqrt{2\pi(2m+1)}}\left(\frac{0.792 e}{\log(2m+2)}L\max_{\tau\in [t_0,t]}\Lambda_{2m+1}(\tau) (t-t_0)\right)^{2m+1} \\
                &< \frac{1}{2\sqrt{\pi m}}\left(2L\max_{\tau\in [t_0,t]}\Lambda_{2m+1}(\tau) (t-t_0)\right)^{2m+1}.
            \end{aligned}
            \end{align}
            The last line is the result of the lemma.
        \end{proof}
        We now state the bound on the Taylor $R_{2m}$ for the time dependent MPF. 

        \begin{lemma}\label{lem:R2m}
            In the notation above, suppose that 
            \begin{equation*}
                e L \max_{\tau\in[t_0,t_1]}\Lambda_{2m+1}(\tau) (t_1 - t_0) < 1.
            \end{equation*}
            Then the remainder term $R_{2m}$ in equation~\eqref{eq:taylor_remainder_Rc} satisfies
            \begin{equation*}
                \|R_{2m}\| < \frac{\|a\|_1}{2\sqrt{\pi m}}\left(5 L \max_{\tau\in [t_0,t] }\Lambda_{2m+1}(\tau) (t-t_0)\right)^{2m+1} \;.
            \end{equation*}
        \end{lemma}
        The proof is more technical than the previous bound, and is given at the end of this section. First, we quickly prove~\thm{errorMain} assuming the truth of the above Taylor remainder lemmas.
        \begin{proof}[Proof of~\thm{errorMain}]
            First, we note that $\|a\|_1 \geq 1$, since $a$ necessarily satisfies $\sum_j a_j = 1$ from the Vandermonde constraints~\eqref{eq:vandermonde}. From equation~\eqref{eq:triangleR}, the error $\|U(t,t_0)-U_{2,m}(t,t_0)\|$ is bounded by the sum of the remainder upper bounds derived in Lemmas~\ref{lem:R2m} and~\ref{lem:Rc2m}. Comparing the two, we see that  $R_{2m}$ dominates $\mathcal{R}_{2m}$ for all $m \geq 1$. We thus take twice the larger as an upper bound
            \begin{align}
            \begin{aligned}
                \|U(t,t_0)-U_{2,m}(t,t_0)\| &< 2 \|R_{2m}\| \\
                &< \frac{\|a\|_1}{\sqrt{\pi m}}\left(5 L\max_{\tau\in[t_0,t]}\Lambda_{2m+1}(\tau) (t-t_0)\right)^{2m+1} \;.
            \end{aligned}
            \end{align}
            This completes the proof.
        \end{proof}
        
        To prove~\lem{R2m}, we will first need a technical lemma that bounds the size of ordinary exponentials of time dependent matrices. 
        \begin{lemma} \label{lem:FaadiBrunoOperator}
            Let $A(t)$ be an anti-Hermitian valued function of $t \in \mathbb{R}$ with $n$ bounded derivatives. Then
            \begin{equation*}
                \norm{d_t^n e^{A(t)}} \leq Y_n\left(\norm{d_t A(t)}, \norm{d_t^2 A(t)}, \dots, \norm{d_t^n A(t)}\right) 
            \end{equation*}
            where $Y_n$ is the complete exponential Bell polynomial.
        \end{lemma}
        In the scalar case, Faà di Bruno's bound is an exact expression (see~\app{milieu}), so the content of our result is that a corresponding bound holds even in the non-scalar case. Th exponential disappears because $e^{A(t)}$ is unitary. The proof of this is provided in~\app{MPFproofs}. 
        
        We finally conclude this section with a proof of the bound on $R_{2m}$.
        \begin{proof}[Proof of~\lem{R2m}]    
            Without loss of generality, we take $t_0 = 0$. The relevant expressions are
            \begin{align}
                U_{2,m}(t,0) = \sum_{j=1}^m a_j U_2^{(k_j)}(t,0)
            \end{align}
            and
            \begin{align}
                U_2^{(k)}(t,0) := \prod_{\ell = 1}^{k} U_2(t_\ell, t_{\ell-1})
            \end{align}
            with $t_{\ell} := t \ell/k$. The Taylor remainder in integral form is given by 
            \begin{align}
            \begin{aligned}
                R_{2m} &= \frac{1}{(2m)!} \int_0^t (t-\tau)^{2m} \frac{d^{2m+1}}{d\tau^{2m+1}} U_{2,m}(\tau,0) d\tau \\
                &= \frac{1}{(2m)!}\sum_{j=1}^m a_j \int_0^t (t-\tau)^{2m} \frac{d^{2m+1}}{d\tau^{2m+1}} U_2^{(k_j)}(\tau,0) d\tau.
            \end{aligned}
            \end{align}
            With a couple triangle inequalities, this is upper bounded as
            \begin{align} \label{eq:Holderbound_R2m}
            \begin{aligned}
                \norm{R_{2m}} &\leq \frac{1}{(2m)!} \sum_{j=1}^m \abs{a_j} \frac{t^{2m+1}}{2m+1} \max_{\tau \in [0,t]} \norm{d_\tau^{2m+1} U_2^{(k_j)}(\tau,0)} \\
                &\leq \frac{\norm{a}_1}{(2m+1)!}  t^{2m+1}\max_{j,\tau} \norm{d_\tau^{2m+1} U_2^{(k_j)}(\tau,0)}
            \end{aligned}
            \end{align}
            where in the last line we employed a Hölder inequality. Our focus is now on bounding the derivative, which we unravel layer by layer using frequent multinomial expansions. First,
            \begin{align} \label{eq:V2_seq_multinom}
                d_\tau^n U_2^{(k)}(\tau,0) = \sum_{N} \binom{n}{n_1,\dots, n_k} \prod_{\ell = 1}^{k} d_\tau^{n_\ell} U_2(\tau_\ell, \tau_{\ell-1}).
            \end{align}
            Next, we write
            \begin{align}
            \begin{aligned}
                U_2(\tau_\ell, \tau_{\ell-1}) &= \prod_{i=L}^1 e^{-i H_i \alpha_i(\tau_{\ell-1/2}) \tau/k}\prod_{i=1}^L e^{-i H_i \alpha_i(\tau_{\ell-1/2}) \tau/k} \\
                &= \prod_{i=1}^{2L} e^{A_{i,\ell}}
            \end{aligned}
            \end{align}
            where 
            \begin{align}
                A_{i,\ell} := -i H_i \alpha_i(\tau_{\ell-1/2})\tau/k
            \end{align}
            and $i$ is defined by reflection for $i > L$. Once again performing a multinomial expansion,
            \begin{align} \label{eq:2step_deriv}
                d_\tau^n U_2(\tau_\ell, \tau_{\ell-1}) = \sum_N \binom{n}{n_1,\dots, n_{2L}} \prod_{i=1}^{2L} d_\tau^{n_i} e^{A_{i,\ell}}.
            \end{align}
            We now bound the individual ordinary operator exponentials. Invoking~\lem{FaadiBrunoOperator},
            \begin{align} \label{eq:FaadiBruno_Bound}
                \norm{d_\tau^n e^{A_{i,\ell}}} \leq Y_n\left(\norm{d_\tau A_{i,\ell}}, \dots, \norm{d_\tau^n A_{i,\ell}}\right).
            \end{align}
            In turn, we have
            \begin{align}
            \begin{aligned}
                d_\tau^n A_{i,\ell} &= -i \frac{H_i}{k} d_\tau^n(\alpha_i(\tau_{\ell-1/2})\tau) \\
                &= -i \frac{H_i}{k} \left[\left(\frac{\ell-1/2}{k}\right)^n \tau \alpha_i^{(n)}(\tau_{\ell-1/2}) + n \left(\frac{\ell-1/2}{k}\right)^{n-1} \alpha_i^{(n-1)}(\tau_{\ell-1/2})\right] \\
            \end{aligned}
            \end{align}
            where $\alpha^{(n)}(x)$ refers to the $n$th derivative of $\alpha$ with respect to its argument, then evaluated at $x$ (i.e., not a $\tau$ derivative). Since $\norm{H_i}\leq 1$ we have
            \begin{align}
                \norm{d_\tau^n A_{i,\ell}} <  \frac{1}{k} (\ell/k)^{n-1} \left((\ell/k)\tau\lvert{\alpha_i^{(n)}(\tau_{\ell-1/2})}\rvert + n\lvert{\alpha_i^{(n-1)}(\tau_{\ell-1/2})}\rvert\right).
            \end{align}
            From~\defn{LambdaBound}, $\alpha_i^{(j)}(t) \leq \Lambda_{i,n}(t)^{j+1}$. Dropping the $n$ and $t$ dependence for the moment,
            \begin{align}
            \begin{aligned}
                \norm{d_\tau^n A_{i,\ell}} &< (\ell/k)^n\left((\tau/k) \Lambda_i^{n+1} + (n/\ell)\Lambda_i^n\right) \\
                &=(\Lambda_i\ell/k)^n\left(\Lambda_i\tau/k + n/\ell\right).
            \end{aligned}
            \end{align}
            We've reached the bottom, and now proceed to work our way back up to the Taylor remainder $R_{2m}$, starting with~\eqref{eq:FaadiBruno_Bound}. Using the equation~\eqref{eq:Faa_di_Bruno} of~\app{milieu},
            \begin{align}
            \begin{aligned}
                \norm{d_\tau^n e^{A_{i,\ell}}} &\leq \sum_C \frac{n!}{c_1! c_2!\dots c_n!} \prod_{j=1}^n \left(\frac{\norm{d_\tau^j A_{i,\ell}}}{j!}\right)^{c_j} \\
                &< \sum_C \frac{n!}{c_1! c_2!\dots c_n!} \prod_{j=1}^n \left(\frac{(\Lambda_i\ell/k)^j\left(\Lambda_i\tau/k + j/\ell\right)}{j!}\right)^{c_j}.
            \end{aligned}
            \end{align}
            Using the sum rule for $C$ we can pull out a factor of $(\Lambda_i \ell/k)$. Using the upper bound $j \leq  n$ and the monotonicity of $Y_n$, we obtain the bound
            \begin{align}
                \norm{d_\tau^n e^{A_{i,\ell}}} &< (\Lambda_i \ell/k)^n B_n(\Lambda_i \tau/k + n/\ell)
            \end{align}
            where $B_n$ is the Bell polynomial (see~\app{milieu}). For simplicity, define 
            \begin{equation}
                x_{i,\ell, n} = \Lambda_i\tau/k + n/\ell
            \end{equation}
             as the argument to $B_n$. Employing the bound~\eqref{eq:Bellfunc_bound}, 
            \begin{align}
                \norm{d_\tau^n e^{A_{i,\ell}}} &< (\Lambda_i \ell/k)^n \left(\frac{n}{\log(1 + n/x_{i,\ell,n})}\right)^n.
            \end{align}
            which is valid for all $n>0,$ and for $n = 0$ when defined by the $0^+$ limit. We can simplify the reciprocal $\log$ with the bound
            \begin{align}
            \begin{aligned}
                \frac{1}{\log(1 + n/x_{i,\ell,n})} &< \left(\frac{1}{2} + \frac{x_{i,\ell,n}}{n}\right)^n \\
                &= \frac{1}{2^n} \left(1 + \frac{2 x_{i,\ell,n}}{n}\right)^n.
            \end{aligned}
            \end{align}
            This gives us the simplified exponential derivative
            \begin{align}
                \norm{d_\tau^n e^{A_{i,\ell}}} &< \left(\Lambda_i \ell/ 2k\right)^n(n+2x_{i,\ell,n})^n.
            \end{align}
            
            We now move up a level to reconsider~\eqref{eq:2step_deriv}. Employing a triangle inequality,
            \begin{align}
            \begin{aligned}
                \norm{d_\tau^n U_2(\tau_\ell, \tau_{\ell-1})} &\leq \sum_N \binom{n}{n_1,\dots, n_{2L}} \prod_{i=1}^{2L} \norm{d_\tau^{n_i} e^{A_{i,\ell}}} \\
                &< \sum_N \binom{n}{n_1,\dots, n_{2L}} \prod_{i=1}^{2L} \left(\Lambda_i \ell/ 2k\right)^{n_i}(n_i+2x_{i,\ell, n_i})^{n_i}.
            \end{aligned}
            \end{align}
            Maximize $\Lambda_i$ over all $i = 1,\dots, L$ and call it $\Lambda$. We can factor out the corresponding term, and with some rewriting obtain
            \begin{align}
                (\Lambda\ell/2k)^n \sum_N \binom{n}{n_1,\dots, n_{2L}}\prod_{i=1}^{2L} (n_i + 2x_{\ell,n_i})^{n_i}.
            \end{align}
            where we've also let $x_{\ell,n_i}$ be $x_{i,\ell,n_i}$ with the subscript dropped on $\Lambda_i$.Focusing on the rightmost product over $i$, one can show using a Lagrange multiplier that the maximum is given by $n_i = n/2L$ for all $i$ (we maximize over $n_i \in \mathbb{R}_+$, which is an upper bound). This is intuitive from symmetry of the product as well. Taking this as an upper bound, we have
            \begin{align}
            \begin{aligned}
                \norm{d_\tau^n U_2(\tau_\ell, \tau_{\ell-1})} &< (\Lambda\ell/2k)^n \left(\frac{n}{2L} + \frac{2 \Lambda\tau}{k} + \frac{n}{L\ell}\right)^n \sum_{N}\binom{n}{n_1,\dots,n_{2L}} \\
                &= (\Lambda\ell/2k)^n \left(n + \frac{4 \Lambda\tau L}{k} + \frac{2n}{\ell}\right)^n  \\
                &= (\Lambda/k)^n \left(n + n\ell/2 + \frac{2\Lambda\tau L \ell}{k}\right)^n.
            \end{aligned}
            \end{align}
            where in going to the second line we evaluated the multinomial sum as $(2L)^n$ and simpified.
            
            With this in hand, we return to~\eqref{eq:V2_seq_multinom} and bound it as
            \begin{align}
            \begin{aligned}
                \norm{d_\tau^n U_2^{(k)}(\tau, 0)} &\leq \sum_N \binom{n}{n_1,\dots, n_k} \prod_{\ell=1}^k \norm{d_\tau^{n_\ell} U_2(\tau_\ell, \tau_{\ell-1})} \\
                &< \sum_N \binom{n}{n_1,\dots,n_k}\prod_{\ell=1}^k (\Lambda/k)^{n_\ell} \left(n_\ell + n_\ell \ell/2 + \frac{2\Lambda \tau L \ell}{k}\right)^{n_\ell}.
            \end{aligned}
            \end{align}
            Using the upper bound $\ell \leq k$ and factoring out the $(\Lambda/k)^{n_\ell}$ using the sum rule,
            \begin{equation}
                \norm{d_\tau^n U_2^{(k)}(\tau, 0)} < (\Lambda/k)^n \sum_N \binom{n}{n_1,\dots,n_k} \prod_{\ell=1}^k (n_\ell + n_\ell k/2 + 2 \Lambda \tau L)^{n_\ell}.
            \end{equation}
            Similar to, we upper bound the product using $n_\ell = n/k$ for all $\ell$, which can be justified through a maximization using Lagrange multipliers. The corresponding bound is
            \begin{align}
            \begin{aligned}
                \norm{d_\tau^n U_2^{(k)}(\tau, 0)} &<(\Lambda/k)^n(n/k + n/2 + 2 \Lambda \tau L)^n \sum_N \binom{n}{n_1,\dots,n_k} \\
                &= (n\Lambda)^n \left(\frac{1}{k} + \frac{1}{2} + \frac{2\Lambda \tau L}{n}\right)^n.
            \end{aligned}
            \end{align}
            
            We are finally ready to return to equation~\eqref{eq:Holderbound_R2m} and bound $R_{2m}$.  We recall that $\Lambda$ has $\tau$ dependence, and let $\Lambda_\mathrm{max} := \max_{\tau\in[0,t]} \Lambda(\tau)$. We also upper bound any appearance of $\tau$ otherwise by $t$ because these are always in the numerator. So far, using $n = 2m+1$, these reductions give
            \begin{equation} \label{eq:final_bossR2m}
            \begin{split}
                \norm{R_{2m}} 
                &< \frac{\norm{a}_1}{(2m+1)!} ((2m+1)\Lambda_\mathrm{max} t)^{2m+1} \max_j  \left(\frac{1}{k_j} + \frac{1}{2}+\frac{2\Lambda t L}{2m+1}\right)^{2m+1}.
            \end{split}
            \end{equation}
            Employing a Stirling bound on the factorial, and factoring out an additional $L$ from the rightmost term, 
            \begin{align}
                \norm{R_{2m}} < \frac{\norm{a}_1}{\sqrt{2\pi(2m+1)}} (e L \Lambda_\mathrm{max} t)^{2m+1} \max_j \left(\frac{1}{Lk_j} + \frac{1}{2L} + \frac{2\Lambda t}{2m+1}\right)^{2m+1}.
            \end{align}
            We now apply the assumption that $e L \Lambda_\mathrm{max} t < 1$ to upper bound the $\max_j$ term, along with $k_j, L \geq 1$.
            \begin{align}
                \max_j \left(\frac{1}{Lk_j} + \frac{1}{2L} + \frac{2\Lambda t}{2m+1}\right)^{2m+1} < \left(\frac32 + \frac{2}{3 e}\right)^{2m+1}
            \end{align}
            Thus,
            \begin{align}
            \begin{aligned}
                \norm{R_{2m}} &< \frac{\norm{a}_1}{2\sqrt{\pi m}}\left(\left(\frac{3e}{2} + \frac{2}{3}\right)L\max_{\tau \in [0,t]} \Lambda_{2m+1}(\tau) t\right)^{2m+1} \\
                &< \frac{\norm{a}_1}{2\sqrt{\pi m}}\left(5 L\max_{\tau \in [0,t]} \Lambda_{2m+1}(\tau) t\right)^{2m+1}.
            \end{aligned}
            \end{align}
            In these last lines, we remind ourselves that $\Lambda$ has the subscript $2m+1$ as per~\defn{LambdaBound}.
        \end{proof}
            
    \subsection{Time Step Analysis} \label{sec:LongTimeSim}
        The next ingredient we need for a complexity analysis is asymptotic bounds on the number of subintervals $r$ needed in the time mesh. This will be the concern of this section. Unfortunately, in pursuing best-case bounds on $r$, we eschew a practical procedure for generating the time points $t_i$.~\app{adapt_deltat} provides a concrete procedure which is based on the analysis of this section.
        
        For time dependent Hamiltonians, because the cost per unit time can vary with $t$ in general, one should adaptively choose the step size depending on the cost. For our purposes, this means choosing a step size inversely proportional to the energy measure $\Lambda_{2m+1}(t)$. We will explore this adaptive time stepping and show $L^1$-norm scaling with $\Lambda_{2m+1}(t)$ here.

        To derive bounds on $r$, we will need to assume something about size of the derivative $\dot{\Lambda}_{2m+1}$ compared to $\Lambda_{2m+1}$ itself. Given a $\Lambda_n$-bound, a differentiable (smooth, even) $\Lambda_n$-bound exists. From now on, we consider $\Lambda_n$-bounds for which there exists a $K \in \mathbb{R}_+$ be such that $\lvert\dot{\Lambda}_n(t)\rvert \leq K \Lambda_n (t)^2$ for all $t \in [0,T]$. Given $H$ that is $\Lambda_{n}$ boundable, there is always, in fact, a $\Lambda_{n}$ bound such that $K$ exists and is arbitrarily close to zero. For example, we may take a constant bound $\Lambda_{n}' := \max_t \Lambda_{n}(t)$, noting that $\Lambda_{n}$ is continuous on a compact interval. Of course, $\Lambda'$ does not capture the changing behavior of $H(t)$, and is therefore suboptimal. Nevertheless, we've demonstrated that our additional assumptions are not much more restrictive than those we've already made. Note that (in natural units) $K$ is dimensionless. 
        
        With these preliminaries in place, the following result provides an upper bound on the number of time steps needed for our MPF algorithm.
        \begin{lemma}\label{lem:adaptive}
            Let $H$ satisfy the assumptions of~\thm{errorMain}, and let $\Lambda_{2m+1}$ be a $\Lambda_{2m+1}$-bound for $H$ such that, for some $K\in \mathbb{R}_+$, $|\dot{\Lambda}_{2m+1} (t)| \leq K \Lambda_{2m+1}(t)^2$ for all $t \in [0,T]$.  For every $\epsilon > 0$, there exists a list $(t_0, t_1, \dots, t_r)$ of monotonically increasing times $t_j \in [0,T]$, with $t_0 = 0$ and $t_r = T$, such that
            \begin{equation*} 
                \|U(T,0) - \prod_{i=1}^{r} U_{2,m}(t_i,t_{i-1})\| \leq \epsilon
            \end{equation*}
            with the number of time steps $r$ bounded above as
            \begin{equation*}
                r  \leq \left\lfloor \left(5 \left(1+\frac{3}{2}K\right) L \|\Lambda\|_1 \right)^{1+\frac{1}{2m}}\left( \frac{ \|a\|_1}{\epsilon\sqrt{\pi m}}\right)^{\frac{1}{2m}}\right\rfloor.
            \end{equation*}
            Here, $\|\Lambda_{2m+1}\|_1$ is the $L^1$ norm.
            \begin{equation*}
                \|\Lambda_{2m+1}\|_1 := \int_{0}^T \Lambda_{2m+1}(t) dt
            \end{equation*}
        \end{lemma}
        \begin{proof}
            As discussed in~\sec{ErrorAnalysis} in order to satisfy the $\epsilon$-error constraint of~\lem{adaptive}, it suffices that the error on each subinterval is less than $\epsilon/r$. Using~\thm{errorMain}, the sum is bounded as
            \begin{equation}
                \sum_{i=1}^{r}  \|U(t_i,t_{i-1}) - U_{2,m}(t_i,t_{i-1})\| \leq \frac{\|a\|_1}{\sqrt{\pi m}}\sum_{i=1}^{r}\left(5 L\max_{\tau\in [t_{i-1},t_i]}\Lambda_{2m+1}(\tau)(t_i - t_{i-1}) \right)^{2m+1}.
            \end{equation}
            To ensure an overall error $\epsilon$, it therefore suffices to produce a mesh such that for each $i$,
            \begin{equation}
                \frac{\|a\|_1}{\sqrt{\pi m}} \left(5 L \max_{\tau\in [t_{i-1},t_i]}\Lambda_{2m+1}(\tau)(t_i-t_{i-1})\right)^{2m+1} \leq \epsilon/r.
            \end{equation}
            Rearranging, this corresponds to choosing $t_i,$ given all other parameters, that satisfy
            \begin{equation} \label{eq:timestepBD}
                L \max_{\tau\in [t_{i-1},t_i]}\Lambda_{2m+1}(\tau)(t_i -t_{i-1}) \leq \frac{1}{5} \left(\frac{\epsilon\sqrt{\pi m}}{\|a\|_1 r} \right)^{1/(2m+1)}.
            \end{equation}
            
            We now digress in order to relate $\max_\tau \Lambda_{2m+1}(\tau)$ and its average. Here is where we will make use of the $K$-bounds on the derivative $\dot{\Lambda}$, we closely follow arguments found in~\cite{Wiebe2010Higher}. From the inequality in the lemma statement, we have
            \begin{align}
            \begin{aligned}
                \left\lvert\frac{\dot{\Lambda}_{2m+1}(t)}{\Lambda_{2m+1} (t)^2}\right\rvert &\leq K \\
                \left\lvert\frac{d}{dt} \frac{1}{\Lambda_{2m+1} (t)}\right\rvert &\leq K.
            \end{aligned}
            \end{align}
            Suppose the time $t_{i-1}$ has been chosen by the previous iteration (if $i = 1$, $t_0 = 0$). Let $t > t_{i-1}$ and integrate the above inequality from $t_{i-1}$ to $t$.
            \begin{align}
            \begin{aligned}
                \int_{t_{i-1}}^{t} \left\lvert \frac{d}{d\tau}\frac{1}{\Lambda_{2m+1} (\tau)}\right\rvert d\tau &\leq K (t-t_{i-1}) \\
                \left\lvert\int_{t_{i-1}}^{t} \frac{d}{d\tau} \frac{1}{\Lambda_{2m+1} (\tau)} d\tau\right\rvert &\leq K (t-t_{i-1}) \\
                \left\lvert\frac{1}{\Lambda_{2m+1}(t)} - \frac{1}{\Lambda_{2m+1}(t_{i-1})}\right\rvert &\leq K (t-t_{i-1}).
            \end{aligned}
            \end{align}
            Let us rearrange this in terms of $\Lambda_{2m+1} (t)$ alone.
            \begin{align} \label{eq:LambdaBounds}
            \begin{aligned}
                -K(t-t_{i-1}) &\leq \frac{1}{\Lambda_{2m+1} (t)} -\frac{1}{\Lambda_{2m+1} (t_{i-1})} \leq K (t-t_{i-1}) \\
                \frac{1}{\Lambda_{2m+1} (t_{i-1})}-K (t-t_{i-1}) &\leq \frac{1}{\Lambda_{2m+1} (t)} \leq \frac{1}{\Lambda_{2m+1} (t_{i-1})} + K (t-t_{i-1}) \\
                \frac{\Lambda_{2m+1} (t_{i-1})}{1+K (t-t_{i-1}) \Lambda_{2m+1} (t_{i-1})} &\leq \Lambda_{2m+1} (t) \leq \frac{\Lambda_{2m+1} (t_{i-1})}{1-K(t-t_{i-1}) \Lambda_{2m+1} (t_{i-1})}.
            \end{aligned}
            \end{align}
            The lowerbound inequality holds for all $t > t_{i-1}$, while the upper bound only holds when 
            \begin{equation}
                (t-t_{i-1})\Lambda_{2m+1}(t_{i-1})K < 1.
            \end{equation} 
            We restrict our attention to $t$ for which both bounds hold. Consider, for the moment, only the leftmost inequality. The lower bound on the left is monotonically decreasing with $t$. This means that it is also a uniform lower bound on $\Lambda_{2m+1} (t')$  for any $t' \in [t_{i-1}, t]$. Therefore, it is a lower bound for the average  $\bar{\Lambda}_{2m+1}(t)$ on the interval $[t_{i-1}, t]$.
            \begin{equation}
                \bar{\Lambda}_{2m+1}(t, t_{i-1}) := \frac{1}{t-t_{i-1}} \int_{t_{i-1}}^{t} \Lambda_{2m+1} (\tau) d\tau
            \end{equation}
            That is,
            \begin{equation}
                 \frac{\Lambda_{2m+1} (t_{i-1})}{1+K (t-t_{i-1}) \Lambda_{2m+1} (t_{i-1})} \leq \bar{\Lambda}_{2m+1}(t, t_{i-1}),
            \end{equation}
            or, after isolating for $\Lambda_{2m+1}(t_{i-1})$
            \begin{equation} \label{eq:average_bounds_value}
                \Lambda_{2m+1}(t_{i-1}) \leq \frac{\bar{\Lambda}_{2m+1}(t, t_{i-1})}{1-K(t-t_{i-1}) \bar{\Lambda}_{2m+1}(t, t_{i-1})}.
            \end{equation}
                  
            At this point, let's now consider the upper bound in equation~\eqref{eq:LambdaBounds}. This bound is monotonically \emph{increasing} in $t$, and therefore also upper bounds $\Lambda_{2m+1}(\tau)$ for any $\tau$ in $[t_{i-1}, t]$. Therefore, it is also a bound for the maximum. 
            \begin{align}
                \max_{\tau\in[t_{i-1},t]}\Lambda_{2m+1}(\tau) \leq \frac{\Lambda_{2m+1}(t_{i-1})}{1- K(t-t_{i-1}) \Lambda_{2m+1}(t_{i-1})}.
            \end{align}
            Substituting bounds for $\Lambda_{2m+1}(t_{i-1})$ from equation~\eqref{eq:average_bounds_value} gives us a bound on the maximum value in terms of the average.
            \begin{align}
                \max_{\tau\in[t_{i-1},t]}\Lambda_{2m+1}(\tau) \leq\frac{\bar{\Lambda}_{2m+1}(t, t_{i-1})}{1-\frac{3}{2} K\bar{\Lambda}_{2m+1}(t, t_{i-1})(t-t_{i-1}) }.
            \end{align}
            Solving for the average value of $\Lambda_{2m+1}$, and multiplying by $t-t_{i-1}$ on both sides,
            \begin{align} \label{eq:LambdabarbyMax}
                (t-t_{i-1})\bar{\Lambda}_{2m+1}(t, t_{i-1})  \geq \frac{(t-t_{i-1})\max_{\tau\in[t_{i-1},t]}\Lambda_{2m+1}(\tau)}{1+\frac{3}{2}K (t-t_{i-1})\max_{\tau\in[t_{i-1},t]}\Lambda_{2m+1}(\tau) }
            \end{align}
            Let us finally choose a $t=t_i$ which will serve as the next time step in the adaptive scheme. We would like come as close as possible to saturating equation~\eqref{eq:timestepBD} while staying within the constraint imposed by the maximum bound of equation~\eqref{eq:LambdaBounds}. Thus, we choose $t_i$ such that
            \begin{align}
                \max_{\tau\in[t_{i-1},t_i]}\Lambda_{2m+1}(\tau) (t_i-t_{i-1}) = \min \left\{\frac{1}{K}, \frac{1}{5L}\left(\frac{\epsilon\sqrt{\pi m}}{\|a\|_1 r}\right)^{1/(2m+1)}\right\}.
            \end{align}
            Since $K$ is a constant, for asymptotic purposes we will assume sufficiently small $\epsilon$ such that the right term is smaller. Plugging in to~\eqref{eq:LambdabarbyMax} yields
            \begin{equation}
                \bar{\Lambda}_{2m+1}(t_i, t_{i-1})(t_i - t_{i-1}) \geq  \frac{\frac{1}{5L} \left(\frac{\epsilon\sqrt{\pi m}}{\|a\|_1 r} \right)^{1/(2m+1)}}{1+\frac{3}{2}\frac{1}{5L} \left(\frac{\epsilon\sqrt{\pi m}}{\|a\|_1 r} \right)^{1/(2m+1)} K }.\label{eq:rgInt}
            \end{equation}
            We then find, by using the fact that $\frac{1}{5L} \left(\frac{\epsilon\sqrt{\pi m}}{\|a\|_1 r} \right)^{1/(2m+1)} < 1$ and by
            summing over $i=1,\ldots,r$ in~\eqref{eq:rgInt} that
            \begin{equation}
                \|\Lambda\|_1 \geq r^{\frac{2m}{2m+1}}
                \frac{1}{5L}\left(\frac{\epsilon\sqrt{\pi m}}{\|a\|_1} \right)^{1/(2m+1)}\left(1+\frac{3}{2} K\right)^{-1}
            \end{equation}
            Finally, rearranging the above, this implies that the number of steps required for the MPF algorithm is upper bounded as
            \begin{align} \label{eq:final_step_bound}
                r  \leq\left(5 \left(1+\frac{3}{2}K\right) L \|\Lambda\|_1 \right)^{1+\frac{1}{2m}}\left( \frac{ \|a\|_1}{\epsilon\sqrt{\pi m}}\right)^{\frac{1}{2m}}.
            \end{align}
            The result then directly follows from the requirement that $r$ is an integer.
        \end{proof}
        To summarize, we've provided an upper bound on the number of steps $r$ needed given assumptions on the derivative of $\Lambda_{2m+1}$. What is perhaps objectionable is that, in determining our subsequent time stepping, we seemed to need information about the total number of steps $r$ that we would end up with. While this does not detract from the correctness of our result, it does indicate possible difficulty in constructing a suitable set of $t_j$ for which the Lemma holds. One approach is to guess the final number $r_\mathrm{try}$ of steps needed, construct the mesh according to the proof, then see if $r_\mathrm{try}$ can be made correct. This approach is considered in~\app{adapt_deltat}.

        \subsection{Query Complexity} \label{sec:QueryComplexity}
            With the results of the previous two sections, we proceed to bound the query complexity needed to perform a time dependent MPF simulation. First, we define a set of oracles that are appropriate for this simulation problem. As discussed above, the most natural input model in our setting is the linear combinations of Hamiltonians model
            \begin{equation}
                H = \sum_{j=1}^L \alpha_j (t) H_j,
            \end{equation}
            where $\alpha_j:[0,T]\rightarrow\mathbb{R}$ has $2m+1$ continuous derivatives and $H_j \in \mathrm{Herm}(\mathbb{C}^{2^n})$. Without loss of generality, we assume $\|H_j\| \leq 1$. We discretize $[0,T]$ into $2^{n_t}$ uniform grid points $t_k = k T/2^{n_t}$ for $k \in [0,2^{n_t})\cap\mathbb{Z}$, and define $\alpha_{jk} := \alpha_j(t_k)$. Let $\delta t := T/2^{n_t}$. Let $U_\alpha$ and $U_H$ be unitary oracles which provide the input Hamiltonian as follows.
            \begin{align} \label{eq:oracles}
            \begin{aligned}
                U_\alpha \ket{j}\ket{k}\ket{\tau}\ket{0} &:=  \ket{j}\ket{k}\ket{\tau}\ket{\alpha_{jk}\tau}. \\
                U_H \ket{j} \ket{\alpha_{jk}\tau} \ket{\psi} &:=\ket{j} \ket{\alpha_{jk}\tau} \exp\{-iH_j \alpha_{jk}\tau\}\ket{\psi}
            \end{aligned}
            \end{align}
            The oracle $U_\alpha$ encodes a reversible classical computation and may be taken as self-inverse. Here $\ket{\tau}$ encodes a step of size $\tau \in \mathbb{R}$ in binary using $n_c$ qubits. Such step sizes are always nonnegative for the low-order formulas we consider, and therefore we take $\tau \in [0,T]$. Hence, $\delta t = T/2^{n_c}$ is the rounding error for the step sizes. We neglect rounding effects due to the values $\alpha_{jk}\tau$.
            
            Our first result concerns the approximate implementation of $U_2$ using the two oracles.
            \begin{lemma}\label{lem:2nd_order_formula_construction}
            Let $U_2(\tau + t,t)$ be the 2nd-order Suzuki-Trotter formula for the midpoint formula, with $t\in[0,T]$ and $\tau\in[0,T-t]$. Then an approximation $W_2$ can be constructed using at most $6L-3$ queries to $U_H$ and $U_\alpha$, such that
            \begin{equation}
                \|U_2(t +\tau,t) - W_2(t + \tau,t) \| \leq L \max_{j, t\in[0,T]} \abs{\dot{\alpha}_j(t + \tau/2)}\frac{T^2}{2^{n_c}}.
            \end{equation}
            \end{lemma}
            \begin{proof}
                Define $W_2$ as $U_2$ but with each $\alpha_j$ evaluated at the nearest discrete times in $\{t_k\}$. Using the techniques of~\cite{Wiebe2010Higher}, two queries to $U_\alpha$ and one query to $U_H$ are needed to exactly simulate each of the $2L-1$ exponentials present in $W_2$. Thus $3\times (2L-1)$ queries are needed total. To evaluate the discretization error, by Box 4.1 of~\cite{Nielsen2002QCQI} we have that
                \begin{equation}\label{eq:unitaryError}
                    \|W_2 - U_2\| \leq 2 \sum_{j=1}^L \|e^{-i H_j \alpha_j(\rnd[t + \tau/2])\tau/2} - e^{-i H_j \alpha_j(t + \tau/2)\tau/2}\|
                \end{equation}
                which in turn is upper bounded, through an application of the fundamental theorem of calculus, by
                \begin{equation}
                    2\sum_{j=1}^L \big\|H_j \alpha_j(\rnd \big[t + \frac{\tau}{2}\big])\frac{\tau}{2} - H_j \alpha_j(t+\frac{\tau}{2}) \frac{\tau}{2}\big\|
                \end{equation}
                where $\rnd$ rounds to the nearest $n_c$-bit value. Since $\|H_j\| \leq 1$ this is merely upper bounded as
                \begin{equation} \label{eq:round_vals_error}
                    \tau\sum_{j=1}^L \big\lvert \alpha_j(\rnd\big[t + \frac{\tau}{2}\big]) - \alpha_j(t+\frac{\tau}{2}) \big\rvert.
                \end{equation}
                By the fundamental theorem of calculus, with an integral upper bound, each term is upper bounded as $\delta t \max_{\delta t\in t \pm \delta t} \abs{\partial_t \alpha_j(t + \tau/2)}$. Maximizing over $[0,T]$ instead, and making other simplifying choices,we get a crude upper bound
                \begin{align}
                \begin{aligned}
                    \|W_2 - U_2\| &\leq \tau L \delta t \max_{j,[0,T]} \abs{\dot{\alpha}_j(t)} \\
                    &\leq L \frac{T^2}{2^{n_c}}\max_{j,[0,T]} \abs{\dot{\alpha}_j(t)}
                \end{aligned}
                \end{align}
                Rearranging this gives the inequality of the lemma statement.
            \end{proof}
            
            Having supplied an approximate base formula $W_2$ with our queries, we next need to implement an approximate MPF $W_{2,m}$ over a subinterval $[t_0,t_1]$. This is conventionally done through the use of "select" $\mathtt{SEL}$ and "prepare" $\mathtt{PREP}$ circuits
            \begin{align} \label{eq:MPF_SEL_PREP}
            \begin{aligned}
                \mathtt{PREP} \ket{0} &:= \sum_{j=1}^{m} \sqrt{\frac{\lvert a_j\rvert}{\|a\|_1}} \ket{j} \\
                \mathtt{SEL} \ket{j} \ket{\psi} &:= \sgn(a_j) \ket{j}W_2^{(k_j)}(t_1,t_0)\ket{\psi}            
            \end{aligned}
            \end{align}
            The circuit $\mathtt{PREP}$ can be implemented without any queries to $U_\alpha$ or $U_H$ whereas $\mathtt{SEL}$ requires $O(L \|\vec{k}\|_\infty)$ queries.  Following the well-conditioned MPF scheme of~\cite{Low2019Multiproduct} we have that $k_j \leq 3 m^2$. This implies that a query to $\mathtt{SEL}$ requires $O(L m^2)$ queries to $U_H$ and $U_{\alpha}$.

            We can use the $\mathtt{SEL}$ and $\mathtt{PREP}$ for a standard LCU block encoding in order to construct a time dependent MPF with base formula $W_2$.
            \begin{lemma}\label{lem:Wapprox}
                Under the assumptions of~\thm{errorMain} and the query model above, for any $[t_0,t_1] \subseteq [0,T]$ the time dependent MPF $W_{2,m}$ with base formula $W_2$ satisfies
                \begin{align}
                    \|W_{2,m}(t_1,t_0) - U(t_1,t_0)\| \in O\left({\|a\|_1} \left(\max_{t\in[t_0,t_1]} \Lambda_{2m+1}(t) T\right)^{2m+1}\right),
                \end{align}
                provided that
                \begin{align}
                    n_c \geq \log\left(\frac{3 \sqrt{\pi}m^{5/2}{L} \max_{t,j} |\partial_t \alpha_j(t)|(t_1 - t_0)^2}{\left(5L\max_{t\in[t_0,t_1]} \Lambda_{2m+1}(t) (t_1 - t_0)\right)^{2m+1}} \right),
                \end{align}
                and can be constructed with a number of queries to $U_H$ and $U_\alpha$ scaling as $O(m^2L)$.
            \end{lemma}
            \begin{proof}
                From Lemma 4 of~\cite{Berry2015Optimal}, we have
                \begin{align}\label{eq:blockEncode}
                \begin{aligned}
                    (\bra{0}\otimes I) (\mathtt{PREP}^\dagger) \mathtt{SEL} (\mathtt{PREP})   (\ket{0}\otimes I) &= \frac{1}{\|a\|_1} \sum_{j=1}^m a_j W_2^{(k_j)}\left(t_1, t_0\right) \\
                    &= W_{2,m}(t_1,t_0)/\|a\|_1.
                \end{aligned}
                \end{align}
                Let $\delta' > 0$ be such that, for all $j$ and $\ell \in \{1,\dots, k_j\}$, 
                \begin{equation}
                    \left\| W_2\left(\Delta t\frac{\ell}{k_j} + t_0,\Delta t\frac{\ell-1}{k_j} + t_0\right) - U_2\left(\Delta t\frac{\ell}{k_j} + t_0,\Delta t \frac{\ell-1}{k_j} + t_0\right) \right\|\leq\delta'.
                \end{equation} 
                where $\Delta t = t_1 - t_0$. Then, by invoking Box 4.1 from~\cite{Nielsen2002QCQI},
                \begin{align}
                    \left\|U_2^{(k_j)}(t_1, t_0) -  W_2^{(k_j)}(t_1, t_0) \right\| \leq k_j \delta'
                \end{align}
                which, since $k_j \leq 3 m^2$, implies that
                \begin{align}
                    \|V_{2,m}(t_1,t_0) - W_{2,m}(t_1,t_0)\| \leq 3 m^2 \delta' \|a\|_1.
                \end{align}
                We supply $\delta'$ using~\lem{2nd_order_formula_construction}, obtaining
                \begin{align}
                    3m^2 \delta' \|a\|_1 \leq 3 m^2 \|a\|_1 L \max_{j, t\in[0,T]} \abs{\dot{\alpha}_j(t + \tau/2)}\frac{T^2}{2^{n_c}},
                \end{align}
                giving us a bound on the discretized MPF $W_{2,m}$ relative to the undiscretized $V_{2,m}$.
                
                It then follows from the triangle inequality and~\thm{errorMain} that
                \begin{align} \label{eq:MPF_discrete_error}
                \begin{aligned}
                    \left\|W_{2,m}(t_1,t_0) -U(t_1,t_0)\right\| &\leq \left\|U_{2,m}(t_1,t_0) -U(t_1,t_0)\right\| 
                    + \|W_{2,m}(t_1,t_0) - U_{2,m}(t_1,t_0)\| \\
                    &\leq \frac{\|a\|_1}{\sqrt{\pi m}} \left(5L\max_{t\in[0,T]} \Lambda_{2m+1}(t) T\right)^{2m+1} + 3m^2L \|a\|_1 \max_{j,t}\abs{\dot{\alpha}_j(t)}\frac{T^2}{2^{n_c}}.
                \end{aligned}
                \end{align}  
                Under the assumption that 
                \begin{equation}
                    n_c \geq \log\left(\frac{3 \sqrt{\pi} m^{5/2}{L} \max_{j,t}\abs{ \dot{\alpha}_j(t)}T^2}{\left(5L\max_{t\in[0,T]} \Lambda_{2m+1}(\tau) T\right)^{2m+1}} \right)
                \end{equation}
                the second term is bounded by the first~\eqref{eq:MPF_discrete_error}, so we have an upper bound
                \begin{equation}
                    \left\|\sum_{j=1}^M a_j \prod_{q=1}^{k_j} W_2\left(Tq/k_j,T(q-1)/k_j\right) -U(T,0)\right\|\leq  \frac{2\|a\|_1}{\sqrt{\pi m}} \left(5L\max_{t\in[0,T]} \Lambda_{2m+1}(t) T\right)^{2m+1}
                \end{equation}
                Since $U(T,0)$ is unitary, we know that the MPF implemented by our algorithm is close to a unitary.  This means that we satisfy the preconditions of robust oblivious amplitude amplification given by Lemma 5 of~\cite{Berry2015Optimal}. This result implies that using $O(\|a\|_1)$ applications of the unitary given by~\eqref{eq:blockEncode}, we can implement an operator $\widetilde{W}(T,0)$ such that (for constant $m$)
                \begin{equation}
                    \|W_{2,m}(t_1,t_0) - U(t_1,t_0)\| \in O\left(\|a\|_1 \left(\max_{t\in[0,T]} \Lambda_{2m+1}(t) (t_1 - t_0)\right)^{2m+1}\right).
                \end{equation}
                The number of queries scales as
                \begin{equation}
                    Q_\mathrm{step} \in O(\|a\|_1 m^2 {L}) \subseteq \widetilde{O}(m^2 {L}).
                \end{equation}
            \end{proof}
            With the short-time simulation costs in place we are now ready to state our main theorem, which bounds the number of queries needed to perform the full multiproduct simulation of a time dependent Hamiltonian.
            \begin{theorem}
                In the query setting above, and under the assumptions of~\thm{errorMain}, and~\lem{adaptive} ($\Lambda_{2m+1}$-bounded $H$ with $K$ bound on $\dot{\Lambda}_{2m+1}$), we have that the number of queries $Q_\mathrm{tot}$ needed to $U_\alpha$ and $U_H$ to construct an operator $W_\mathrm{tot}(T,0)$ simulate a time dependent Hamiltonian of the form $\sum_{j=1}^{L} \alpha_j(t) H_j$ such that $\|(\bra{0} \otimes \openone)W_\mathrm{tot}(T,0) (\ket{0} \otimes \openone) - U(T,0)\| \leq \epsilon $ satisfies
                \begin{equation*}
                    Q_\mathrm{tot} \in \widetilde{O}\left( {L(1+K) \|\Lambda_{2m+1}\|_1 \log^2(1/\epsilon)}\right),
                \end{equation*}
                and the total number of auxiliary qubits is in
                \begin{equation*}
                    \widetilde{O}\left(\log\left(\frac{L(1+K)\norm{\Lambda_{2m+1}}_1\max_{j,t} \lvert \dot{\alpha}_j(t)\rvert T^2}{\epsilon}\right)\right)\;.
                \end{equation*}
            \end{theorem}
            
            \begin{proof}
                From~\lem{adaptive} we have that the number of segments needed to perform a the simulation within error $\epsilon$ obeys
                \begin{equation}
                    r \in \widetilde{O}\left(\frac{((1+K) \norm{\Lambda_{2m+1}}_1)^{1+1/(2m)} }{\epsilon^{1/(2m)}} \right).
                \end{equation}
                Therefore, using~\lem{Wapprox},
                \begin{equation}
                \label{eq:m2lr}
                    Q_\mathrm{tot} \in \widetilde{O}(m^2 L r) \subseteq \widetilde{O}\left(\frac{m^2 L((1+K) \norm{\Lambda_{2m+1}}_1)^{1+1/(2m)} }{\epsilon^{1/(2m)}} \right)
                \end{equation}
                the approximate value of the optimal $m$ can be found by equating the exponentially shrinking component of the cost to the polynomially increasing value of $m$.  We choose $m$ to satisfy
                \begin{equation}
                    m^2 = \left(\frac{(1+K) \norm{\Lambda_{2m+1}}_1}{\epsilon} \right)^{1/2m}.
                \end{equation}
                Solving for $m$ yields
                \begin{equation} \label{eq:fun_with_lambert}
                    m = \frac{\log \big((1+K) \norm{\Lambda_{2m+1}}_1/\epsilon \big)}{4\LambertW\Big(\log \big((1+K) \norm{\Lambda_{2m+1}}_1/\epsilon\big)/4\Big)} \in \widetilde{O}\left(\log \left(\frac{(1+K) \norm{\Lambda_{2m+1}}_1}{\epsilon} \right) \right)
                \end{equation}
                This implies that the query complexity $Q_\mathrm{tot}$ is in
                \begin{equation}
                    \widetilde{O}\left( {L(1+K) \norm{\Lambda_{2m+1}}_1 \log^2(1/\epsilon)}\right).
                \end{equation}
                The number of auxiliary qubits needed in the construction is in $O(\log(m))$ to implement the MPF and $(\lceil\log{L}\rceil+n_c)$ to implement the $U_\alpha$ oracle. From the result of~\lem{Wapprox} we see that $n_c$ dominates this cost. We thus have a number of auxiliary qubits scaling as
                \begin{align}
                \begin{aligned}
                    n_\mathrm{aux} &\in O\left(\log\left(\frac{m^2L\max|\partial_t\alpha_j(t)|T^2}{\left(\max_{t\in[0,T]} \Lambda_{2m+1}(t) T\right)^{2m+1}}\right)\right)\\
                    &\in O\left(\log\left(\frac{m^2Lr\|a\|_1\max|\partial_t\alpha_j(t)|T^2}{\epsilon}\right)\right)\\
                    &\in\widetilde{O}\left(\log\left(\frac{L(1+K)\norm{\Lambda_{2m+1}}_1\max|\partial_t\alpha_j(t)|T^2}{\epsilon}\right)\right)\;
                \end{aligned}
                \end{align}
                where used Eq.~\eqref{eq:m2lr} and Eq.~\eqref{eq:fun_with_lambert} above.
            \end{proof}
            This shows that the cost of quantum simulation using MPFs broadly conforms to the cost scalings that one would expect of previous methods.  In particular, similar to the truncated Dyson series simulation method~\cite{Low2018Interaction,Kieferova2019Dyson} we obtain that the cost of simulating a time dependent Hamiltonian scales near-linearly with time $T$ and poly-logarithmically with $1/\epsilon$. 
            
    \subsection{Numerical Demonstrations} \label{sec:NumericalDemos}
        In the above sections, we developed and characterized MPFs for time dependent simulations by showing their existence and proving error bounds. However, these bounds are unlikely to be the final word on the performance of the algorithm. For example, we already mentioned that, for time independent $H$, the MPF of ~\defn{timedependentMPF} is exact in cases where the Hamiltonian consists of only commuting terms. Yet this behavior is not captured in the bound of \thm{errorMain} because $\Lambda_{2m+1}$ is at least as large as $\norm{H}$. This discrepancy is unrelated to the fact that, in practice, the $2^\mathrm{nd}$-order formula $U_2$ can only be computed approximately.
            
        To begin bridging the gap between algorithm's actual performance and our bounds, we investigate time dependent MPFs empirically through two numerical examples. We compute $U_{2,m}$ for these systems on a classical computer (using matrix computations) and compare the result with the exact propagator (computed within machine $\epsilon$). The vector $\vec{k}\in \mathbb{Z}_+^m$ we will use comes from the bottom half of Table I from \cite{Low2019Multiproduct}, which minimizes $\vec{\norm{k}}$ for $\norm{a}_1 \leq 2$.
            
        In general, deriving an analytical solution for the propagator given a time dependent Hamiltonian is challenging or impossible. To bypass this problem, we will consider a time independent Hamiltonian which is viewed from a ``non-inertial" frame, thereby rendering the dynamics time dependent in the new frame. More specifically, suppose $H$ is a time independent Hamiltonian with propagator $U(t) = e^{-iHt}$ (henceforth the initial time is set to zero). Let $\ket{\psi_t}$ be the solution to the Schrödinger equation $i\partial_t \ket{\psi_t} = H \ket{\psi_t}$. Under a frame transformation $T(t)$, which transforms vectors as $\vert\tilde{\psi}_t\rangle = T(t)\ket{\psi_t}$, the Hamiltonian and propagator transform as 
        \begin{align}
        \begin{aligned}
            \tilde{U}(t) &= T(t) U(t) \\
            \tilde{H}(t) &= i \frac{\partial T(t)}{\partial t} T(t)^\dagger + T(t) H(t) T(t)^\dagger.
        \end{aligned}
        \end{align}
        Thus, in order to benchmark the error of the MPF, we compute $\tilde{U}_k$ for Hamiltonian $\tilde{H}$, then compare with the exact propagator (accurate to machine precision).
        \begin{equation}
            \epsilon_c = \norm{\tilde{U}_{\vec{k}}(t)-T(t)U(t)}
        \end{equation}

        \subsubsection{Example 1: Electron in Magnetic field, Rotating Frame} \label{electron_numerical_example}
            As a very simple first demonstration, consider a spin-1/2 particle (say, electron) in a homogeneous external magnetic field $B$. Choose a coordinate system such that $B$ makes an angle $\theta$ with respect to the $z$-axis, and lies within the $xz$ plane. This system can be described by the Hamiltonian
            \begin{equation}
                H = \mu B(\cos\theta Z/2 + \sin\theta X/2)
            \end{equation}
            where $Z$ and $X$ (and later $Y$) are Pauli operators, and $\mu$ is a coupling parameter that will henceforth be set to one. The propagator $U(t) = e^{-i H t}$ is easy to compute, and corresponds to precession about the magnetic field axis with frequency $B$.
            
            To obtain a time dependent problem, let's shift to a reference frame that rotates with angular frequency $\omega$ about the $z$-axis. The transformation is given by $R_z(\omega t)$, where $R_a$ is the usual $SU(2)$ rotation operator about axis $a$. The Hamiltonian in the rotating frame is
            \begin{equation}
                \tilde{H}(t) = (\omega + B\cos\theta) Z/2 + B\sin\theta(\cos\omega t X/2 + \sin\omega t Y/2)
            \end{equation}
            Because we know that this Hamiltonian is just a transformed time independent system, it is easy to compute the exact propagator $\tilde{U}(t)$.
            \begin{equation}
                \tilde{U}(t) = R_z (\omega t) U(t)
            \end{equation}
    
            Though it is not strictly necessary to run the algorithm, let's compute an appropriate $\Lambda(t)$ upper bound. The spectral norm of $\tilde{H}$ may be upper bounded as
            \begin{equation}
                \norm{\tilde{H}} \leq \frac{|\omega + B \cos\theta|}{2} +  |B \sin\theta| 
            \end{equation}
            while the derivatives $\tilde{H}^{(n)}(t)$ have the bound
            \begin{align}
            \begin{aligned}
                \norm{\tilde{H}^{(n)}(t)} &\leq |B \sin\theta \omega^n| \\
                \sqrt[n+1]{\norm{\tilde{H}^{(n)}(t)}} &\leq \omega \left|\frac{B\sin\theta}{\omega}\right|^{1/n+1}.
            \end{aligned}
            \end{align}
            For $\omega$ not too much larger than $B$, we see then that $\Lambda(t) = \omega$ is an appropriate choice. 
            
            The first thing to check will be that the error has the appropriate power law scaling. Namely, for $M$-term formulas, the error $\epsilon_c$ for small $t$ should scale as $O(t^{2m+1})$ or better. We can check this by computing the ``running power" $p(t,t')$.
            \begin{equation} \label{eq:running_power}
                p(t, t') := \frac{\log \epsilon_{t}/\epsilon_{t'}}{\log t/t'}
            \end{equation}
            For different but small values of $t, t'$, the value of $p$ should approach the expected order of the error: $2m+1$. Indeed, this is precisely the behavior observed in~\fig{power_law_scaling}. For sufficiently small simulation times, a power-law dependence on the simulation error is observed, and the corresponding power is as anticipated. Additionally, we see that the error decreases by orders of magnitude with each additional term once the power-law regime is reached. Choosing $m > 4$ in this example quickly leads to machine precision being the dominant error source.
            \begin{figure}[t] 
                \centering
                \includegraphics[width=0.45\linewidth]{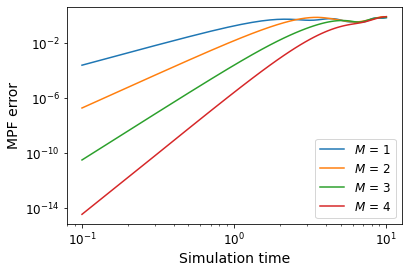}\hspace{0.5cm}
                \includegraphics[width=0.45\linewidth]{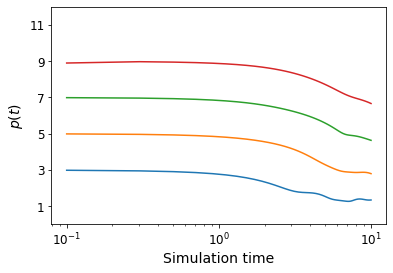}
                \caption{(left) Multiproduct errors plotted against simulation time, for several low-order MPFs, on a log-log plot. Notice the power law scaling for small values of $t$. The parameters used here are $B=1, \omega = 4$, $\theta = \pi/6$. For larger $M$, one quickly runs into machine precision becoming the dominant error source. (right) The running power $p(t, t')$ defined in equation~\eqref{eq:running_power}, with $t' = .3$. Note the plateau corresponds with the anticipated value of $2M+1$. \label{fig:power_law_scaling}}
            \end{figure}
            
            Next, we vary the MPF order $m$ for fixed simulation time $t$. Since $\Lambda = \omega$, our bounds predict an exponential decay in the error, but only provided $t < 1/\omega$. Otherwise, the bounds grow exponentially and say nothing useful about performance. In \fig{exp_decrease_error}, we fix $t$ at several different times and plot the error dependence on the multiproduct order $m$. Past a certain threshold value for $m$ (which increases with $t$) an exponential decay in error is observed, possibly superexponential. It is promising that, even for $t = 10$, the exponential decay is eventually achieved at $m \gtrsim 6$. This suggests our error bounds may be too conservative, and in particular MPFs could \emph{absolutely} converge to $U$ as $m\rightarrow \infty$ in certain circumstances. This would be a notable improvement to product formulas alone, which tend to lead to errors that diverge as $m\rightarrow \infty$ if the time step $t$ remains fixed~\cite{Berry2007Sparse,Wiebe2010Higher,Childs2019Randomization}.  In contrast,~\thm{errorMain} shows that if the time step is sufficiently small, then the MPF converges to the exact result. However, such convergence is not anticipated from the bounds for a large value such as $t=10$.
    
            Indeed, there are good reasons to believe the absolute convergence property holds more generically than this example. No matter how large the order $m$, we are still using a low order formula (such as the midpoint formula $U_2$) as a base. Moreover, recall that the MPF is essentially a sum of product formulas with different numbers of time steps (for the same time interval). As the order $m$ increases, higher weight is given to terms in the multiproduct sum with finer meshes. Correspondingly, terms which have larger time steps, and therefore may not converge properly, become suppressed at large $m$. Such behavior is not reflected in our derived error bounds, so there is likely room for improvement.  
            
            \begin{figure}[t] 
                \centering
                \includegraphics[width=0.6\linewidth]{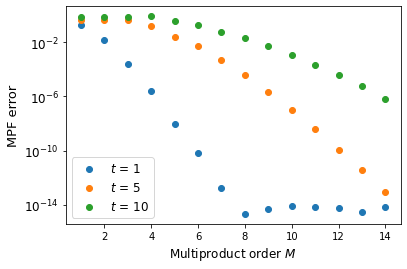}
                \caption{Multiproduct error shows an (super)exponential decrease in error for sufficiently large order $m$. The threshold for this regime is seen to increase as the simulation time increases. This behavior surpasses the expectation of our proven bounds, since there are no guarantees if the time step is too large. Note that, in practice, one should typically split a longer simulation time into smaller steps. The plateau for $t = 1, m>8$ is a result of machine precision limitations. Parameter values: $B=1, \omega = 4, \theta = \pi/6$.\label{fig:exp_decrease_error}}
            \end{figure}
            
            Practitioners in quantum simulation will likely want to know how MPFs fare against the more-familiar and simpler Trotter techniques. To facilitate this, numerical studies across a broad range of physically interesting systems would be desirable. Such a comprehensive analysis must be left to future work; here we will be satisfied with comparing MPFs with Trotterization for our spin-1/2 example. Our Trotterization is just an MPF with $m=1$, corresponding to a midpoint-formula approximation. To facilitate as fair a comparison as possible, we will keep the number of midpoint-formula queries between the two methods the same. That is, we will enforce the requirement
            \begin{align} \label{eq:fairness_condition}
                r_\mathrm{trot} = r_\mathrm{mpf} \max_j| k_j|
            \end{align}
            where $r_\mathrm{trot}$ and $r_\mathrm{mpf}$ are the number of time steps for Trotter and MPF, respectively. Note that the number of midpoint queries per time step for Trotter and MPFs are 1 and $ \max_j| k_j|$ respectively.
    
            \fig{Trott_MPF_fight} shows the results of these head-to-head comparisons for the several values of the magnetic field $B$ and rotation frequency $\omega$. The number of MPF steps $r_\mathrm{mpf}$ is fixed at 10, a reasonable value since it makes $\Lambda \Delta t \sim 1$ on each subinterval. As the MPF order increases, so does the number of Trotter steps $r_\mathrm{trot}$ by the  condition \eqref{eq:fairness_condition}. These results show that, for $m$ not too large, MPFs outperform Trotterization, at a value of the error $\epsilon$ which is large enough to be of practical significance for scientific or industrial applications.
            \begin{figure}[t] 
                \centering
                \includegraphics[width = \linewidth]{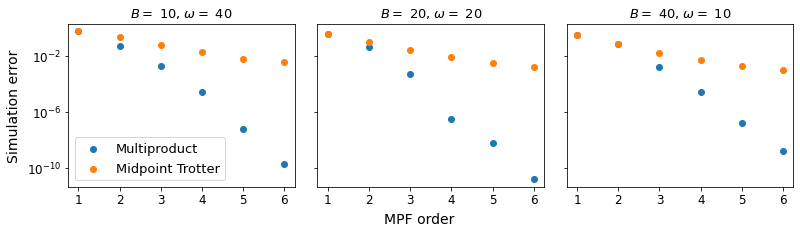}
                \caption{Simulation error (spectral norm) of MPFs and midpoint-formula Trotterization, for the spin-1/2 system, with number of midpoint-formula queries kept fixed between the two. Each plot corresponds to different values for the parameters $B$ and $\omega$, always with $\theta = \pi/6$. The number of MPF steps $r_\mathrm{mpf}$ is fixed at 10. The crossover point tends to occur for error $\epsilon > 10^{-3}$, which is large enough for practical significance. Such error tolerances can be orders of magnitude larger than those required in many quantum simulation proposals.~\cite{Reiher2017Reaction,Lee2021Tensor} \label{fig:Trott_MPF_fight}}
            \end{figure}
            Admittedly, the spin-1/2 system considered above is rather simplistic. However, we anticipate most of the inferences drawn above to hold even as we increase the dimensionality of the Hilbert space. For example, though the complexity of simulating $U_2$ generally increases as $\dim(H)$ grows, it does so both for MPFs and Trotterization. Nevertheless, benchmarking of MPFs on more complex systems would be a welcomed proof (or disproof) of concept.
    
        \subsubsection{Example 2: Spin Chain in Interaction Picture} \label{spinchain}
            As a first step towards more complicated many-body quantum systems, we investigate the use of MPFs for a particular one-dimensional chain of spins with nearest-neighbor interactions. As before, we will take advantage of a change of reference frame, allowing us to compare the multiproduct simulations with an machine precision simulation in an equivalent, time independent frame. In pursuit of a good case study, we seek a (time independent) Hamiltonian $H = H_0 + H_1$ which produces nontrivial time-dependence in the so-called ``interaction picture." We also ask that it satisfies a simple conservation law. A special instance of the 1D $XX$ model will suffice to meet these conditions. Consider a circular chain of $N$ qubits with nearest-neighbor hopping interactions, with Hamiltonian $H = H_0 + H_1$ of the form
            \begin{align}
            \begin{aligned}
                H_0 &= \sum_{k=1}^N \frac{\omega_k}{2} Z_k \\
                H_1 &= \sum_{k=1}^N \frac{J_k}{2} \left(X_k X_{k+1} + Y_k Y_{k+1}\right).
            \end{aligned}
            \end{align}
            Here, $\omega_k, J_k$ are real, site-dependent parameters, and any index increments are done modulo $N$. For any value of the parameters, the Hamiltonian conserves the total magnetization $\mu := \sum_k Z_k$. 
            \begin{equation}
                [\mu,H] = 0
            \end{equation}
    
            Conceptually will think of $H_0$ as a ``base" Hamiltonian, with perturbation $H_1$ generating interactions, though we make no assumptions as to the smallness of $H_1$. We will switch to an interaction picture which is comoving with the simple dynamics of $H_0$. In this frame, the Hamiltonian $\tilde{H}(t)$ is given by
            \begin{align} \label{eq:XX_interaction_pic}
            \begin{aligned}
                \tilde{H}(t) &= e^{i H_0 t} H_1 e^{-i H_0 t} \\
                &=\sum_{k=1}^N \frac{J_k}{2} \left(X_k(t) X_{k+1}(t) + Y_k(t) Y_{k+1}(t)\right)
            \end{aligned}
            \end{align}
            where
            \begin{align} \label{eq:time_e}
            \begin{aligned}
                X_k(t) &:= e^{i H_0 t} X_k e^{-i H_0 t} = \cos(\omega_k t) X_k - \sin(\omega_k t) Y_k \\
                Y_k(t) &:= e^{i H_0 t} Y_k e^{-i H_0 t} =\cos(\omega_k t) Y_k + \sin(\omega_k t) X_k
            \end{aligned}
            \end{align}
            correspond to rotating the pauli vectors about the $z$-axis with frequency $\omega_k$. We can express equation \eqref{eq:XX_interaction_pic} in terms of the time independent $X_k$ and $Y_k$ of the original frame,
            \begin{equation} \label{eq:XX_interaction_pic_2}
                \tilde{H}(t) = \sum_{k=1}^N \frac{J_k}{2}\big\{ \cos(\Delta\omega_k t) (X_k X_{k+1} + Y_k Y_{k+1}) + \sin(\Delta \omega_k t)(X_k Y_{k+1}-Y_k X_{k+1})\big\},
            \end{equation}
            where $\Delta \omega_k = \omega_{k+1}-\omega_k$. We see that having different qubit frequencies $\omega_k$ on neighboring sites should give rise to a nontrivial time-dependence in $\tilde{H}$. Another indication is gleaned from the commutator of $H_0$ and $H_1$.
            \begin{equation}
                [H_0, H_1] = -i \sum_k \frac{J_k}{2}(X_k Y_{k+1}-Y_k X_{k+1})(\Delta \omega_k).
            \end{equation}
            The time dependence in $H_I$ will be nontrivial when the commutator does not vanish, as occurs when $\Delta\omega_k \neq 0$. A simple choice is to set
            \begin{equation} \label{eq:minimal_parameters}
                J_k = J, \quad \omega_k = (-1)^k \omega.
            \end{equation}
            That is, the qubit frequency alternates sign at each site, and the coupling is translation invariant. For simplicity, we consider only even numbers of qubits to avoid frequency-matching at $k = N$.  
            Plugging \eqref{eq:minimal_parameters} into the expression for $\tilde{H}$ in \eqref{eq:XX_interaction_pic_2},
            \begin{equation} \label{eq:TD_spin_chain}
                \tilde{H}(t) = \frac{J}{2} \big(\cos(2\omega t) G_1 + \sin(2\omega t) G_2\big)
            \end{equation}
            where
            \begin{align}
            \begin{aligned}
                G_1 &= \sum_{k=1}^N X_k X_{k+1} + Y_k Y_{k+1} \\
                G_2 &= \sum_{k=1}^N (-1)^k (X_k Y_{k+1} - Y_k X_{k+1})
            \end{aligned}
            \end{align}
            As a final check, one can see that $G_1$ and $G_2$ do not commute with each other. Yet they both commute with $\mu$. Thus, $\tilde{H}(t)$ given in \eqref{eq:TD_spin_chain} is our model system to investigate.
    
            Assuming $\tilde{H}$ commutes with an observable $\mu$, to what degree does the MPF $U_{2,m}$ conserve $\mu$? Since $U_{2,m}$ is an algebraic combination of exponentials of $\tilde{H}$, $U_{2,m}$ also commutes with $\mu$. If $U_{2,m}$ were truly unitary, then the operator $\mu$ would evolve in the Heisenberg picture as
            \begin{equation}
                \mu_{2,m}(t) := U_{2,m}^\dagger(t) \mu U_{2,m}(t) = \mu
            \end{equation}
            as it would under the exact propagator $U$. However, $U_{2,m}$ is not necessarily unitary.
            \begin{equation}
                U_{2,m}^\dagger(t) U_{2,m}(t) \neq \openone 
            \end{equation}
            This implies that conservation laws are only approximately conserved. 
            \begin{equation}
                \mu_{2,m}(t) - \mu = \left(U_{2,m}^\dagger(t) U_{2,m}(t) - \openone\right)\mu \neq 0.
            \end{equation}
            Because $U_{2,m}(t) - U(t) \in O(t^{2m+1})$, so is $\left(U_{2,m}^\dagger(t) U_{2,m}(t) - \openone\right)$. 
            
            \fig{Conservation_laws} plots the deviations in the conserved $\mu$, $\|\mu - \mu_{2,m}(t)\|$, with respect to the simulation time. As the simulation time tends to zero, we see the expected power-law scaling, as evidence by the linear relationship on a log-log plot. For larger $m$, the slope and hence power $p$ increases, corresponding to improved performance. We can extract the power as the slope of the line, and this is plotted in the right frame. Notice there are sudden dips in the error at specific simulation times, which tend to occur before reaching the power law scaling regime. This could be due to cancellation between two terms in an error series of comparable magnitude. Similar phenomenon occurs in several other contexts, such as the error from adiabatic evolution~\cite{Wiebe2012Adiabatic}. Conclusive identification of these phenomenon will require further study.
            \begin{figure}[t] 
                \centering
                \includegraphics[width=0.45\linewidth]{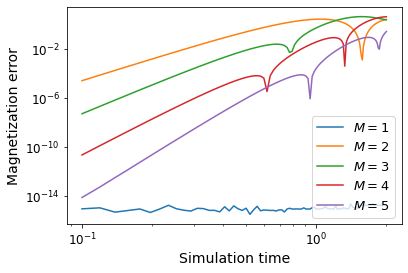}\hspace{0.5cm}
                \includegraphics[width=0.45\linewidth]{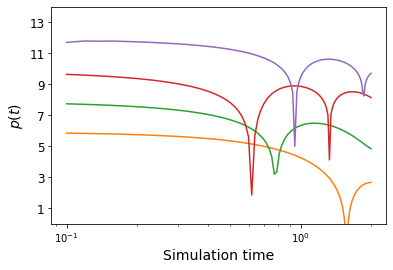}
                \caption{(left) Deviations from the conservation of magnetization $\mu$ under time-evolution by MPFs. Note that the order $m=1$ is simply a product formula evolution, which conserves $\mu$ exactly. For small simulation times, the expected power-law scaling is observed, with larger powers as $m$ increases. (right) The running power $p(t, t')$ as defined in \eqref{eq:running_power}, with $t' = .3$. Note the plateau at $2m + 2$, which indicates slightly better convergence than naively expected ($p = 2m + 1$). This phenomenon generalizes to other systems and is formalized by~\thm{Deviation_from_unitarity}. Parameter values: $N = 4$, $J = 1$, $\omega = 4$\label{fig:Conservation_laws}}
            \end{figure}
            Naively, we would expect $p = 2m+1$, but here we actually get slightly better: $p = 2m+2$. In fact, this scaling can be justified.  The following argument, a variant of which can be found in~\cite{Chin2010Multiproduct}, shows that the integrator is nearly unitary.
            \begin{theorem} \label{thm:Deviation_from_unitarity}
               The deviation of $U_{2,m}$ from being unitary obeys
                \begin{equation*}
                    \|U_{2,m}^\dagger(t) U_{2,m}(t) - \openone\| \in O(t^{2m+2})
                \end{equation*}
            \end{theorem}
            \begin{proof}
                We suppress all function evaluations at $t$ when convenient. Let $E := U_{2,m} - U$, so that $U_{2,m} = U + E$. Then, using the unitarity of $U$ and the fact that $E \in O(t^{2m+1})$,
                \begin{equation}
                    U_{2,m}^\dagger U_{2,m} = \openone + N + O(t^{4m+2})
                \end{equation}
                where
                \begin{align}
                    N := U^\dagger E + E^\dagger U.
                \end{align}
                Since $N \in O(t^{2m + 1})$, all of its derivatives up to degree $2m$ vanish when evaluated at $t = 0$. Hence, it suffices to show that
                \begin{equation}
                    N^{(2m + 1)}(0) = 0.
                \end{equation}
                We can expand this derivative in terms of $E$ and $U$ using the binomial theorem. When we evaluate at $t=0$, those terms with derivative less than degree $2m+1$ in $E$ vanish. We are left with
                \begin{equation} \label{eq:N_deriv}
                    N^{(2m + 1)}(0) = E^{\dagger (2m + 1)}(0) U(0) + U^\dagger(0) E^{(2m+1)}(0).
                \end{equation}
                We have $U(0) = U^\dagger(0) = \openone$. Moreover, by the time-symmetric property of $U$ and $U_{2,m}$, $E(t)$ is also symmetric. Therefore
                \begin{equation}
                    E^{\dagger(2m+1)}(0) = E^{(2m+1)}(-t) \Big\rvert_{t=0} = -E^{(2m+1)}(0).
                \end{equation}
                Hence, the two terms in~\eqref{eq:N_deriv} cancel, yielding $N^{(2m+1)}(0) = 0$. This completes the proof. 
            \end{proof}
            In summary, though MPFs do not inherently preserve commutations laws, the error is due to nonunitarity in $U_{2,m}$. This can be bounded and reduced in a systematic way, either by decreasing the time step or increasing the MPF order.

\section{Conclusion} \label{sec:Conclusion}
    The main contribution of this paper is a computational reduction, based on the $(t,t')$ formalism, of time dependent systems to time independent ones, allowing for the replacement of ordered operator exponentials with ordinary operator exponentials acting on a higher but finite dimensional space. This augmented clock system may be directly simulated by quantum algorithms designed for time independent Hamiltonians, thus extending their domain of applicability. In particular, we provide the first nontrivial application of qubitization to time dependent Hamiltonians. Though our analysis does not show improvements over alternatives for time dependent simulation, such as Trotter, we expect the fault to lay within the analysis rather than the method itself. Simple numerics may elucidate whether this claim is plausible. Besides direct simulation, the clock framework provides a useful conceptual tool for developing algorithms. As a demonstration, we adapt the multiproduct formalism to the time dependent case, and in turn provide a simulation algorithm that not only has commutator scaling, but also outperforms time dependent Trotter-Suzuki methods. We support our theoretical findings with numerical demonstrations, which indicate the improved performance of time dependent MPFs over low-order product formulas. 
    
    This work opens up a number of interesting possibilities. Most obvious, in our view, is to determine whether the clock formalism can lead to algorithms for time dependent Hamiltonian simulation that match proven lower bounds through the use of qubitization. Another open question is whether these techniques could be used to translate commutator bounds for product formulas~\cite{Childs2021Theory} over to the time dependent case. This would be a significant step towards the development of a complete understanding of the error in Trotter-Suzuki formulas, since for the first time we would have a bound on the error of ordered operator exponentials that yields the anticipated commutator scaling.
    
    Regarding the MPF algorithm specifically, there is a possibility that MPFs converge to the propagator in the limit of large order $m$ regardless of the time step size, assuming sufficient smoothness in $H$. The corresponding statement is not true for Trotter-Suzuki: smaller and smaller time intervals must be taken to ensure convergence as one reaches higher order formulas. Proving (or disproving) absolute convergence would be a valuable avenue for future research. On the numerical side, more convincing demonstrations of time dependent MPFs, using larger systems, would be desirable. 

\section*{Acknowledgments}
    We thank Jeffrey Schenker and Dominic Berry for helpful discussions. This work was supported in part by the U.S. Department of Energy (DOE), Office of Science, Office of Nuclear Physics, Inqubator for Quantum Simulation (IQuS) under Award Number DOE (NP) Award DE-SC0020970, as well as awards DE-SC0021152 and DE-SC0013365, and by the National Science Foundation Graduate Research Fellowship under Grant No. DGE-1848739.
    NW's work on this project is supported by ``Embedding  Quantum Computing  into  Many-body  Frameworks  for  Strongly  Correlated  Molecular  and  Materials  Systems''  project,  which  is funded  by  the DOE Office  of Science,  Office  of  Basic  Energy  Sciences,  the  Division of Chemical Sciences, Geosciences, and Biosciences.
    
\appendix

\section{Clock Space Technical Lemmas} \label{app:clocklemmas}
    Here we provide the proofs of several technical lemmas which are listed in~\sec{Finite_Clock_Space}.

    \begin{proof}[Proof of~\lem{canonical_commutator}]
        We proceed in several steps, first by computing $[U_+, C(H)]$. We have
        \begin{align}
        \begin{aligned}
            [U_+, C(H)] &= \sum_{j=0}^{N_c-1} H_j \otimes [U_+, \vert j\rangle\langle j\vert] \\
            &= \sum_{j=0}^{N_c-1} H_j \otimes (\vert j+1\rangle\langle j\vert - \vert j\rangle \langle j-1 \vert.
        \end{aligned}
        \end{align}
        By splitting the sum and reindexing (all increments modulo $N_c$), we can move the difference to the $H_j$, giving
        \begin{align}
        \begin{aligned}
            [U_+, C(H)] &= \sum_j (H_j - H_{j+1}) \otimes \vert j+1\rangle\langle j\vert \\
            &= - U_+ \sum_j (H_{j+1} - H_j) \otimes \vert j\rangle\langle j\vert.
        \end{aligned}
        \end{align}
        Next, we have that $[U_-, C(H)] = -[U_+, C(H)]^\dagger$. Thus,
        \begin{align}
            [U_+ - U_-, C(H)] = -2 \Re\left(U_+ \sum_j (H_{j+1} - H_j) \otimes \vert j\rangle\langle j\vert\right)
        \end{align}
        and the full result follows almost immediately from the definition of $\Delta$ given in equation~\eqref{eq:Deltadef}.

        As for the upper bound, we note that $\|\Re(A)\| \leq \|A\|$ for any finite-dimensional $A$, and by unitary invariance of the spectral norm we have
        \begin{align}
            \|[\Delta, C(H)]\| \leq \left\|\sum_j \frac{H_{j+1}-H_j}{\delta t} \otimes \vert j\rangle\langle j\vert\right\| = \max_j \left\|\frac{H_{j+1} - H_j}{\delta t}\right\|.
        \end{align}
        The upper bound then follows from the claim
        \begin{align}
            \left\|\frac{H_{j+1} - H_j}{\delta t}\right\| \leq \max_{t \in [t_j, t_{j+1}]} \|\dot{H}(t)\|
        \end{align}
        coming from a the fundamental theorem of calculus and the triangle inequality.
    \end{proof}
    \begin{proof}[Proof of~\lem{normalization}]
        By cyclicity, the normalization $\mathcal{N}$ is the same for all $\vert\phi_j\rangle$, so we consider $j = 0$. Because $\vert\phi_0\rangle$ is normalized in the Euclidean norm, we have
        \begin{align}
        \begin{aligned}
            \mathcal{N}&=\sum_{j=0}^{N_c-1}e^{-2|\delta t^2j|_c^2/\sigma^2} \\
            &=\sum_{j=0}^{N_c/2-1}e^{-2j^2\delta t^2/\sigma^2}+\sum_{j=N_c/2}^{N_c-1}e^{-2(N_c-j)^2\delta t^2/\sigma^2} \\
            &=1+\sum_{j=1}^{N_c/2-1}e^{-2j^2\delta t^2/\sigma^2}+\sum_{j=1}^{N_c/2}e^{-2j^2\delta t^2/\sigma^2} \\
            &=\sum_{j=0}^{\frac{N_c}{2}-1}e^{-2j^2\delta t^2/\sigma^2}+\sum_{j=0}^{\frac{N_c}{2}}e^{-2j^2\delta t^2/\sigma^2}-1
        \end{aligned}
        \end{align}
        We may lower bound the sums as Riemann approximations to a Gaussian integral, giving error functions $\erf$.
        \begin{equation}
        \mathcal{N}\geq \sqrt{\frac{\pi}{8}}\left(\erf\left(\frac{T+2\delta t}{\sqrt{2}\sigma}\right)+\erf\left(\frac{T}{\sqrt{2}\sigma}\right)\right)-1>\sqrt{\frac{\pi}{2}}\frac{\sigma}{\delta t}\erf\left(\frac{T}{\sqrt{2}\sigma}\right)-1\;,
        \end{equation}
        which then implies
        \begin{align}
        \begin{aligned}
            \frac{1}{\mathcal{N}}\leq \sqrt{2/\pi}(\delta t/\sigma)\frac{1}{\erf\left(\frac{T}{\sqrt{2}\sigma}\right)-\sqrt{\frac{2}{\pi}}\frac{\delta t}{\sigma}}&=\sqrt{\frac{2}{\pi}}(\delta t/\sigma)+O\left((\delta t/\sigma)\left(\frac{\delta t}{\sigma}+e^{-\frac{T^2}{2\sigma^2}}\right)\right) \\
            &\in O(\delta t/\sigma).
        \end{aligned}
        \end{align}
        The result follows simply from taking a square root.      
    \end{proof}

\section{Signature Matrix Decomposition} \label{app:signature_matrix_decomposition}
    Here we provide an overview of the signature matrix decomposition, a.k.a. the alternating sign trick, which was used in our clock space qubitization algorithm of~\sec{TDQubitization} for achieving an LCU expression for diagonal (or easily diagonalized) linear operators. While this technique has been a part of the digital quantum simulation toolbox for some time, unfortunately the literature leaves no clear trace of it. Because of this, we hope the reader will find this overview helpful beyond our present application, by filling in a needed record.

    Let $H$ be a Hermitian operator on $D$ dimensions, with diagonal decomposition $H = \sum_{j=1}^D \lambda_j \ketbra{j}$. The question is how we can write this operator as a sum of unitaries, at least to some apporximation. In looking for an appropriate set $\{U_j\}$, it makes sense to restrict our attention to those diagonal in the same basis as $H$. We may naturally restrict $U_j$ to be Hermitian because $H$ is as well. These stringent requirements force $U_j$ to be a so-called signature matrices: diagonal matrices with nonzero entries $\pm 1$.

    To state the idea clearly, we focus on a single eigenvalue $\lambda_j$. We count up to $\lambda_j$ by units of 1 until $\lceil\lambda_j\rceil$ is reached. Then, we alternate between adding units of $-1$ and $+1$. The last step may seem odd, but is necessary because, with unitaries, we can't simply add zero. Nor can we stop the adding procedure before \emph{all} of the eigenvalues have been reached by additions of 1. 

    Let's now proceed more formally. Let $L = \lceil \norm{H} \rceil$ be the first integer larger than the largest eigenvalue of $H$. Define a signature matrix $U_k$ for each $k \in \{1,\dots,L\}$ as follows.
    \begin{align}
        U_k = \sum_{j=1}^D (-1)^{k [k> \lambda_j]} \ketbra{j}
    \end{align}
    Here, $[P]$ is the boolean function for proposition $P$ assigning 1 to true, 0 to false. We see that, for $k$ even, $U_k = \openone$ is the identity operator, while for odd $k$ $U_k$ has eigenvalue $-1$ whenever $j$ is such that $k > \lambda_j$.

    Let $G = \sum_{k=1}^L U_k$. Then $G$ is also diagonal in the $\ket{j}$ basis, and moreover the associated eigenvalue $\eta_j$ is given by
    \begin{align}
    \begin{aligned}
        \eta_j &= \sum_{k=1}^\Lambda (-1)^{k [k >\lambda_j]}
        &= \lfloor \lambda_j \rfloor + O(1)
    \end{aligned}
    \end{align}
    where $O(1)$ in fact denotes an integer from the set $\{-1,0,1\}$. Thus, the error between $\eta_j$ and $\lambda_j$ is upper bounded by 2. 

    This might not seem like a good approximation, especially when $\lambda_j$ is small. But we can artificially increase the size of $\lambda_j$ by performing the same procedure for $H/\delta$ for suitably small $\delta$, then multiplying by $\delta$. Let $L_\delta = \lceil \norm{H}/\delta \rceil$. Then
    \begin{align}
        H/\delta = \sum_{k=1}^{L_\delta} U_k + O(1) 
    \end{align}
    so
    \begin{align}
        H = \sum_{k=1}^{L_\delta} \delta U_k + O(\delta).
    \end{align}
    We've succeeded at expressing $H$ in LCU form to accuracy $O(\delta)$ using $L_\delta$ terms. 

    What about LCU computation? If $H$ is defined on $n$ qubits, we need $H$ to be efficiently diagonalizable by a unitary circuit $W$ into the computational basis. We then need to construct the $\mathtt{PREP}$ and $\mathtt{SEL}$ oracles. The $\mathtt{PREP}$ is simple enough: after normalization we just need a uniform superposition. Meanwhile, the $\mathtt{SEL}$ requires controlled $U_k$ operations. Each $U_k$ can be constructed with the help of a classical comparator circuit to compare each $\lambda_j$ to the integer $k$. The number of auxiliary qubits we will need is $\lceil \log L_\delta \rceil \in O(\log \norm{H}/\delta)$ to get accuracy $\delta$. We will leave the discussion at that: suffice to say that because these constructions exist, our query complexities give an accurate reading on the gate simulation complexity.

\section{Tools from Combinatorics} \label{app:milieu}
    This section is a reference for several tools from combinatorics used, especially in connection to the MPF error analysis of~\sec{ErrorAnalysis}.

    The simple factorial $n!$ counts the number of permutations of $n$ objects, and is usefully approximated by \emph{Stirling's approximation}. In the paper, we always make use of a version of the approximation which gives strict bounds for $n \in \mathbb{Z}_+$. 
    \begin{equation}
        \sqrt{2\pi n}\left(\frac{n}{e}\right)^n < n! < \sqrt{2\pi n}\left(\frac{n}{e}\right)^n e^{1/(12n)}
    \end{equation}
    These bounds are extremely tight, even for small $n$.
    
    The \emph{multinomial coefficient} is a generalization of the more common binomial coefficient, and it arises in several combinatorial situations. It is defined by
    \begin{align} \label{eq:multinomial}
        \binom{n}{n_1, ..., n_k} := \frac{n!}{n_1!n_2!...n_k!}
    \end{align}
    where $n \in \mathbb{Z}_+$ and the $(n_\ell)_{\ell=1}^k$ are nonnegative integers which sum to $n$. It is a positive integer corresponding to the number of distinct ways of placing $n$ distinguishable items into $k$ boxes, where each box has a fixed number $n_\ell$ of items. In this work, we will find occasion to make use of the multinomial when evaluating high-order derivatives of a product.
    \begin{align}
        \left(\frac{d}{dt}\right)^n f_1(t)f_2(t)\dots f_k (t)
    \end{align}
    Here, $(f_\ell)_{\ell=1}^k$ are $n$-differentiable functions of $t \in \mathbb{R}$. Employing the product rule, one is left to count all the possible combinations of derivatives of each $f_\ell$. It turns out that the multinomial is suited for this. 
    \begin{align} \label{eq:multinomial_expansion}
        \left(\frac{d}{dt}\right)^n \prod_{\ell=1}^k f_\ell (t) = \sum_N \binom{n}{n_1,\dots, n_k} \prod_{\ell=1}^k \left(\frac{d}{dt}\right)^{n_\ell} f_\ell (t)
    \end{align}
    The sum is taken over the set $N$ of sequences of nonnegative integers $(n_\ell)_{\ell=1}^k$ summing to $n$. A useful property is that
    \begin{align}
        \sum_N \binom{n}{n_1,\dots,n_k} = k^n
    \end{align}
    for nonnegative integers $k, n$ (with convention $0^0 = \lim_{x\rightarrow0} x^x = 1$).
    
    Besides derivatives of products, we will also need to bound derivatives of ordinary exponentials of a time dependent matrix. Useful for this purpose is an expression for derivatives of exponentials of a scalar function $a(t)$.
    \begin{align}
        \left(\frac{d}{dt}\right)^n e^{a(t)}
    \end{align}
    The solution we rely on is \emph{Faà di Bruno's formula}, which asserts that
    \begin{align} \label{eq:Faa_di_Bruno}
        \left(\frac{d}{dt}\right)^n e^{a(t)} = e^{a(t)} Y_n(a'(t), a''(t),\dots, a^{(n)}(t)) 
    \end{align}
    where $Y_n$ is the \emph{complete exponential Bell polynomial}~\cite{Comtet2012Combinatorics}. An explicit formula is given by
    \begin{align} \label{eq:Bell_def}
        Y_n(x_1, x_2, \dots, x_n) = \sum_C \frac{n!}{c_1!c_2!\dots c_n!} \prod_{j=1}^n \left(\frac{x_j}{j!}\right)^{c_j}
    \end{align}
    where the sum is taken over the set $C$ of all sequences $(c_j)_{j=1}^n$ such that $c_j \geq 0$ and
    \begin{align} \label{eq:sum_rule}
        c_1 + 2 c_2 + \dots + n c_n = n.
    \end{align}
    Essentially, each coefficient in $Y_n$ counts the ways one can partition a set of fixed size $n$ into subsets of given sizes and number. When one simply wants to count the total number of possible partitions, one is led to the \emph{Bell numbers} $b_n$. These are related to the $Y_n$ by evaluating all arguments to $1$.
    \begin{align}
        b_n = Y_n(1,1,...1)
    \end{align}
    More generally, for any $x \in \mathbb{R}$,
    \begin{align}
        Y_n(x, x^2, \dots, x^n) = x^n b_n, 
    \end{align}
    which can be seen directly from \eqref{eq:Bell_def} along with the sum rule \eqref{eq:sum_rule}. The Bell numbers $b_n$ grow combinatorially; in particular, the following upper bound~\cite{Berend2010Improved} is useful.
    \begin{align} \label{eq:bellnumbound}
        b_n < \left(\frac{.792n}{\log(n+1)}\right)^n, \quad \forall n\in \mathbb{Z}_+
    \end{align}
    More generally, the single-variable \emph{Bell polynomial}, or \emph{Touchard polynomial} $B_n(x)$, is simply $Y_n$ with all arguments evaluated to $x$.
    \begin{align}
        B_n(x) = Y_n(x,x,\dots,x).
    \end{align}
    Of course, $b_n = B_n(1)$. The $n$th Bell polynomial $B_n(x)$ is also the value of the $n$th moment of the Poisson distribution with mean $x$. From \cite{Ahle2022Sharp} we have the following upper bound on $B_n$
    \begin{align} \label{eq:Bellfunc_bound}
        B_n(x) \leq \left(\frac{n}{\log(1+\frac{n}{x})}\right)^n,\quad \forall x \geq 0
    \end{align}
    which we observe is very close to that for the Bell numbers ($x=1$) in equation~\eqref{eq:bellnumbound}. From their definitions, $Y_n$, $B_n$ and $b_n$ all grow monotonically, both in their functional arguments and their index $n$. This is intuitive from being combinatorial functions whose coefficients count something according to the size of $n$.

\section{Proof of Operator Faà di Bruno Bound} \label{app:MPFproofs}    
    In this appendix, we prove the Faà di Bruno type bound used in~\sec{ErrorAnalysis}. 

    \begin{proof}[Proof of~\lem{FaadiBrunoOperator}]
        From the Trotter product theorem, we have
    \begin{equation}
        \partial_{t}^n \exp(A(t)) = \partial_{ t}^n\lim_{r\rightarrow \infty}(\exp(A( t)/r))^r.
    \end{equation}
    Using the fact that the series converges uniformly, we may interchange the order of differentiation and the limit.  This leads to
    \begin{equation}
        \|\partial_{t}^n \exp(A(t))\| \le \lim_{r\rightarrow \infty}  \sum_S \binom{n}{s_1,\ldots,s_r}\prod_{q=1}^r\left\|\partial_{t}^{s_q}\exp(A(t)/r)\right\|.
    \end{equation}
    Here the sum over $S$ is constrained such that $s_j \geq 0$ and $s_1 +\cdots+ s_r = n$.  Then using Taylor's theorem we have 
    \begin{equation}
        \left\|\partial_{t}^{s_q}\exp(A(t)/r)\right\| \le \frac{\|A^{(s_q)}(t)\|}{r}  + O(1/r^2).
    \end{equation}
    for $s_q > 0$, where the $O(1/r^2)$ terms will vanish as $r\rightarrow \infty$. The $s_q = 0$ case has upper bound 1 by unitarity. Hence, put together,
    \begin{equation}
         \|\partial_{t}^n \exp(A(t))\| \le \lim_{r\rightarrow \infty}  \sum_S \binom{n}{s_1,\ldots,s_r}\prod_{q=1}^r\left(\frac{\|A^{(s_q)}(t)\|(1-\delta_{s_q,0})}{r} +\delta_{s_q,0}\right). \label{eq:deriv1}
    \end{equation}
    
    Now let us define a scalar function $a(x)$ defined for $x$ in a neighborhood of $t$ such that, for any $k$ such that $0 \leq k \leq n$,
    \begin{equation}
        a^{(k)}(t) = \|A^{(k)}(t)\|(1-\delta_{k,0}).
    \end{equation}
    for a particular $x = t$. Such a function can be seen to exist by considering the $n$th degree Taylor polynomial.
    We may apply the standard Faà di Bruno formula~\eqref{eq:Faa_di_Bruno} to $a$, so that
    \begin{equation}
       \partial_{x}^n e^{a(x)}\bigg\rvert_{x = t} =e^{a(t)} Y_n(\|A^{(1)}(t)\|, \dots, \|A^{(n)}(t)\|) =Y_n(\|A^{(1)}(t)\|, \dots, \|A^{(n)}(t)\|).\label{eq:Ybd}
    \end{equation}
    On the other hand we can split $a(t)$ into $r$ steps and compute the $n$th derivative, just as for the Trotter product theorem.
    \begin{equation}
        \partial_{x}^n e^{a(x)}\bigg\rvert_{x = t} =\lim_{r\rightarrow \infty}  \sum_S \binom{n}{s_1,\ldots,s_r}\prod_{q=1}^r\left(\frac{\|A^{(s_q)}(\Delta t)\|(1-\delta_{s_q,0})}{r} +\delta_{s_q,0}\right) \label{eq:deriv2}
    \end{equation}
    By comparing expressions~\eqref{eq:deriv1} and~\eqref{eq:deriv2}, we see that
    \begin{align}
        \norm{\partial_t^n \exp{A(t)}} \leq \partial_x^n e^{a(x)}\bigg\rvert_{x = t}
    \end{align}
    and applying~\eqref{eq:Ybd}, we reach our desired bound Faà di Bruno bound.
    \begin{equation}
         \|\partial_{t}^n \exp(A(t))\| \le Y_n(\|A^{(1)}(t)\|, \dots, \|A^{(n)}( t)\|)\label{eq:fake_di_bruno}
    \end{equation}
    
    We evaluate the derivatives of $A(t)$, and express them in terms of the derivatives of the Hamiltonian, $H^{(j)}$ (for simplicity, we leave off the evaluation point. The derivative is with respect to the Hamiltonian's single argument). The result is
     \begin{align}
         \partial_{t}^j A(t) = \frac{-i}{k} \left[ \left(\frac{q-1/2}{k}\right)^j (t-t_0) H^{(j)} + j \left(\frac{q-1/2}{k}\right)^{j-1} H^{(j-1)}\right]
     \end{align}
     Employing the $\Lambda_n$-bound from~\defn{LambdaBound}, we have that
     \begin{align}
     \begin{aligned}
         \norm{\partial_{t}^j A(t)} &\leq  \frac{1}{k} \left[ \left(\frac{q-1/2}{k}\right)^j (t-t_0)\Lambda_{n, q}^{j+1}  + j \left(\frac{q-1/2}{k}\right)^{j-1} \Lambda_{n,q}q^{j}\right] \\
         &= \left(\frac{q-1/2}{k}\right)^j \Lambda_{n,q}^j \left[\frac{j}{q-1/2}+\frac{1}{k} (t-t_0)\Lambda_{n,q} \right].
        \end{aligned}
     \end{align}
     Here,
     \begin{align}
         \Lambda_{n,q} := \max_{\tau \in I_q} \Lambda_n(\tau)
     \end{align}
     and $I_q = [t_0+(q-1)(t-t_0)/k, t_0+q(t-t_0)/k]$ is the $q$th interval in the mesh from $t_0$ to $t$ with $k$ even spaces. Since $\Lambda_{n,q} \leq \max_{\tau \in [t_0,t]}\Lambda_n(\tau)$, from the assumptions of the lemma, $\Lambda_{n,q} (t-t_0) <1$. Hence,
     \begin{align}
     \begin{aligned}
         \norm{\partial_{t}^j A(t)} &\leq  \tilde{\Lambda}_{n,q}^j \left[\frac{j}{q-1/2}+\frac{1}{k}\right] \\
    \end{aligned}
    \end{align}
    where $\tilde{\Lambda}_{n,q} 
    \equiv \Lambda_{n,q} (q-1/2)/k$.  
    
    Plugging this into the formula into~\eqref{eq:fake_di_bruno} and using the definition of $Y_n$ given by \eqref{eq:Bell_def}, our bound becomes
    \begin{align}
        \norm{\partial_{t}^n U_2(t)} \leq \sum_C \frac{n!}{c_1!\dots c_n!} \prod_{j=1}^n \left(\frac{(\frac{j}{q-1/2}+\frac{1}{k}) \tilde{\Lambda}_{n,q}^j}{j!}\right)^{c_j}.
    \end{align}
    Using the sum property of the coefficients $c_j$, we can move the $\tilde{\Lambda}_{n,q}^j$ out of the sum. 
    \begin{align}
        \norm{\partial_{t}^n U_2(t)} &\leq \left(\Lambda_{n,q} \frac{q-1/2}{k}\right)^n \sum_C \frac{n!}{c_1!\dots c_n!} \prod_{j=1}^n \left(\frac{\frac{j}{q-1/2}+\frac{1}{k}}{j!}\right)^{c_j} \\
        &= \left(\Lambda_{n,q} \frac{q-1/2}{k}\right)^n Y_n\left(\vec{x}_{q,k}^{(n)}\right).
    \end{align}
    In the last line, we reapplied the definition of $Y_n$ and of the vectors $\vec{x}_{q,k}^{(n)}$. This completes our bound for the $U_2$ formula for the $q$th segment of mesh defined by $k_j$.
    \end{proof}

\section{Greedy Algorithm for Adaptive Time Steps} \label{app:adapt_deltat}
    Here we discuss schemes for constructing the adaptive, nonuniform mesh of time steps used in the MPF algorithm described in~\sec{pseudoMPF}. Specifically, we seek a decomposition of the desired simulation interval $[0,T]$ into a monotonically increasing sequence of times $t_0,t_1,\dots,t_r$, with $t_0 = 0$, $t_r=T$. The mesh construction of~\sec{LongTimeSim}, although theoretically sound, is not directly implementable since it requires knowing the total number of steps while constructing each new point based on local data. To avoid this issue, as well as the restriction $|\Dot{\Lambda} (\tau)| \leq K \Lambda^2(\tau)$  we seek a simple-to-use greedy algorithm. 
    
    One possibility is to use a direct approach which first selects a candidate number of steps $r_{\text{try}}$. Starting from $r_{\text{try}}=1$, we then build recursively a sequence of times using the condition (see Eq.~\eqref{eq:timestepBD} in the main text)
    \begin{equation} \label{eq:app_cond}
            \max_{t\in [t_{i-1},t_i]}\Lambda(t) \left(t_i-t_{i-1}\right)\le \frac{1}{41} \left(\frac{\epsilon}{0.32 \norm{a}_1 r} \right)^{1/(2m+1)}\;,
    \end{equation}
    with $r=r_{\text{try}}$. Starting from $t_0 = 0$ and looking for the largest $t_i$ that satisfies the condition, we finally check whether the generated number of intervals is greater than $r_{\text{try}}$ in which case we increase $r_{\text{try}}$ by one and repeat. When the algorithm stops at the optimal value $r_{\text{opt}}$, we have performed a total of $r_{\text{opt}}(r_{\text{opt}}+1)/2$ non-linear optimization steps, each one requiring multiple evaluations of the left hand side of Eq.~\eqref{eq:app_cond}. This can be very demanding when the left hand side of Eq.~\eqref{eq:app_cond} is expensive to evaluate and the optimal number of intervals is around half the upperbound
    \begin{equation}
        r_{\text{max}}= \left(41(t-t_0) \max_{\tau\in[t_0,t]}\Lambda_{2m+1}(\tau)\right)^{\frac{2m+1}{2m}}\left(\frac{0.32\norm{a}_1}{\epsilon}\right)^{\frac{1}{2m}}
    \end{equation}
    obtained considering identical intervals and bounding $\Lambda(t)$ with its maximum value over the whole simulation interval $[0,T]$. In this case, finding an approximation to the optimal decomposition requires $O(r_{\text{max}}^2)$ optimization steps, each one requiring multiple evaluations of the lefty hand side of Eq.~\eqref{eq:app_cond}.
    
    We now describe an alternative approach which determines $r_{\text{opt}}$ within a factor of 2 and uses only $r_{\text{max}}$ evaluations of $\max_{t \in [t_{i-1},t_i]}\Lambda(t)$ and additional $O(\log(r_{\text{max}})r_{\text{max}})$ simple arithmetic operations. This procedure can be used to find a viable, and approximately optimal, decomposition of the time interval or as a good starting point to find the optimal one using a procedure as the one described above. The idea is to start by decomposing the interval $[0,T]$ into $r_{\text{max}}$ segments with equal length and storing the maximum of $\Lambda(t)$ in each segment in an array $A$ of size $r_{\text{max}}$. We then introduce an additional array of the same size
    \begin{equation}
        L_m = \left[ \max_{k\leq m} A_k \right] m \frac{T}{r_{\text{max}}}\;,
    \end{equation}
    together with an additional set of vectors of the same size
    \begin{equation}
        R^{(n)}_{m} = \left[ \max_{n\geq k > m} A_k \right] (n-m) \frac{T}{r_{\text{max}}}\;,
    \end{equation}
    with $n$ an additional index between $1$ and $r_{max}$. The first vector stores the left hand side of Eq.~\eqref{eq:app_cond} for the interval up to the $m$-th time while the second vector stores the same information for the interval starting at the $m$-th time and ending at the $n$-th one. The algorithm proceeds by splitting the time interval recursively into two parts so that the left hand side of Eq.~\eqref{eq:app_cond} takes (approximately) the same value on both halves (ie. we are splitting the error equally on both sides). At every iteration the number of
    intervals doubles and the right hand side of Eq.~\eqref{eq:app_cond} shrinks accordingly. We stop the procedure once Eq.~\eqref{eq:app_cond} is satisfied on one interval (since we are guaranteed it will in all others). The procedure will stop at some $r_K$ at which point we know the optimal value $r_{\text{opt}}$ is in $[\lceil r_K/2\rceil,r_K]$. The algorithm can then be described as follows
    \begin{enumerate}
        \item Compute $L_m$ for all $m=1,...,r_{\text{max}}$
        \item Set $n=r_{\text{max}}$ and $r=2$
        \item Compute the elements of $R^{(n)}_m$ for all $m=1,...,n-1$
        \item Initialize an auxiliary array $D_m$ as $D_m = L_m - R^{(n)}_m$
        \item Find the least index $k$ for which $D_k>0$
        \item If $L_k$ is less than the right hand side of Eq.~\eqref{eq:app_cond} with the current value of $r$, set $r_K=r$ and exit
        \item If $2r \geq r_{\text{max}}$ set $r_K=r_{\text{max}}$ and exit
        \item set $r=2r$, $n=k$ and repeat from step 3
    \end{enumerate}
    
    Step 1 requires $r_{\text{max}}$ operations while Steps 3 and 4 cost $n$ operations each. Since the number of iterations is bounded by $\log_2(r_{\text{max}})$, their combined cost is bounded by $2\log_2(r_{\text{max}})r_{\text{max}}$. If we use binary search, Step 5 costs $\log_2(n)$ operations so its total cost is at most $\log_2(r_{\text{max}})^2$ operations. From this analysis we see that Steps 3 and 4 are the most expensive ones and they dominate the cost of the scheme. On exit we have $r_K\approx r_{\text{opt}}$ together with the first interval $[t_0,t_1]$. The rest of the intervals can then be found keeping $r=r_K$ fixed with additional $O(r_{\text{max}})$ operations.

\bibliographystyle{unsrt}
\bibliography{ref}


\end{document}